%% file: main.tex
\documentclass{article}

\pdfoutput=1

\input{packages}

\input{preamble}
\title{Multi-Absorbing Phase-Type Distributions for \subchg{Right-Censored} Competing Risks \subchg{Data}}

\usepackage{authblk}

\newbox{\orcid}\sbox{\orcid}{\includegraphics[scale=0.06]{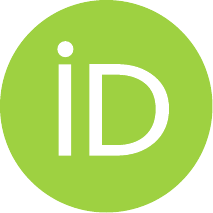}}

\author[1]{%
	\href{https://orcid.org/0000-0000-0000-0000}{\usebox{\orcid}\hspace{1mm}Zhihao Qiao\thanks{\texttt{zhihao.qiao@uq.edu.au}}}%
}
\author[2]{%
	\href{https://orcid.org/0000-0000-0000-0000}{\usebox{\orcid}\hspace{1mm}Budhi Surya\thanks{\texttt{dr.ir.b.a.surya@gmail.com}}}%
}
\author[3]{%
	\href{https://orcid.org/0000-0000-0000-0000}{\usebox{\orcid}\hspace{1mm}Azam Asanjarani\thanks{\texttt{azam.asanjarani@auckland.ac.nz}}}%
}
\author[1]{%
	\href{https://orcid.org/0000-0000-0000-0000}{\usebox{\orcid}\hspace{1mm}Yoni Nazarathy\thanks{\texttt{y.nazarathy@uq.edu.au}}}%
}
\affil[1]{School of Mathematics and Physics, University of Queensland, St Lucia, Queensland}
\affil[2]{School of Mathematics and Statistics, Victoria University of Wellington, New Zealand}
\affil[3]{Faculty of Science, Statistics, University of Auckland, New Zealand}

\renewcommand{\shorttitle}{Multi-Absorbing Phase-Type Distributions for Right-Censored Competing Risks Data}

\hypersetup{
hidelinks,   
pdftitle={Multi-Absorbing Phase-Type Distributions for Right-Censored Competing Risks Data},
pdfsubject={statistics, phase-type distributions, right-censored competing risks},
pdfauthor={Zhihao Qiao, Budhi Surya, Azam Asanjarani, Yoni Nazarathy},
pdfkeywords={Phase-type distributions, competing risks, right censoring, EM algorithm, MAPH},
}

\begin{document}
\maketitle

\input{sections/abstract}
\input{sections/introduction}
\input{sections/maph-distributions}
\input{sections/estimation}
\input{sections/censoring}
\input{sections/implementation}
\input{sections/conclusion}

\input{sections/appendix}

\bibliographystyle{unsrtnat}
\bibliography{references}
\end{document}

%% file: packages.tex
\usepackage{arxiv}
\usepackage[toc,page]{appendix}
\usepackage[utf8]{inputenc} 
\usepackage[T1]{fontenc}    
\usepackage{hyperref}       
\usepackage{url}            
\usepackage{booktabs}       
\usepackage{amsfonts}       
\usepackage{nicefrac}       
\usepackage{microtype}      
\usepackage{graphicx}
\usepackage[numbers]{natbib}   
\usepackage{doi}
\usepackage{amsmath}
\usepackage{bbm}
\usepackage{amsmath,amsfonts,amssymb}
\usepackage{euscript}
\usepackage{tikz}
\usetikzlibrary{automata,arrows,positioning,calc}
\usepackage{pgfplots}
\pgfplotsset{compat=1.16}
\usepackage{algorithm2e}
\usepackage{MnSymbol} 
\usepackage{amsthm}

\usepackage{cleveref}
\crefname{equation}{}{Eqs.}
\crefname{figure}{Fig.}{Figs.}

%% file: preamble.tex
\newtheorem{lemma}{Lemma}
\newtheorem{corollary}{Corollary}
\newtheorem{example}{Example}
\newtheorem{proposition}{Proposition}

\newcommand{\R}{\mathbb{R}}
\newcommand{\Prob}{\mathbb{P}}
\newcommand{\EX}{\mathbb{E}}

\newcommand{\X}{{\mathcal{X}}}
\newcommand{\ind}[1]{\mathbbm{1}\{#1\}}
\newcommand{\ex}[1]{\mathbb{E}\left[#1\right]}
\newcommand{\var}[1]{\mathrm{Var}\left[#1\right]}

\newcommand{\cs}{\mathcal{S}}

\newenvironment{rcrevision}{}{}
\newenvironment{aarevision}{}{}
\newenvironment{subrevision}{}{}
\newcommand{\aachg}[1]{#1}
\newcommand{\subchg}[1]{#1}

%% file: sections/abstract.tex
\begin{abstract}
\begin{rcrevision}
Phase-type (PH) distributions are versatile semi-parametric models for lifetime duration and can be used in survival and reliability analysis. In this paper we put forward methods and software for using PH distributions in a competing-risks model. The resulting multi-absorbing phase-type ($\text{MAPH}_{m,n}$) distribution records both the time until absorption and its cause. After formulating basic properties of this family, we develop \subchg{an} EM algorithm for parameter inference from exact event-time/cause observations and independently right-censored observations. For a censored subject, the E-step conditions on survival \aachg{up to} the censoring time and imputes both the latent transient path and its eventual \aachg{cause of absorption}; closed-form cause-resolved sufficient-statistic expectations are obtained from matrix exponentials. We illustrate applications and numerical properties and describe two accompanying Julia packages, in which both the exact-event and the censored E-step are implemented. Applied to intensive-care length-of-stay data in full, censored records included, the method attains a higher likelihood than the phase-type competing-risks fit previously published for those data.
\end{rcrevision}
\end{abstract}

\begin{rcrevision}
\keywords{Phase-type distributions \and competing risks \and right censoring \and EM algorithm \and survival analysis}
\end{rcrevision}

%% file: sections/introduction.tex
\section{Introduction}

The competing risks framework is a cornerstone of survival analysis \citep{kalbfleisch_prentice_2002, crowder_2001, beyersmann_2012}. In this framework one models the time until an event together with the cause of that event, drawn from one of several mutually exclusive types. Such a setting arises naturally in biostatistics, where a patient may leave a given health state through one of several competing causes \citep{putter_2007, austin_2016}, in reliability, where a component may fail through one of several distinct failure modes \citep{crowder_2001, meeker_escobar_1998}, and in many other domains. A range of statistical methods and accompanying software exist for competing risks, including the cause-specific and subdistribution hazards approaches \citep{fine_gray_1999} together with their implementations in widely used \textsf{R} packages \citep{dewreede_2011, therneau_survival}. In this paper we instead study the use of phase-type distributions for competing risks, an avenue that to date has not been developed extensively within an EM algorithm framework.

Phase type distributions arose in applied probability and are a general mechanism for modelling non-negative random variables via an absorption time of a continuous-time Markov chain; see for example \citet{latouche_introduction_1999} or \citet{he_fundamentals_2014} for an overview. In applied probability, phase-type distributions are often attractive since they allow one to incorporate non-exponential random variables within stochastic models, while retaining the Markovian nature of the model. Further, in statistics and actuarial science, phase-type distributions have also gained popularity since they present a semi-parametric model which in essence can approximate any non-negative distribution. When one considers phase-type distributions, one approach is to use the distributions as a first-principles stochastic modelling tool, while a second approach is to carry out parameter estimation for such distributions. A seminal paper for the second approach is \citet{Asmussen1996FittingPD}; see also \citet{Olsson_1996_censored} and software packages \citep{matrixdist} and more. Our approach in this paper follows such an inference approach where we develop an EM algorithm for phase-type distribution generalizations for competing risks. The idea of merging phase-type distribution principles with competing-risks models was recently advocated by \citet{lindqvist_phase-type_2022} (see also \citep{garcia-maya_competing_2022}) where instead of using a scalar distribution of a hitting time in a Markov chain, one considers a bi-variate distribution with the first component indicating the time until absorption, and the second discrete finite component indicating the exact absorbing state to which absorption occurred. Related earlier ideas, in which event times are modelled through a Markov chain with several absorbing states, appear in the healthcare modelling literature \citep{marshall_mcclean_2004}.

\begin{rcrevision}
In general when considering phase-type random variables or the multi-absorbing (competing-risks) generalizations used here, there is an underlying continuous-time Markov chain with unobserved transitions. For an event record we observe the absorption time and cause but not the transient trajectory; for an independently right-censored record we observe only that the chain is still transient at the censoring time. This setup naturally lends itself to \aachg{the} expectation maximization (EM) \aachg{algorithm}. We take the complete \aachg{data} to be the entire Markov path through its eventual absorption, including the eventual cause even for \aachg{right-censored subjects}. Thus censoring \aachg{affects only} the conditioning event in the E-step \aachg{(expectations are taken conditional on survival up to the censoring time),} but leaves the complete-data likelihood and \aachg{the} M-step unchanged. The initial application of EM algorithms for phase-type distributions was in \citet{Asmussen1996FittingPD}, and censored PH observations were treated by \citet{Olsson_1996_censored}; subsequent variants and refinements include \citep{thummler_2006, okamura_2011, bladt_nielsen_2017}. Importantly, phase-type distributions and the generalizations considered here are not identifiable, so inference concerns the \aachg{underlying} distributional law rather than a unique parameter representation.
\end{rcrevision}

\begin{rcrevision}
Our main contribution is an end-to-end \subchg{exact} EM algorithm for \aachg{fitting MAPH distributions to competing-risk data in the presence of independent right censoring.} Let \aachg{$c^\ell$} be the censoring time, \aachg{the observed time} \aachg{$y^\ell=\min(\tau^\ell,c^\ell)$} and \aachg{the event indicator} \aachg{$\delta^\ell=\mathbf 1\{\tau^\ell\le c^\ell\}$}. The observed data are ${\cal D}=\{(y^\ell,\delta^\ell,k^\ell):\ell=1,\ldots,L\}$, where the cause $k^\ell\in\{1,\ldots,n\}$ is recorded only when $\delta^\ell=1$. We assume \aachg{independent} non-informative \aachg{right-}censoring\aachg{, so that,} conditional on the censoring times, the censoring mechanism carries no parameters of the MAPH law. \subchg{Each} modelled cause is assumed to \subchg{be observed} at least once.
\end{rcrevision}

\aachg{We consider a} generalization of phase-type distributions introduced in \citet{lindqvist_phase-type_2022}, where we denote the random pair associated with the absorption time and absorption cause as $\zeta = (\tau, \kappa)$. \aachg{While \citet{lindqvist_phase-type_2022} established the fundamental probabilistic properties of this model, our focus is on statistical inference. Specifically, we develop a practical end-to-end implementation of the expectation-maximization (EM) algorithm for fitting the model, including computational procedures for evaluating the required sufficient statistics. To facilitate estimation, we adopt an alternative parameterization that is more suitable for inference and introduce a moment-based initialization heuristic that provides suitable starting values and improves the practical performance of the EM algorithm.} Our algorithms are implemented in two companion Julia language packages, \texttt{PhaseTypeDistributions.jl} for representing the distributions and \texttt{PhaseTypeDistributionsFitting.jl} for the EM fitting procedure. We present numerical examples using our software.

The probabilistic model that we use was initially introduced as ``a phase-type model for competing risks'' in \citet{lindqvist_phase-type_2022}, where the author establishes a number of its basic properties. In the present paper we refer to this distribution as a {\em Multi-Absorbing Phase-Type} (MAPH) distribution, since it may also find applications beyond the statistical competing risks framework. A by-product of our analysis is a collection of distributional properties for MAPH distributions, including formulas for the moments and the density. A MAPH distribution is indexed by two positive integers $m$ and $n$: the number of unobserved transient states, or phases, is $m$, and the number of absorbing states, namely the competing causes, is $n$. When $n=1$ the MAPH family reduces to the standard phase-type distributions. As is common in inference for phase-type distributions, the order $m$ is typically not estimated directly but selected by external model-selection criteria such as AIC or BIC \citep{bladt_nielsen_2017}. Finally, we note the considerable recent progress on general multivariate phase-type and matrix distributions \citep{albrecher_bladt_yslas_2022, bladt_2023, albrecher_bladt_bladt_2020}; we emphasize that our MAPH distributions are of a different nature, being a bivariate time-and-cause law rather than a multivariate vector of durations.

\begin{rcrevision}
At a high level, our algorithm works with two parameterizations of the same MAPH law: the natural generator parameterization $(\alpha,T,D)$, and a second, inference-oriented parameterization \subchg{$(\alpha,q,R,P^\lambda)$ in which $q$ collects the phase exit rates, $R$ the eventual absorption probabilities, and $P^\lambda$ the one-step transition matrix of the embedded jump chain}. The full Markov trajectory is latent, and we fit the model \aachg{using the EM algorithm}. For an exact event, the E-step conditions on its observed time and cause. For a record censored at $c$, it conditions on $\{\tau>c\}$ and imputes the transient path before and after $c$, its eventual cause, and its eventual absorbing jump. Both sets of conditional expectations have closed forms in matrix exponentials. The common M-step then maximizes a relaxed surrogate of the complete-data likelihood, \subchg{giving the exact maximizer in closed form: simple ratios of the expected statistics, which lie in the parameter set automatically, so no feasibility step is needed and the observed-data likelihood is non-decreasing along the iteration.} A censoring-compatible simplified heuristic supplies initial parameters; the more elaborate moment-based initialization requires uncensored or externally censoring-adjusted moment targets.
\end{rcrevision}

\begin{rcrevision}
The remainder of the paper is structured as follows. In Section~\ref{sec:maph-dist} we define the MAPH family and its basic distributional properties, and then introduce the embedded jump chain together with an inference-oriented parameterization, establishing its correspondence with the generator parameterization. In Section~\ref{sec:full-path-mle} we derive the maximum-likelihood estimators under full path observation and then develop from them the EM algorithm for exactly observed competing-risks data, including the conditional expectations of the E-step. Section~\ref{sec:censoring} gives the adaptations required by right-censored records, together with the resulting properties of the iteration. Section~\ref{sec:impl} presents numerical examples for the existing uncensored implementation and delineates what remains to validate the censoring extension computationally. \aachg{Finally, Section~\ref{sec:conclusion} concludes with a discussion and directions for future research.} Supporting appendices give the initialization heuristic (\Cref{sec:init-heuristic}) and then the proofs, one appendix per section of the body: \Cref{app:proofs-maph} for the distributional results, \Cref{app:proofs-est} for the estimators and the exact-event E-step, and \Cref{app:proofs-cens} for the right-censoring results.
\end{rcrevision}

%% file: sections/maph-distributions.tex
\section{MAPH distributions}
\label{sec:maph-dist}

\subsection{Definition and basic properties}
Consider a continuous-time Markov chain $X=\{X(t):t\geq 0\}$, also called a Markov jump process (MJP),  that lives on a finite-state space $\overline{\mathcal{S}}=\mathcal{S}\cup \mathcal{S}^0$, with $\cs = \{s_1,\ldots,s_m\}$ and  $\mathcal{S}^0=\{s^0_1,\ldots,s^0_n\}$ disjoint sets. The states in $\cs$ are transient states whilst the states in $\cs^{0}$ are absorbing states. We assume that with probability one, $X(0) \in \cs$. For background about MJPs, see for example \citet{norris97}.

The probability law of this MJP is specified via an initial probability $\alpha_i$ of selecting initial state $s_i \in \mathcal{S}$ such that $\Prob(X(0) = s_i) = \alpha_i$ and the elements $\alpha_i$ are organized in the row vector $\alpha$. The transition rates
are organized in a generator matrix $Q$ where we order the states first via $\cs$ and then via $\cs^0$. We denote the entries of this matrix as $q_{ij}$, where without ambiguity we let the indices $i$ and $j$ span $1,\ldots,m,m+1,\ldots,m+n$. The first $m$ diagonal entries of the matrix are denoted via $-q_i$, for $i=1,\ldots,m$ and are strictly negative such that the row sums are $0$.

We parameterize this generator matrix via $T \in \R^{m \times m}$ and $D \in \R^{m \times n}$ where $T$ is a sub-generator (row sums are at most $0$) and $D$ has non-negative entries. We require $T\mathbf{1}_m + D\mathbf{1}_n=0$ where $\mathbf{1}_m$ is a column vector of $1$'s of length $m$ and similarly for $\mathbf{1}_n$. Each state $s^0_k$ is absorbing and forms its own class. We assume that absorption is certain: from every transient state, the absorbing set $\cs^0$ is reached in finite time with probability one. As discussed below, these conditions are weaker than requiring $T$ to be irreducible and admit the reducible structure of standard phase-type distributions. With these matrices we have,
\begin{align*}
Q=
\begin{bmatrix}
T & D\\
0_{n \times m} & 0_{n \times n}
\end{bmatrix}
\quad \textrm{with} \quad T\mathbf{1}_m + D\mathbf{1}_n=0.
\end{align*}

Hence in summary, $\alpha$, $T$, and $D$ define the probability law of $X(\cdot)$. We denote the off-diagonal elements of $T$ via $\lambda_{ij}$ for index $(i,j)$ and the diagonal elements via $-q_i$ for index $(i,i)$. Further, we denote every element of $D$ as $\mu_{ik}$.

We stress that we do not require $T$ to be irreducible. That absorption is certain is exactly the standard phase-type condition; equivalently, $-T$ is a non-singular $M$-matrix (see the discussion preceding \Cref{prop:maph-props}). In particular, common phase-type structures such as Erlang, hyper-exponential, and Coxian distributions -- whose $T$ is reducible -- are included.

A trajectory $X=\{X(t):t\geq 0\}$ induces a random variable $\tau > 0$ with, $\tau = \inf\{ t > 0 ~:~ X(t) \in \cs^0\}$. This is the hitting time of the set $\cs^0$ and in the case where $n=1$, $\tau$ is called a phase-type random variable. See for example \citet{he_fundamentals_2014}. Further, the trajectory induces a discrete random variable $\kappa$ such that $s^0_\kappa \in \cs^0$ where $s^0_\kappa = X(\tau)$. Hence $\kappa$ is the index of the state of absorption. In the context of the competing-risks model, the pair $(\tau, \kappa)$ describes the lifetime (or time until absorption) $\tau$ together with the cause for absorption, $s^0_\kappa$. An MAPH distribution, is the distribution of the pair $(\tau, \kappa)$ and is parameterized via $\alpha$, $T$, and $D$ or via alternative parameterizations as presented below.

Ordinary phase-type distributions are the special case $n=1$: with a single absorbing state, $\kappa\equiv1$ is degenerate, $D$ reduces to the exit-rate (column) vector $-T\mathbf{1}_m$, and the law of $\tau$ alone is a phase-type distribution $\mathrm{PH}(\alpha,T)$. Thus phase-type distributions are precisely the $\text{MAPH}_{m,1}$ family, and MAPH extends phase-type to the competing-risks setting with $n\ge1$ causes. We can now denote the $\text{MAPH}_{m,n}$ distribution and its parameters via
\begin{equation}
\label{eq:first-param-maph}
\text{MAPH}_{m,n}(\alpha,T,D).
\end{equation}
With this parameterization, the free parameters of the distribution are specified by the non-negative quantities,
\begin{equation}
\label{eq:theta-first-param}
    \theta= \big(\alpha_i, \lambda_{ij}, D_{ik}
    ~~~\text{for}~~~
    i=1,\ldots,m, ~j = 1,\ldots,m,~j\neq i,~k = 1,...,n\big),
\end{equation}
where $\lambda_{ij}$ are the off-diagonal entries of $T$ (the diagonal entries $-q_i$ are then determined by the constraint below). These are subject to the standard constraints that $\alpha$ has non-negative entries with a sum of unity. The matrix $T$ is a sub-generator from which absorption is certain, and finally, the matrix $D$ satisfies
\begin{equation}
\label{eq:D-matrix-constraint}
D\mathbf{1}_n= - T\mathbf{1}_m.
\end{equation}

Figure~\ref{fig:example-maph} presents an illustration of the underlying states of the MJP $X(\cdot)$ behind an MAPH distribution with $m=4$ transient states and $n=3$ absorbing states. The transient states are connected through a sparse set of transitions -- a nearest-neighbour backbone $s_i \leftrightarrow s_{i+1}$ together with a single feedback transition $s_4 \to s_1$ -- while absorption takes place out of selected transient states; consequently both $T$ and $D$ are sparse, as shown to the right of the diagram. The initial state is drawn according to $\alpha$, indicated by the dashed arrows. For a fully connected $\text{MAPH}_{m,n}$, the number of free parameters is $(m-1) + m(m-1) + m\,n$: the initial vector $\alpha$ contributes $m-1$ (it is a probability vector), the sub-generator $T$ contributes its $m(m-1)$ off-diagonal rates (its diagonal $-q_i$ is then fixed by the constraint $T\mathbf{1}_m + D\mathbf{1}_n=0$), and $D$ contributes $m\,n$ absorption rates. For $m=4$ and $n=3$ this is $3 + 12 + 12 = 27$, whereas the sparse example of the figure uses only the displayed non-zero entries. We note that the distribution of $(\tau, \kappa)$ does not uniquely determine such a parameter set. Such non-identifiability issues are well known with phase-type distributions and transcend to our MAPH distributions \citep{bladt_sorensen}.

\begin{figure}[h]
\caption{Transition diagram of an MAPH distribution with $m=4$ transient states (white, centre) and $n=3$ absorbing states (red, right). The initial state is selected according to $\alpha$, indicated by the solid start node at the top left and the dashed arrows. The transient dynamics are sparse -- a nearest-neighbour backbone together with a feedback transition $\lambda_{41}$ -- and absorption occurs out of states $s_1,\ldots,s_4$. The corresponding $\alpha$, $T$ and $D$ are shown to the right. In the competing-risks reading used as a running example (\Cref{ex:running}), absorption into $s^0_1$ represents a completely healthy discharge, while $s^0_2$ and $s^0_3$ represent a discharge requiring further care of type~2 and type~3, respectively.
\label{fig:example-maph}}
\centering
\begin{minipage}[c]{0.60\textwidth}
\centering
\begin{tikzpicture}[->, >=stealth', auto, semithick]
\tikzstyle{every state}=[fill=white,draw=black,thick,text=black,scale=0.9]
\node[circle, fill=black, inner sep=2.5pt, label={[font=\scriptsize]below left:Start}] (start) at (-2.8,6) {};
\node[state]              (s1) at (0,6)   {$s_1$};
\node[state]              (s2) at (0,4)   {$s_2$};
\node[state]              (s3) at (0,2)   {$s_3$};
\node[state]              (s4) at (0,0)   {$s_4$};
\node[state,fill=red]     (a1) at (6,5) {$s^0_1$};
\node[state,fill=red]     (a2) at (6,3) {$s^0_2$};
\node[state,fill=red]     (a3) at (6,1) {$s^0_3$};
\path[dashed, every node/.style={font=\scriptsize}]
(start) edge node[above]{$\alpha_1$} (s1)
(start) edge node[pos=0.65,above]{$\alpha_2$} (s2)
(start) edge node[pos=0.65,above]{$\alpha_3$} (s3)
(start) edge node[pos=0.7,left]{$\alpha_4$} (s4);
\path[every node/.style={font=\scriptsize}]
(s1) edge[bend left=12] node[right]{$\lambda_{12}$} (s2)
(s2) edge[bend left=12] node[left]{$\lambda_{21}$} (s1)
(s2) edge[bend left=12] node[right]{$\lambda_{23}$} (s3)
(s3) edge[bend left=12] node[left]{$\lambda_{32}$} (s2)
(s3) edge[bend left=12] node[right]{$\lambda_{34}$} (s4)
(s4) edge[bend left=12] node[left]{$\lambda_{43}$} (s3)
(s4) edge[bend left=60] node[left]{$\lambda_{41}$} (s1)
(s1) edge node[above]{$\mu_{11}$} (a1)
(s2) edge node[pos=0.6,above]{$\mu_{22}$} (a2)
(s3) edge node[pos=0.6,above]{$\mu_{32}$} (a2)
(s3) edge node[pos=0.6,below]{$\mu_{33}$} (a3)
(s4) edge node[below]{$\mu_{43}$} (a3)
;
\end{tikzpicture}
\end{minipage}%
\begin{minipage}[c]{0.38\textwidth}
\centering
\begin{align*}
\alpha &= \begin{bmatrix} \alpha_1 & \alpha_2 & \alpha_3 & \alpha_4 \end{bmatrix},\\[1.5ex]
T &=
\begin{bmatrix}
-q_1 & \lambda_{12} & 0 & 0\\
\lambda_{21} & -q_2 & \lambda_{23} & 0\\
0 & \lambda_{32} & -q_3 & \lambda_{34}\\
\lambda_{41} & 0 & \lambda_{43} & -q_4
\end{bmatrix},\\[1.5ex]
D &=
\begin{bmatrix}
\mu_{11} & 0 & 0\\
0 & \mu_{22} & 0\\
0 & \mu_{32} & \mu_{33}\\
0 & 0 & \mu_{43}
\end{bmatrix}.
\end{align*}
\end{minipage}
\end{figure}
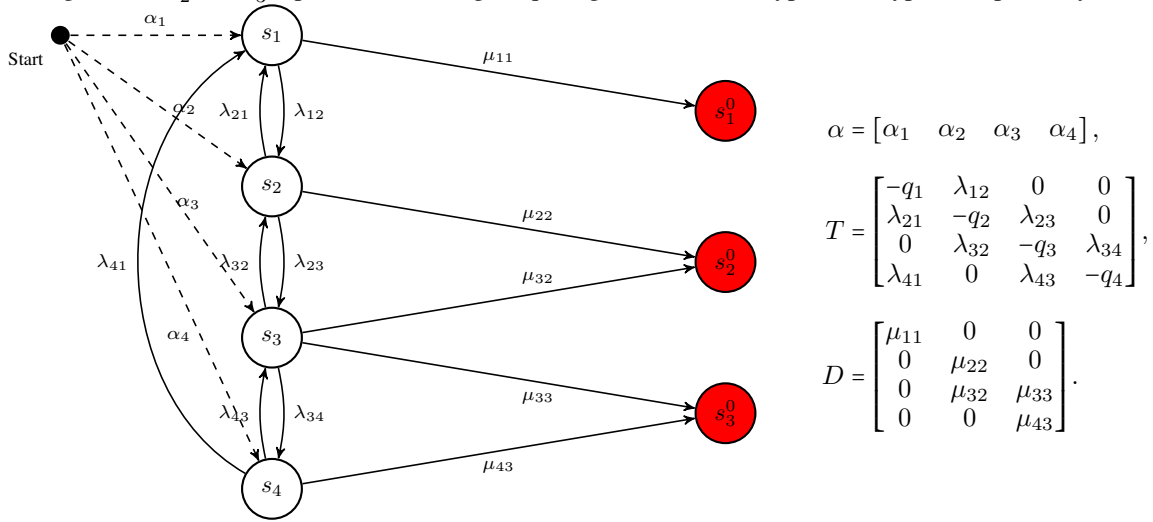

We now summarize a variety of properties of $\text{MAPH}_{m,n}$ distributions for the random vector $\zeta = (
\tau,\kappa)$. We use the notation $D_k$ to indicate the $k$'th column of the matrix D.

Before stating these properties we record a structural fact used throughout. Since absorption is certain, $\cs^0$ is accessible from $\cs$ and $D \neq 0$, so the row sums of $T$ satisfy $T\mathbf{1}_m = -D\mathbf{1}_n \le 0$ with strict inequality in at least one row. A sub-generator $T$ from which absorption is certain -- equivalently, every transient state can reach a state with a strictly negative row sum -- has $-T$ a non-singular $M$-matrix; equivalently, every eigenvalue of $T$ has strictly negative real part \citep[Ch.~6]{berman_plemmons_1994}. (Were every row sum of $T$ zero, $T$ would generate a conservative, recurrent chain with $T\mathbf{1}_m=0$, hence singular; certain absorption rules this out.) This holds whether or not $T$ is irreducible, and in particular for the reducible $T$ of Erlang, hyper-exponential, and Coxian distributions. Consequently $T$ is invertible with $-T^{-1} \ge 0$ entrywise and $e^{Tu} \to 0$ as $u \to \infty$, so the matrix exponentials and integrals appearing below are well defined.

\begin{proposition}
\label{prop:maph-props}
Let $\zeta = (\tau,\kappa)$ be an $\text{MAPH}_{m,n}(\alpha,T,D)$ random pair and let $D_k$ denote the $k$'th column of $D$. Since $-T$ is a non-singular $M$-matrix (as established above), $-T^{-1}$ has non-negative entries, and for $u \ge 0$ and $k=1,\ldots,n$ the following hold.
\begin{itemize}
    \item[(i)] The sub-distribution function is $F(u,k) = \Prob(\tau \le u, \kappa = k)= -\alpha(I-e^{Tu})T^{-1}D_k,$ and the all-cause distribution function is $F(u) = \Prob(\tau \le u) = \sum_{k=1}^n F(u,k) = 1 - \alpha e^{Tu}\mathbf{1}_m$.
    \item[(ii)] The sub-density over non-negative $u$ is $f(u,k) = \frac{\partial}{\partial u} F(u,k) =  \alpha e^{T u}D_k.$
    \item[(iii)] The marginal absorption probability is $\Prob(\kappa = k) = F(\infty,k) = -\alpha T^{-1}D_k$, and the conditional density of $\tau$ given $\kappa = k$ is $f(u \mid \kappa = k) = f(u,k)/\Prob(\kappa = k)$.
    \item[(iv)] The cause-specific hazard rate for cause $k$ is
    $H(u,k) = \dfrac{f(u,k)}{1-F(u)} = \dfrac{\alpha e^{Tu}D_k}{\alpha e^{Tu}\mathbf{1}_m}$.
    It satisfies $H(u,k)\,\mathrm{d}u = \Prob\big(\tau \in [u,u+\mathrm{d}u),~ \kappa = k ~\big|~ \tau > u\big)$, namely the instantaneous rate of absorption through cause $k$ given survival past $u$.
    \item[(v)] The (partial) Laplace--Stieltjes transform is $\phi(s,k) = \mathbb{E}\left[e^{-s\tau}\ind{\kappa=k}\right]  =\alpha(sI-T)^{-1}D_k$, for $s \ge 0$.
    \item[(vi)] The (partial) moment generating function is  $M(s,k)= \mathbb{E}\left[e^{s\tau}\ind{\kappa=k}\right] = -\alpha(sI+T)^{-1}D_k$, for $s$ in a neighbourhood of $0$.
    \item[(vii)] The $j$'th partial moment is $\mathcal{M}_{j,k} = \mathbb{E} \left[\tau^j \ind{\kappa=k} \right]= \int_{0}^{\infty} u^j f(u,k)\,\mathrm{d}u = (-1)^{j+1}\, j!\, \alpha\, T^{-(j+1)}D_k$, for $j = 1,2,\ldots$.
\end{itemize}
\end{proposition}
\begin{proof}
The proof uses standard matrix-analytic computations and is given in \Cref{app:proof-maph-props}.
\end{proof}
The quantities in items $(v)$--$(vii)$ are termed \emph{partial} (or sub-) quantities because they carry the indicator $\ind{\kappa=k}$ rather than conditioning on the event $\{\kappa=k\}$. They are linked to the corresponding conditional quantities through the marginal absorption probability of item $(iii)$: for any integrable function $g$,
\[
\ex{g(\tau)\,\ind{\kappa=k}} = \ex{g(\tau)\mid \kappa=k}\,\Prob(\kappa=k),
\]
since the contribution of the event $\{\kappa\neq k\}$ vanishes. In particular the partial moment satisfies $\mathcal{M}_{j,k} = \ex{\tau^j \mid \kappa=k}\,\Prob(\kappa=k)$, so \subchg{whenever $\Prob(\kappa=k)>0$, dividing by it} recovers the genuine conditional moment $\ex{\tau^j \mid \kappa=k}$; the same normalization applies to the partial transform of item $(v)$ and the partial moment generating function of item $(vi)$.

Further, the cause-specific hazard rate function  is specified in $(iv)$.

To make these quantities concrete we introduce a running example, used again in \Cref{sec:verification} to verify the fitting procedure.

\begin{example}[A running $\text{MAPH}_{4,3}$ model]
\label{ex:running}
Consider the sparse $\text{MAPH}_{4,3}$ of \Cref{fig:example-maph} with the (made-up) rates
\begin{equation}
\label{eq:verif-truth}
\alpha = (0.4,\,0.3,\,0.2,\,0.1),\quad
T = \begin{pmatrix} -3.0 & 1.0 & 0 & 0\\ 0.5 & -3.0 & 1.5 & 0\\ 0 & 0.5 & -3.0 & 1.0\\ 1.0 & 0 & 0.5 & -3.5 \end{pmatrix},\quad
D = \begin{pmatrix} 2.0 & 0 & 0\\ 0 & 1.0 & 0\\ 0 & 0.5 & 1.0\\ 0 & 0 & 2.0 \end{pmatrix}.
\end{equation}
Reading the three competing causes as discharge types -- $s^0_1$ a completely healthy discharge and $s^0_2$, $s^0_3$ a discharge requiring further care of type~2 and type~3 -- a patient enters one of the four transient phases according to $\alpha$, moves along the backbone of \Cref{fig:example-maph}, and is eventually absorbed into one of the three discharge types. Evaluating the formulas of \Cref{prop:maph-props} gives the marginal absorption probabilities $\Prob(\kappa=k)=(0.384,\,0.282,\,0.335)$, the conditional means $\ex{\tau\mid\kappa=k}=(0.54,\,0.66,\,0.74)$, and the conditional squared coefficients of variation $(1.18,\,0.96,\,0.87)$, the last straddling unity. \Cref{fig:running-subdist} plots the corresponding sub-distribution functions $F(u,k)$ and sub-densities $f(u,k)$ of items (i)--(ii), one curve per cause; the sub-distribution functions reappear in \Cref{fig:verif-cdf} as the target of the fit.
\end{example}

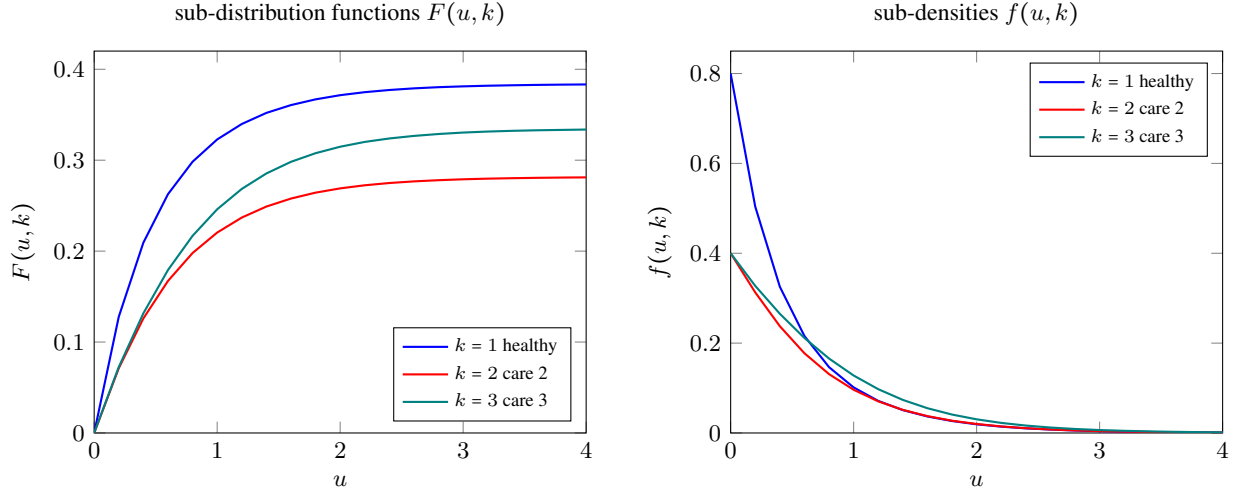
\begin{figure}[ht]
\centering
\begin{minipage}{0.49\textwidth}
\centering
\begin{tikzpicture}
\begin{axis}[width=\textwidth, height=0.82\textwidth,
  title={sub-distribution functions $F(u,k)$}, title style={font=\small},
  xlabel={$u$}, ylabel={$F(u,k)$},
  xmin=0, xmax=4, ymin=0, ymax=0.42, legend pos=south east,
  legend cell align=left, legend style={font=\scriptsize}, tick label style={font=\small}, label style={font=\small}]
\addplot[blue, thick] coordinates {(0.00,0.0000)(0.20,0.1278)(0.40,0.2093)(0.60,0.2626)(0.80,0.2983)(1.00,0.3228)(1.20,0.3398)(1.40,0.3520)(1.60,0.3606)(1.80,0.3669)(2.00,0.3715)(2.20,0.3748)(2.40,0.3772)(2.60,0.3790)(2.80,0.3803)(3.00,0.3812)(3.20,0.3819)(3.40,0.3824)(3.60,0.3828)(3.80,0.3831)(4.00,0.3833)};
\addlegendentry{$k=1$ healthy}
\addplot[red, thick] coordinates {(0.00,0.0000)(0.20,0.0712)(0.40,0.1260)(0.60,0.1672)(0.80,0.1979)(1.00,0.2204)(1.20,0.2369)(1.40,0.2490)(1.60,0.2578)(1.80,0.2643)(2.00,0.2690)(2.20,0.2724)(2.40,0.2749)(2.60,0.2767)(2.80,0.2780)(3.00,0.2790)(3.20,0.2797)(3.40,0.2802)(3.60,0.2806)(3.80,0.2809)(4.00,0.2811)};
\addlegendentry{$k=2$ care~2}
\addplot[teal, thick] coordinates {(0.00,0.0000)(0.20,0.0725)(0.40,0.1316)(0.60,0.1792)(0.80,0.2168)(1.00,0.2460)(1.20,0.2684)(1.40,0.2854)(1.60,0.2982)(1.80,0.3077)(2.00,0.3148)(2.20,0.3201)(2.40,0.3239)(2.60,0.3267)(2.80,0.3288)(3.00,0.3304)(3.20,0.3315)(3.40,0.3323)(3.60,0.3329)(3.80,0.3333)(4.00,0.3337)};
\addlegendentry{$k=3$ care~3}
\end{axis}
\end{tikzpicture}
\end{minipage}\hfill
\begin{minipage}{0.49\textwidth}
\centering
\begin{tikzpicture}
\begin{axis}[width=\textwidth, height=0.82\textwidth,
  title={sub-densities $f(u,k)$}, title style={font=\small},
  xlabel={$u$}, ylabel={$f(u,k)$},
  xmin=0, xmax=4, ymin=0, ymax=0.85, legend pos=north east,
  legend cell align=left, legend style={font=\scriptsize}, tick label style={font=\small}, label style={font=\small}]
\addplot[blue, thick] coordinates {(0.00,0.8000)(0.20,0.5038)(0.40,0.3255)(0.60,0.2158)(0.80,0.1465)(1.00,0.1014)(1.20,0.0714)(1.40,0.0509)(1.60,0.0367)(1.80,0.0266)(2.00,0.0193)(2.20,0.0141)(2.40,0.0103)(2.60,0.0076)(2.80,0.0055)(3.00,0.0041)(3.20,0.0030)(3.40,0.0022)(3.60,0.0016)(3.80,0.0012)(4.00,0.0009)};
\addlegendentry{$k=1$ healthy}
\addplot[red, thick] coordinates {(0.00,0.4000)(0.20,0.3129)(0.40,0.2377)(0.60,0.1774)(0.80,0.1310)(1.00,0.0962)(1.20,0.0704)(1.40,0.0514)(1.60,0.0375)(1.80,0.0273)(2.00,0.0200)(2.20,0.0146)(2.40,0.0106)(2.60,0.0078)(2.80,0.0057)(3.00,0.0042)(3.20,0.0030)(3.40,0.0022)(3.60,0.0016)(3.80,0.0012)(4.00,0.0009)};
\addlegendentry{$k=2$ care~2}
\addplot[teal, thick] coordinates {(0.00,0.4000)(0.20,0.3274)(0.40,0.2653)(0.60,0.2114)(0.80,0.1657)(1.00,0.1279)(1.20,0.0975)(1.40,0.0736)(1.60,0.0551)(1.80,0.0410)(2.00,0.0304)(2.20,0.0224)(2.40,0.0165)(2.60,0.0121)(2.80,0.0089)(3.00,0.0065)(3.20,0.0048)(3.40,0.0035)(3.60,0.0026)(3.80,0.0019)(4.00,0.0014)};
\addlegendentry{$k=3$ care~3}
\end{axis}
\end{tikzpicture}
\end{minipage}
\caption{The running $\text{MAPH}_{4,3}$ example of \eqref{eq:verif-truth}: sub-distribution functions $F(u,k)=\Prob(\tau\le u,\kappa=k)$ (left) and sub-densities $f(u,k)=\alpha e^{Tu}D_k$ (right) for the three competing discharge types, from items (i)--(ii) of \Cref{prop:maph-props}. Each sub-distribution function rises to its cause's marginal probability $\Prob(\kappa=k)$; the left panel reappears as the fitting target in \Cref{fig:verif-cdf}.}
\label{fig:running-subdist}
\end{figure}

\subsection{The embedded jump chain and the absorption probabilities}
\label{sec:second-param}

\subchg{The competing-risks content of an $\text{MAPH}_{m,n}$ distribution is carried by the matrix
$R$ of absorption probabilities: $\rho_{ik}$ is the probability that a trajectory started in phase
$s_i$ is eventually absorbed in $s^0_k$. This section constructs $R$ from $(T,D)$, records how it
interacts with the jump chain, and derives the bounds relating the two that are used later.}

We consider the embedded Markov chain (jump chain) associated with $X(\cdot)$. For this consider the transition epochs of $X(\cdot)$ as $0 = \sigma_0 < \sigma_1 < \sigma_2 < \ldots$. Keeping in mind that trajectories of $X(\cdot)$ are right continuous with left limits, denote $I_j=X(\sigma_j)$. Here $\{I_j\}_{j=0}^\infty$ is a discrete time Markov chain where for some finite $N$, for all $j \ge N$, $I_j$ is at a fixed state in $\cs^0$ . We denote this fixed absorbing state via $I_\infty$. 

This embedded Markov chain is based on non-trivial probabilities denoted via $p_{ij}^\lambda$ for $s_i,s_j \in \cs$ and $p_{ik}^\mu$ for $s_i \in \cs$ and $s^0_k \in \cs^0$. Specifically,
\begin{align}
\label{eq:transition_prob}
   p_{ij}^\lambda & :=\mathbb{P}(I_{1}=s_j\vert I_0=s_i)=
\begin{cases}
\frac{\lambda_{ij}}{q_i}, & j\neq i,\\
0, & j = i.
\end{cases}  
\\
p_{ik}^\mu & :=\mathbb{P}(I_{1}=s^0_k\vert I_0=s_i)=
\frac{\mu_{ik}}{q_i}.
\end{align}
These probabilities are organized in an $(m+n) \times (m+n)$ transition probability matrix $P$,
\begin{equation}\label{eq:transition_matrix}
    P = \begin{bmatrix}
        P^\lambda & P^\mu\\
        0_{n\times m} &  I_{n \times n}
    \end{bmatrix},
\end{equation}
where $P^\lambda$ is the matrix of $p_{ij}^\lambda$ entries and $P^\mu$ is the matrix of $p_{ik}^\mu$ entries. Moreover,
we have $P^\lambda = I_{m\times m} - \mathrm{Diag}(T)^{-1} T$  and $P^\mu = -\mathrm{Diag}(T)^{-1}D$, where $\mathrm{Diag}(T)$ is the diagonal matrix containing the diagonal elements of the matrix $T$ and thus has entries on the diagonal $-q_1,\ldots,-q_m$.

We are also interested in additional quantities associated with the jump chain. Specifically for $i$ such that $s_i \in \cs$ and $k$ such that $s^0_k \in \cs^0$ we denote,
\begin{equation}\label{eq:rho_ik}
\rho_{ik}:= \mathbb{P}(I_{\infty}= s^0_k ~\vert~ I_0 = s_i).
\end{equation}
These are the absorption probabilities for the competing risks. \subchg{The entries $\rho_{ik}$ are non-negative, and $\rho_{ik}=0$ means that cause $k$ cannot occur from phase $s_i$.} We denote $R$ as the $m \times n$ matrix with entries $\rho_{ik}$, and a standard ``first step analysis'' computation yields,
\begin{equation}
\label{eq:r-matrix}    
R = (I - P^\lambda)^{-1}P^\mu = -T^{-1}D.
\end{equation}

The inverse $T^{-1}$ in \eqref{eq:r-matrix} is well defined: as established preceding \Cref{prop:maph-props}, absorption being certain makes $-T$ a non-singular $M$-matrix, so $T$ is invertible with $-T^{-1} \ge 0$. It follows that $R = -T^{-1}D \ge 0$ is well defined and, using $D\mathbf{1}_n = -T\mathbf{1}_m$, is row-stochastic since $R\mathbf{1}_n = -T^{-1}D\mathbf{1}_n = T^{-1}T\mathbf{1}_m = \mathbf{1}_m$; \subchg{with $\sum_k \rho_{ik} = 1$ for every $i$}.

Further, we are interested in conditional transition probabilities for transitions between states in $\cs$, when conditioning on eventually being absorbed in a specific state. Specifically for $i,j \in \{1,\ldots,m\}$ and $k \in \{1,\ldots,n\}$, define,
\[
p^\lambda_{ij |k} = \mathbb{P}(I_{1}=s_j\vert I_0 = s_i, I_{\infty}=s^0_k).
\]
The conditional and unconditional one-step probabilities determine one another through the
absorption probabilities, in both directions, as recorded next.

\begin{subrevision}
\begin{lemma}
\label{lem:cond-jump-probs}
Let $s_i, s_j \in \cs$ with $i \neq j$, and let $s^0_k \in \cs^0$. Then the following hold.

\noindent \textup{(i)} If $\rho_{ik} > 0$, the conditioning event $\{I_0 = s_i,\, I_\infty = s^0_k\}$
has positive probability and
\begin{align}
\label{eq:main-pijgivenk-equation}
p^\lambda_{ij|k} &= p_{ij}^{\lambda}\, \frac{\rho_{jk}}{\rho_{ik}}.
\end{align}
\noindent \textup{(ii)} If in addition $\rho_{jk} > 0$, then \eqref{eq:main-pijgivenk-equation} may be
inverted,
\begin{align}
\label{eq:pij-from-pijk-cond}
p_{ij}^{\lambda} &= p^\lambda_{ij|k}\, \frac{\rho_{ik}}{\rho_{jk}},
\end{align}
and, since the left hand side of \eqref{eq:pij-from-pijk-cond} does not depend on $k$, the right hand
side takes the same value for every $k$ with $\rho_{ik}, \rho_{jk} > 0$; equivalently, for any two
such indices $k$ and $l$,
\begin{equation}
\label{eq:pijk-consistency}
p^\lambda_{ij|k}\, \frac{\rho_{ik}}{\rho_{jk}}
=
p^\lambda_{ij|l}\, \frac{\rho_{il}}{\rho_{jl}}.
\end{equation}
\noindent \textup{(iii)} If $\rho_{ik} > 0$ but $\rho_{jk} = 0$, then \eqref{eq:main-pijgivenk-equation}
still holds and yields $p^\lambda_{ij|k} = 0$, whereas $p^\lambda_{ij}$ may be strictly positive. In
this case $p^\lambda_{ij}$ is not recoverable from $p^\lambda_{ij|k}$ and
\eqref{eq:pij-from-pijk-cond} is unavailable: conditioning on absorption in $s^0_k$ carries no
information about jumps into states from which $s^0_k$ is unreachable.

\noindent \textup{(iv)} If $\rho_{ik} = 0$, the conditioning event is null, so $p^\lambda_{ij|k}$ is
undefined; we then adopt the convention $p^\lambda_{ij|k} := 0$. In this case necessarily
$p^\mu_{ik} = 0$ and $p^\lambda_{ij}\, \rho_{jk} = 0$ for every $j$, so $\rho_{jk} = 0$ for every $j$
with $p^\lambda_{ij} > 0$.

\noindent \textup{(v)} Without any positivity requirement, and with the convention of
\textup{(iv)} in force,
\begin{equation}
\label{eq:pijk-aggregation}
p^\lambda_{ij} = \sum_{k=1}^n \rho_{ik}\, p^\lambda_{ij|k} .
\end{equation}
\noindent In all cases $p^\lambda_{ii|k} = p^\lambda_{ii} = 0$. By \textup{(iii)} and \textup{(iv)} a
single conditional slice need not determine $p^\lambda_{ij}$, whereas the whole family of slices,
weighted by the absorption probabilities as in \textup{(v)}, returns $p^\lambda_{ij}$ whatever the
pattern of zeros in $R$. It is \textup{(v)} that is used in the sequel.
\end{lemma}

\begin{proof}
Given in \Cref{app:proofs-maph}.
\end{proof}
\end{subrevision}

We note that for any $i \in \{1,\ldots,m\}$ and $k \in \{1,\ldots,n\}$, we can represent the probability of immediate (one step) absorption in $s_k^0$ given that we start in $s_i$ and know that we ultimately get absorbed in $s_k^0$. We denote this as, 
\begin{equation}
\label{eq:pik-mu-from-pijk-1}  
p^\mu_{ik|k} 
=
\mathbb{P}(I_1 = s_k^0 ~\vert~ I_0 = s_i, I_
\infty = s_k^0),
\end{equation}
\subchg{Take $\rho_{ik}>0$ for the moment, so that case \textup{(i)} of \Cref{lem:cond-jump-probs}
applies.} We represent it as,
\begin{equation}
\label{eq:pik-mu-from-pijk-2}  
p^\mu_{ik|k}  = 1 - \sum_{j=1}^m p^\lambda_{ij |k} = 1 - \sum_{j = 1}^{m} p_{ij}^{\lambda} \frac{\rho_{jk}}{\rho_{ik}} .
\end{equation}

\subchg{Since $\{I_1 = s^0_k\} \subseteq \{I_\infty = s^0_k\}$ -- a first jump into $s^0_k$ is
already absorption there -- the event measured by $p^\mu_{ik|k}$ is the one measured by the
unconditional $p^\mu_{ik}$ of \eqref{eq:transition_prob}, and the two differ only through the
conditioning. Hence} $p^\mu_{ik|k} \ge p^\mu_{ik}$, with \subchg{equality exactly when
$\rho_{ik}=1$, and} strict inequality whenever $p^\mu_{ik}>0$ and $\rho_{ik}<1$. To see this, observe,
\begin{equation}
\label{eq:p-ik-mu-der}
p^\mu_{ik|k} = \frac{\mathbb{P}(I_1 = s_k^0,I_0 = s_i)}{\mathbb{P}(I_0 = s_i, I_
\infty = s_k^0)} = \frac{p^\mu_{ik} \alpha_i}{\rho_{ik} \alpha_i} =\frac{p^\mu_{ik} }{\rho_{ik}},
\end{equation}
and notice that $\rho_{ik} \in (0,1]$ in this case, with $\rho_{ik}<1$ whenever $n \ge 2$ and some other cause is possible from $s_i$. \subchg{If instead $\rho_{ik}=0$ then, by \Cref{lem:cond-jump-probs}\,\textup{(iv)}, $p^\mu_{ik}=0$ and both $p^\mu_{ik|k}$ and every $p^\lambda_{ij|k}$ vanish under the convention adopted there, so \eqref{eq:pik-mu-from-pijk-2} continues to hold while \eqref{eq:p-ik-mu-der} is empty.} Note further that by equating \eqref{eq:pik-mu-from-pijk-2} and \eqref{eq:p-ik-mu-der} we recover the first step identity,
\begin{equation}
\label{eq:first-step-1}
p^\mu_{ik} + \sum_{j = 1}^{m} p_{ij}^{\lambda} \rho_{jk} = \rho_{ik} .
\end{equation}
\subchg{Written in matrix form this is $P^\lambda R + P^\mu = R$, the row-wise statement of
\eqref{eq:r-matrix}, so it holds for all $i$ and $k$ with no positivity required.}

\subchg{The jump matrix $P^\lambda$ and the absorption matrix $R$ are not independent of one another.
Together they determine the conditional slices, by \Cref{lem:cond-jump-probs}\,\textup{(i)},}
\begin{equation}
\label{eq:pij-relationship-const-2}
p_{ij|k}^\lambda = p^\lambda_{ij} \frac{\rho_{jk}}{\rho_{ik}}
\qquad \text{whenever } \rho_{ik}>0,
\end{equation}
\subchg{with $p^\lambda_{ij|k}=0$ when $\rho_{ik}=0$; and, conversely, $P^\lambda$ is recovered from
the slices by the aggregation identity \eqref{eq:pijk-aggregation}. They are also linked by a family
of inequalities, which we record next because they are what confines $R$ once $P^\lambda$ is known:
the exit probabilities implied by the first-step identity must be non-negative. The lemma states this
for arbitrary matrices $R$ and $P^\lambda$, since it is used both for a fitted model and, in
\Cref{sec:init-heuristic}, for a constructed one.}

\begin{subrevision}
\begin{lemma}
\label{lem:feasible-jump}
Let $R$ be a non-negative row-stochastic $m \times n$ matrix and let $P^\lambda$ be a non-negative
$m \times m$ matrix with zero diagonal. Define $P^\mu$ by the first-step identity
\eqref{eq:first-step-1},
\[
p^\mu_{ik} := \rho_{ik} - \sum_{j=1}^m p^\lambda_{ij}\, \rho_{jk},
\qquad i=1,\ldots,m, \quad k=1,\ldots,n.
\]
Then the following hold.

\noindent \textup{(i)} For every $i$,
$\sum_{j=1}^m p^\lambda_{ij} + \sum_{k=1}^n p^\mu_{ik} = 1$.

\noindent \textup{(ii)} $P^\mu \ge 0$ if and only if
\begin{equation}
\label{eq:second-param-main-constraint}
\sum_{j=1}^m \rho_{jk}\, p^\lambda_{ij}
\le \rho_{ik}
\quad
\text{for}
\quad
i=1,\ldots,m, \quad k=1,\ldots,n,
\end{equation}
that is, if and only if $P^\lambda R \le R$ entrywise. In that case $[\,P^\lambda \mid P^\mu\,]$ is a
stochastic matrix; in particular every $p^\mu_{ik} \in [0,1]$ and $P^\lambda$ is sub-stochastic, so
sub-stochasticity is implied by \eqref{eq:second-param-main-constraint} rather than separately
imposed.

\noindent \textup{(iii)} For a given $i$, the $n$ constraints of row $i$ in
\eqref{eq:second-param-main-constraint} all hold with equality if and only if row $i$ of $P^\mu$
vanishes, equivalently $\sum_{j} p^\lambda_{ij} = 1$.

\noindent \textup{(iv)} If moreover $P^\lambda$ and $R$ are the jump matrix \eqref{eq:transition_prob}
and the absorption matrix \eqref{eq:rho_ik} of an $\text{MAPH}_{m,n}$ as defined above, then
$P^\mu$ agrees with the one-step absorption matrix of \eqref{eq:transition_prob} and the conditional
slices are recovered from $P^\lambda$ and $R$ by
\begin{equation}
\label{eq:pij-relationship-const-2-lem}
p_{ij|k}^\lambda = p^\lambda_{ij} \frac{\rho_{jk}}{\rho_{ik}}
\qquad \text{whenever } \rho_{ik}>0,
\end{equation}
and $p^\lambda_{ij|k}=0$ when $\rho_{ik}=0$. Conversely $P^\lambda$ is recovered from the slices by
the aggregation identity \eqref{eq:pijk-aggregation}.
\end{lemma}

\begin{proof}
Given in \Cref{app:proofs-maph}.
\end{proof}
\end{subrevision}

\subchg{The constraints \eqref{eq:second-param-main-constraint} are what confine $P^\lambda$, for a
given $R$, to the tuples that correspond to an actual MAPH distribution: with $R$ carrying the
absorption probabilities, $P^\mu$ is no longer free, and \Cref{lem:feasible-jump}\,\textup{(ii)} is
exactly the requirement that the implied one-step absorption probabilities be non-negative.
Equivalently, $D=-TR$ has no negative entry. There are
$m \times n$ such constraints on the $m^2$ elements of $P^\lambda$, together with a further $m^2$
non-negativity constraints.}

\subchg{Two consequences of \Cref{lem:feasible-jump} are worth isolating, since they are the bounds
the rest of the paper uses. First, by \textup{(ii)}, the absorption probabilities of any
$\text{MAPH}_{m,n}$ dominate their one-step averages,
\begin{equation}
\label{eq:R-bound}
P^\lambda R \le R
\qquad \text{entrywise},
\end{equation}
with the slack in entry $(i,k)$ equal to the one-step absorption probability $p^\mu_{ik}$; the
inequality is an equality in row $i$ for every $k$ exactly when $s_i$ cannot be left for absorption
in one step. Second, by \textup{(i)}, the row sums of $P^\lambda$ and $P^\mu$ are complementary, so
$P^\lambda$ is sub-stochastic and its $i$-th row sum is the probability of taking more than one step
to be absorbed from $s_i$. Both are properties of every MAPH distribution, not conditions to be
imposed: given $(T,D)$ they hold automatically, and given a construction that specifies $R$ and
$P^\lambda$ directly, as in \Cref{sec:init-heuristic}, they are what has to be checked.}

%% file: sections/estimation.tex
\section{Estimation and the EM Algorithm}
\label{sec:full-path-mle}

We now momentarily assume full observations, where we denote the collection of observations as $\X = \{X^\ell,~\ell=1,\ldots,L \}$, with $L$ samples of fully observed Markov chains.

Each $X^\ell=\{X^\ell(t):t\geq 0\}$ has epoch times $0 = \sigma_0^\ell < \sigma_1^\ell < \sigma_2^\ell < \cdots$ and an embedded jump chain $\{I_n^\ell\}$ with $I_n^\ell = X^\ell(\sigma_n^\ell)$, indexed by discrete time $n$; we write $i_n^\ell$ for the index of the visited state $I_n^\ell$. As before, $I_\infty^\ell \in \cs^0$ denotes the absorbing state, and we set $n_0^\ell = \inf\{n : I_n^\ell = I_\infty^\ell\}$, the number of jumps until absorption, so that $I_0^\ell,\ldots,I_{n_0^\ell-1}^\ell$ are transient and $I_{n_0^\ell}^\ell = I_\infty^\ell$ is absorbing. The path thus starts in $I_0^\ell = s_{i_0^\ell}$, with index $i_0^\ell$, and is absorbed in $I_{n_0^\ell}^\ell = s^0_{k_0^\ell}$, with index $k_0^\ell$; its last transient state before absorption is $I_{n_0^\ell-1}^\ell$, with index $i_{n_0^\ell-1}^\ell$. We also write $X^\ell_k$ for a path absorbed in $s_k^0$, i.e.\ one with $k_0^\ell = k$. We suppress the superscript $\ell$ when a single path is understood; for instance the proofs below write $n_0, \sigma_n, i_n$ for $n_0^\ell, \sigma_n^\ell, i_n^\ell$.

For each path $X^\ell$ we use the per-path statistics listed in the upper block of \Cref{tab:stats}, and we aggregate them over all $L$ observations into the per-sample statistics in the lower block. Throughout, $s_i, s_j \in \cs$ index transient phases ($i,j = 1,\ldots,m$) and $s_k^0 \in \cs^0$ indexes absorbing states ($k = 1,\ldots,n$).

\begin{table}[ht]
\centering
\caption{Per-path (left) and aggregated per-sample (right) sufficient statistics used in the complete-data likelihood for MAPH parameter estimation, each given as $\text{symbol} = \text{definition}$ and aligned row-by-row by symbol. Indices range over transient phases $s_i, s_j \in \cs$ ($i,j = 1,\ldots,m$) and absorbing states $s_k^0 \in \cs^0$ ($k = 1,\ldots,n$).}
\label{tab:stats}
\footnotesize
\renewcommand{\arraystretch}{1.35}
\setlength{\tabcolsep}{4pt}
\begin{tabular}{@{}l p{0.165\textwidth} @{\quad} l p{0.205\textwidth}@{}}
\toprule
\multicolumn{2}{@{}l}{\emph{Per-path, for a single path $X^\ell$:}} & \multicolumn{2}{l}{\emph{Per-sample, aggregated over $\ell = 1,\ldots,L$:}} \\
\cmidrule(r){1-2}\cmidrule(l){3-4}
Statistic & Meaning & Statistic & Meaning \\
\midrule
$B_{i}^\ell = \mathbf{1}\{I_0^\ell = s_i\}$ & Starts in phase $s_i$. & $B_{i} = \sum_{k} B_{ik}$ & Paths starting in $s_i$. \\
$B_{ik}^\ell = B_i^\ell J_k^\ell$ & Starts in $s_i$, absorbed in $s_k^0$. & $B_{ik} = \sum_{\ell} B_{ik}^\ell$ & Paths starting in $s_i$, absorbed in $s_k^0$. \\
$J_k^\ell = \mathbf{1}\{k_0^\ell = k\}$ & Absorbed in $s_k^0$. & & \\
$E_{ik}^\ell = \mathbf{1}\{I_{n_0-1}^\ell = s_i, I_{n_0}^\ell = s_k^0\}$ & Last transient phase $s_i$, absorbed in $s_k^0$. & $E_{ik} = \sum_{\ell} E_{ik}^\ell$ & Last transient phase $s_i$, absorbed in $s_k^0$. \\
$M_{ij}^\ell = \sum_{n\ge1} \mathbf{1}\{I_{n-1}^\ell = s_i, I_n^\ell = s_j\}$ & Count of $s_i\!\to\!s_j$ transitions ($M_{ii}^\ell=0$). & $M_{ijk} = \sum_{\ell} M_{ij}^\ell J_k^\ell$ & $s_i\!\to\!s_j$ transitions over paths absorbed in $s_k^0$. \\
 & & $M_{ij} = \sum_{\ell} M_{ij}^\ell = \sum_k M_{ijk}$ & $s_i\!\to\!s_j$ transitions over all paths. \\
$N_{i}^\ell = \sum_{k} E_{ik}^\ell + \sum_{j} M_{ij}^\ell$ & Transitions leaving $s_i$. & $N_{i} = \sum_{\ell} N_i^\ell$ & Transitions leaving $s_i$. \\
 & & $N_{ik} = \sum_{\ell} N_i^\ell J_k^\ell$ & Transitions leaving $s_i$ over paths absorbed in $s_k^0$. \\
$Z_i^\ell = \int_0^\infty \mathbf{1}\{X^\ell(t) = s_i\}\,dt$ & Time spent in $s_i$. & $Z_i = \sum_{\ell} Z_i^\ell$ & Time spent in $s_i$. \\
\bottomrule
\end{tabular}
\end{table}

For a fixed path $X^\ell$, exactly one of the indicators $\{B_i^\ell\}_i$ equals one and exactly one of $\{J_k^\ell\}_k$ equals one, hence exactly one $B_{ik}^\ell$ equals one. \subchg{The event $\{J_k^\ell = 1\}$ is precisely the event that the path is of type $X^\ell_k$, so the two carry the same information and we write $f(X^\ell_k\,;\theta)$ for the density of a path absorbed in $s^0_k$; the indicator $J_k^\ell$ is retained only where it selects a term or aggregates a statistic by cause.} Finally, the per-sample counts satisfy
\begin{equation}
\label{eq:M-N-E-relationship}
\sum_{j=1}^m M_{ijk} = N_{ik} - E_{ik}.
\end{equation}
\subchg{Per path there is also a flow balance at each phase: every visit to $s_i$ is entered once and
left once, the entries being either the start of the path or a jump from another transient phase, and
the exits being either a jump to another transient phase or the final jump to absorption. Hence}
\begin{equation}
\label{eq:flow-balance}
B_i^\ell + \sum_{j=1}^m M_{ji}^\ell
\;=\;
N_i^\ell
\;=\;
\sum_{j=1}^m M_{ij}^\ell + \sum_{k=1}^n E_{ik}^\ell,
\qquad i=1,\ldots,m,
\end{equation}
\subchg{where at most one $E^\ell_{ik}$ is non-zero, namely at the last transient phase of the path
and its own cause. Multiplying the right hand equality by $J^\ell_k$ and summing over $\ell$ returns
\eqref{eq:M-N-E-relationship}.}

Note that most of the above statistics are quite standard in the Markov jump process estimation literature; see for example \citet{Albert}. However, the additional statistics that depend on the eventual absorbing state, $s_k^0$, are less common and are novel for our approach.

\subchg{With these statistics defined, \Cref{prop:path-likelihood} gives the likelihood of a fully
observed path, and of the sample, working directly with the embedded jump chain $I_0,I_1,\ldots$ and
the jump matrix $P^\lambda$.}

\begin{subrevision}
\begin{proposition}
\label{prop:path-likelihood}
Let $\theta=(\alpha,q,P^\lambda,P^\mu)$ and let the statistics be those of \Cref{tab:stats}, with the
convention $0\log 0=0$ and an empty product equal to $1$. Then the following hold.

\noindent \textup{(i)} \emph{A path.} For a fully observed path $X^\ell$,
\begin{equation}
\label{eq:path-likelihood-general}
f(X^\ell~;~\theta)
=
\Big( \prod_{i=1}^m \alpha_i^{B^\ell_i} \Big)
\Big( \prod_{i=1}^m q_i^{N^\ell_i} e^{-q_i Z^\ell_i} \Big)
\Big( \prod_{i=1}^m \prod_{j=1}^m (p^\lambda_{ij})^{M^\ell_{ij}} \Big)
\Big( \prod_{i=1}^m \prod_{k=1}^n (p^\mu_{ik})^{E^\ell_{ik}} \Big).
\end{equation}

\noindent \textup{(ii)} \emph{The aggregated observations.} For the sample
$\X=\{X^\ell,~\ell=1,\ldots,L\}$ of independent paths, with
${\cal L}(\theta)=\prod_{\ell=1}^L f(X^\ell~;~\theta)$ and the per-sample statistics of
\Cref{tab:stats} obtained by aggregating their per-path counterparts,
\begin{equation}
\label{eq:complete-log-like-jump}
\log{\cal L}(\theta)
=
\sum_{i=1}^m B_i \log \alpha_i
+ \sum_{i=1}^m \big( N_i \log q_i - q_i Z_i \big)
+ \sum_{i=1}^m \sum_{j=1}^m M_{ij} \log p^\lambda_{ij}
+ \sum_{i=1}^m \sum_{k=1}^n E_{ik} \log p^\mu_{ik}.
\end{equation}
\end{proposition}

\begin{proof}
Given in \Cref{app:proofs-est}.
\end{proof}
\end{subrevision}

\subchg{Equation \eqref{eq:complete-log-like-jump} is the form used in the sequel, and its statistics
are aggregated over causes. The E-step of \Cref{sec:em-alg} below, however, produces statistics that are
resolved by cause: each record contributes conditional expectations given its own observed absorbing
state, so the natural bookkeeping there is $B_{ik}$, $M_{ijk}$ and $E_{ik}$. The next lemma records
that the two bookkeepings agree. It is an equivalence between two expressions for one likelihood, not
a second model: the cause-split statistics may be used provided the probabilities multiplying them
are the cause-conditional ones of \Cref{lem:cond-jump-probs}.}

\begin{subrevision}
\begin{lemma}
\label{lem:cause-specific}
Let $\theta$ and the statistics be as in \Cref{prop:path-likelihood}, and let $R$,
$p^\lambda_{ij|k}$ and $p^\mu_{ik|k}$ be as in \Cref{sec:second-param}. Let $X^\ell_k$ denote a path
absorbed in $s^0_k$; every transient phase such a path visits has $\rho_{ik}>0$, so the conditional
probabilities are well defined along it. Then the density \eqref{eq:path-likelihood-general} may equally be written
\begin{equation}
\label{eq:path-likelihood-cause}
f(X^\ell~;~\theta)
\;=\;
f(X^\ell_k~;~\theta)
\;=\;
\underbrace{\Big( \prod_{i=1}^m (\alpha_i \rho_{ik})^{B^\ell_i} \Big)}_{\text{start, and cause } s^0_k}
\underbrace{\Big( \prod_{i=1}^m q_i^{N^\ell_i} e^{-q_i Z^\ell_i} \Big)
\Big( \prod_{i=1}^m \prod_{j=1}^m (p^\lambda_{ij|k})^{M^\ell_{ij}} \Big)
\Big( \prod_{i=1}^m (p^\mu_{ik|k})^{E^\ell_{ik}} \Big)}_{\text{the path, given that cause}} ,
\end{equation}
in which the cause is drawn first, with probability $\prod_{i=1}^m \rho_{ik}^{B^\ell_i}$, the
absorption probability of the phase in which the path starts, and the path is then generated by the
conditional probabilities.

\end{lemma}

\begin{proof}
Given in \Cref{app:proofs-est}.
\end{proof}
\end{subrevision}

\subsection{\subchg{The maximization problem}}
\label{sec:mle-exact}

\subchg{The complete-data log-likelihood is not observable: only $(\tau,\kappa)$ is seen, not the path. The
EM iteration replaces it by its conditional expectation given the data. Writing $\theta^{(\nu)}$ for
the current iterate and ${\cal D}$ for the observed sample, the M-step solves
\begin{equation}
\label{eq:em-objective}
\theta^{(\nu+1)} = \arg\max_{\theta \, \in \, \Theta} \; Q\big(\theta \mid \theta^{(\nu)}\big),
\qquad
Q\big(\theta \mid \theta^{(\nu)}\big)
= \ex{ \log {\cal L}(\theta) ~\big|~ {\cal D} },
\end{equation}
the conditional expectation being taken under $\theta^{(\nu)}$ and the maximization being subject to
the constraints that define a MAPH distribution in jump-chain coordinates,
\begin{equation}
\label{eq:theta-set}
\Theta = \left\{
(\alpha,q,P^\lambda,P^\mu)
~:~
\begin{aligned}
&\alpha \ge 0, \quad \textstyle\sum_{i=1}^m \alpha_i = 1, \quad q>0,\\
&P^\lambda \ge 0, \quad p^\lambda_{ii}=0, \quad P^\mu \ge 0,\\
&\textstyle\sum_{j=1}^m p^\lambda_{ij} + \sum_{k=1}^n p^\mu_{ik} = 1, \quad i=1,\ldots,m
\end{aligned}
\right\} :
\end{equation}
$\alpha$ is a probability vector, the exit rates are positive, and each row of
$[\,P^\lambda \mid P^\mu\,]$ is a probability vector, being the jump-chain row \eqref{eq:transition_matrix}
of the corresponding phase. By \eqref{eq:transition_prob} the triple $(q,P^\lambda,P^\mu)$ is $(T,D)$ in
jump-chain coordinates, through $T=\mathrm{Diag}(q)(P^\lambda-I)$ and $D=\mathrm{Diag}(q)P^\mu$, so
$\Theta$ is the parameterization of \Cref{sec:maph-dist} and carries no restriction beyond its
standing assumption. By
\Cref{prop:path-likelihood}\,\textup{(ii)} the complete-data log-likelihood is \emph{linear} in the
statistics, so
\[
Q\big(\theta \mid \theta^{(\nu)}\big)
=
\sum_{i=1}^m \ex{B_i \mid {\cal D}} \log \alpha_i
+ \sum_{i=1}^m \big( \ex{N_i \mid {\cal D}} \log q_i - q_i \ex{Z_i \mid {\cal D}} \big)
+ \sum_{i,j} \ex{M_{ij} \mid {\cal D}} \log p^\lambda_{ij}
+ \sum_{i,k} \ex{E_{ik} \mid {\cal D}} \log p^\mu_{ik},
\]
which is \eqref{eq:complete-log-like-jump} with every statistic replaced by its conditional
expectation. Maximizing $Q$ is therefore the same problem as maximizing the complete-data
log-likelihood itself, and it is enough to solve the latter once. That is the content of the next
proposition; the conditional expectations are supplied in closed form by \Cref{prop:e-step}
(\Cref{app:e-step}).}

\begin{subrevision}
\begin{proposition}
\label{prop:approx-mle}
Suppose every phase is visited, $N_i>0$ for $i=1,\ldots,m$, and $L>0$. Over the parameter set
$\Theta$ of \eqref{eq:theta-set}, the complete-data log-likelihood
\eqref{eq:complete-log-like-jump} is maximized uniquely at
\begin{equation}\label{eq:mle}
\hat{\alpha}_{i}=\frac{B_i}{L}; \qquad
\hat{q}_{i}=\frac{N_{i}}{Z_{i}}; \qquad
\hat{p}^{\lambda}_{ij}=\frac{M_{ij}}{N_{i}} \quad (i \neq j), \quad \hat p^\lambda_{ii}=0; \qquad
\hat{p}^{\mu}_{ik}=\frac{E_{ik}}{N_{i}}.
\end{equation}
The maximizer lies in $\Theta$ as it stands: the ratios are non-negative, and the row sums are one by
the flow balance \eqref{eq:flow-balance}. Two further points are worth recording.

\noindent \textup{(i)} \emph{Absorption is certain, and $R$ is available.} $I-\widehat P^\lambda$ is
non-singular, so the fitted generator
$\widehat T = \mathrm{Diag}(\hat q)(\widehat P^\lambda - I)$ and
$\widehat D = \mathrm{Diag}(\hat q)\,\widehat P^\mu$ have $-\widehat T$ a non-singular $M$-matrix:
$(\hat\alpha,\widehat T,\widehat D)$ is an $\text{MAPH}_{m,n}$ distribution satisfying the standing
assumption of \Cref{sec:maph-dist}. Its absorption matrix
\begin{equation}
\label{eq:R-from-mle}
\widehat R = \big(I-\widehat P^\lambda\big)^{-1}\widehat P^\mu = -\widehat T^{-1}\widehat D
\end{equation}
is then non-negative and row-stochastic and satisfies the bound \eqref{eq:R-bound}; it is the
competing-risks summary of the fit, reported alongside it.

\noindent \textup{(ii)} \emph{Zero counts and non-integer statistics.} With the convention
$0\log0=0$, a vanishing count $M_{ij}=0$ or $E_{ik}=0$ simply sets the corresponding probability to
zero, and the maximizer remains unique. Nothing in the statement requires the statistics to be
integers, so \eqref{eq:mle} applies verbatim when they are replaced by conditional expectations, as
in \eqref{eq:em-objective}; the hypothesis $N_i>0$ then reads $\ex{N_i \mid {\cal D}}>0$, and it
implies $Z_i>0$ since a visited phase is occupied for a positive length of time.
\end{proposition}

\begin{proof}
Given in \Cref{app:proofs-est}.
\end{proof}

\begin{corollary}
\label{cor:monotone}
Let $\theta^{(\nu+1)}$ be obtained from $\theta^{(\nu)}$ by evaluating \eqref{eq:mle} at the
conditional expectations of the statistics given ${\cal D}$ under $\theta^{(\nu)}$. Then
$\theta^{(\nu+1)}$ is the exact maximizer of $Q(\,\cdot \mid \theta^{(\nu)})$ in
\eqref{eq:em-objective}, and consequently the observed-data log-likelihood
$\ell_{\mathrm{obs}}(\theta) = \log p({\cal D} \mid \theta)$ does not decrease,
\[
\ell_{\mathrm{obs}}\big(\theta^{(\nu+1)}\big) \;\ge\; \ell_{\mathrm{obs}}\big(\theta^{(\nu)}\big),
\]
with equality only if $Q(\theta^{(\nu+1)} \mid \theta^{(\nu)}) = Q(\theta^{(\nu)} \mid \theta^{(\nu)})$.
For a sample of exact observations
$\ell_{\mathrm{obs}}(\theta)=\sum_{\ell=1}^L \log f(t^\ell,k^\ell~;~\theta)$; the statement does not
depend on this form, so it covers the censored records of \Cref{sec:em-censored} as well.
\end{corollary}

\begin{proof}
Given in \Cref{app:proofs-est}.
\end{proof}
\end{subrevision}

\subchg{\Cref{cor:monotone} is what the exactness of the M-step buys: the iteration never
decreases the observed-data log-likelihood, so where it comes to rest is governed by that same
log-likelihood rather than by a surrogate. No such statement was available for a relaxed M-step,
whose fixed points solve a different problem.}

\subsection{The EM algorithm}
\label{sec:em-alg}

\begin{rcrevision}
We now assemble the developments above into an EM algorithm that fits an $\text{MAPH}_{m,n}$ distribution from exact and independently right-censored competing-risks observations. \aachg{The complete-data formulation remains unchanged, and only the conditional expectations in the E-step are modified.}
\end{rcrevision}

\begin{rcrevision}
Let \aachg{$c^\ell$} denote the censoring time, and define \aachg{$y^\ell=\min(\tau^\ell,c^\ell)$} and \aachg{$\delta^\ell=\mathbf 1\{\tau^\ell\le c^\ell\}$}. The observed data are
\end{rcrevision}
\[
{\cal D}=\{(y^\ell,\delta^\ell,k^\ell):\ell=1,\ldots,L\},
\qquad y^\ell>0,\quad \delta^\ell\in\{0,1\},
\]
\begin{rcrevision}
where $k^\ell\in\{1,\ldots,n\}$ is observed only if $\delta^\ell=1$. The latent event pair $(\tau^\ell,\kappa^\ell)$ for each subject is distributed according to the MAPH law of \aachg{\Cref{sec:maph-dist}}, independently over $\ell$. We observe either the pair itself or the right-censoring event $\{\tau^\ell>c^\ell\}$. \aachg{Censoring times are independent of the MAPH process and we condition on them, assuming non-informative right-censoring, so the censoring-law factors contain no MAPH parameters and their distribution can be omitted when estimating $(\alpha,T,D)$.} The phase order $m\ge1$ is selected externally. Let $d_k=|\{\ell:\delta^\ell=1,k^\ell=k\}|$ and $d=\sum_k d_k$, and \subchg{assume $d_k>0$ for every modelled cause, so that each cause is represented in the data; the causes are not otherwise ordered or distinguished}.
\end{rcrevision}

\begin{aarevision}
The two kinds of record are treated separately below. \Cref{sec:em-exact} develops the algorithm for exact observations; \Cref{sec:em-censored} then states the modifications required by right-censored records. A sample in which no record is censored is handled entirely by \Cref{sec:em-exact}.
\end{aarevision}

\subsubsection{Exact observations}
\label{sec:em-exact}

\begin{aarevision}
Throughout this subsection every record is an exact observation, so $\delta^\ell=1$ and $(y^\ell,k^\ell)=(\tau^\ell,\kappa^\ell)$ for every $\ell$.
\end{aarevision}

\begin{rcrevision}
The E-step uses $(\alpha,T,D)$ to compute conditional expectations of the full-path sufficient statistics in \Cref{tab:stats}. \aachg{An exact-event contribution conditions on $\{\tau=y,\kappa=k\}$.}
\end{rcrevision}

For a single observation with absorption time $t$ and absorbing state $k$, write
\[
f(t,k) = \alpha e^{Tt}D_k
\]
for the density of $(\tau,\kappa)=(t,k)$ (item (ii) of \Cref{prop:maph-props}), and let $\mathbf C(t,k)$ be the $m\times m$ matrix with entries
\begin{equation}
\label{eq:estep-fC}
C_{ij}(t,k) = \int_0^t \big(\alpha e^{Tu}\big)_i\,\big(e^{T(t-u)}D_k\big)_j \, du .
\end{equation}
Absorption in cause $k$ has positive probability for every cause represented in the data, and
then $f(t,k)>0$ for every $t>0$, so the regular conditional distribution given
$\{\tau=t,\kappa=k\}$ is well defined and the conditional expectations below exist; this is
verified in \Cref{app:e-step-exact}.

\begin{proposition}
\label{prop:e-step}
For an $\text{MAPH}_{m,n}(\alpha,T,D)$ distribution, the conditional expectations given $\{\tau=t,\ \kappa=k\}$ of the per-path statistics of \Cref{tab:stats} are
\begin{equation}
\label{eq:cond-exp-estats}
\begin{aligned}
\ex{B_i \mid \tau=t,\ \kappa=k} &= \frac{\alpha_i\,\big(e^{Tt}D_k\big)_i}{f(t,k)}, &\qquad
\ex{E_{ik} \mid \tau=t,\ \kappa=k} &= \frac{\big(\alpha e^{Tt}\big)_i\,D_{ik}}{f(t,k)},\\[2pt]
\ex{Z_i \mid \tau=t,\ \kappa=k} &= \frac{C_{ii}(t,k)}{f(t,k)}, &\qquad
\ex{M_{ij} \mid \tau=t,\ \kappa=k} &= \frac{T_{ij}\,C_{ij}(t,k)}{f(t,k)}\quad(i\neq j),
\end{aligned}
\end{equation}
with $\ex{M_{ii} \mid \tau=t,\ \kappa=k}=0$ \subchg{and, since the cause is observed,
$\ex{E_{ik'} \mid \tau=t,\ \kappa=k}=0$ for every $k'\neq k$}. The conditional expected total number
of exits from $s_i$ is
\begin{equation}
\label{eq:e-step-N-total}
\ex{N_i \mid \tau=t,\ \kappa=k} = \sum_{j=1}^m \ex{M_{ij} \mid \tau=t,\ \kappa=k} + \ex{E_{ik} \mid \tau=t,\ \kappa=k}.
\end{equation}
\end{proposition}

\begin{proof}
Given in \Cref{app:e-step-exact}.
\end{proof}

\begin{rcrevision}
\aachg{For each record, package these expectations into the cause-resolved contributions}
\end{rcrevision}
\begin{equation}
\label{eq:estats-exact}
\begin{aligned}
b_{i}^\ell&=\ex{B_i \mid \tau=y^\ell,\ \kappa=k^\ell},
&z_i^\ell&=\ex{Z_i \mid \tau=y^\ell,\ \kappa=k^\ell},\\[4pt]
n_{i}^\ell&=\ex{N_i \mid \tau=y^\ell,\ \kappa=k^\ell},
&m_{ij}^\ell&=\ex{M_{ij} \mid \tau=y^\ell,\ \kappa=k^\ell},\\[4pt]
e_{ik}^\ell&=\ind{k^\ell=k}\,\ex{E_{ik} \mid \tau=y^\ell,\ \kappa=k^\ell}. &&
\end{aligned}
\end{equation}
\subchg{Evaluated at $(t,k)=(y^\ell,k^\ell)$ these are exactly the per-record contributions
\eqref{eq:estats-exact} of a record with $\delta^\ell=1$, which the M-step \eqref{eq:approx-m-step}
sums over $\ell$. No further resolution by cause is needed: the conditioning on $\kappa=k^\ell$ is
already in the expectations themselves, and the only statistic that retains a cause index in the
M-step is $E_{ik}$.}

\begin{rcrevision}
\subchg{Only the exit statistic carries an indicator. The others are needed by \eqref{eq:approx-m-step} only in their cause-aggregated form, and for an exact record the aggregation is trivial: the expectations are already conditional on the observed cause $k^\ell$, so summing over causes returns the single term at $k=k^\ell$.} Replacing the full-path statistics in \Cref{prop:approx-mle} by the sums of \aachg{\eqref{eq:estats-exact}}, the M-step is
\end{rcrevision}
\begin{equation}
\label{eq:approx-m-step}
\begin{aligned}
\hat\alpha_i &= \frac1L\sum_{\ell=1}^L b_i^\ell, &\qquad
\hat q_i &= \frac{\sum_{\ell=1}^L n_i^\ell}{\sum_{\ell=1}^L z_i^\ell}, \\[4pt]
\hat p^\lambda_{ij} &= \frac{\sum_{\ell=1}^L m_{ij}^\ell}{\sum_{\ell=1}^L n_{i}^\ell}\quad(i\neq j),\quad \hat p^\lambda_{ii}=0,
&\qquad
\hat p^\mu_{ik} &= \frac{\sum_{\ell=1}^L e_{ik}^\ell}{\sum_{\ell=1}^L n_{i}^\ell}.
\end{aligned}
\end{equation}
\begin{rcrevision}
\subchg{This gives $(\hat\alpha,\hat q,\widehat P^\lambda,\widehat P^\mu)$, and the absorption matrix, when it is wanted, follows from \eqref{eq:R-from-mle}. Every completed path enters $\widehat P^\lambda$ and $\widehat P^\mu$, whatever its cause. By \Cref{prop:approx-mle} the update is the exact maximizer of the expected complete-data log-likelihood \eqref{eq:em-objective}, and it lands in the parameter set automatically: the row sums are one by the flow balance \eqref{eq:flow-balance}, so no constraint has to be enforced afterwards. The only requirement is that each phase be visited, that is $\sum_\ell n^\ell_i>0$; no absorption probability has to be bounded away from zero. In implementation the denominators should nevertheless be checked before an update. The parameters are converted to $\text{MAPH}_{m,n}(\alpha,T,D)$ by $T=\mathrm{Diag}(\hat q)(\widehat P^\lambda - I)$ and $D=\mathrm{Diag}(\hat q)\widehat P^\mu$.}
\end{rcrevision}

\begin{rcrevision}
The algorithm is then an iteration of these two steps. It takes the data ${\cal D}$, with
$d_k>0$ for every cause, and a phase order $m$, and starts from a strictly positive
$\text{MAPH}_{m,n}(\alpha,T,D)$, for which we use the censoring-compatible simplified
initialization of \Cref{sec:init-simplified}. Each sweep computes the exact-event expectations of
\Cref{prop:e-step} and forms the per-record contributions \eqref{eq:estats-exact} (the E-step),
forms $(\hat\alpha,\hat q,\widehat P^\lambda,\widehat P^\mu)$ via \eqref{eq:approx-m-step}
(the M-step), and converts back to $\text{MAPH}_{m,n}(\alpha,T,D)$ as above. Sweeps repeat until
the stopping rule of \Cref{sec:em-properties} is met, and the current
$\text{MAPH}_{m,n}(\alpha,T,D)$ is returned.
\end{rcrevision}

%% file: sections/censoring.tex
\section{Adaptations for Right Censoring}
\label{sec:censoring}
\label{sec:em-censored}

\begin{aarevision}
We now admit records with $\delta^\ell=0$. Everything in \Cref{sec:em-exact} carries over; this section states only what changes.
\end{aarevision}

\begin{rcrevision}
For the right-censored extension, ``full observation'' continues to mean the entire path through eventual absorption, not merely the path up to a censoring time. A censored record will therefore have its unobserved transient path, its post-censoring continuation, and its eventual absorbing state completed in the E-step. With this augmentation the sufficient statistics and complete-data likelihood derived above apply without alteration.
\end{rcrevision}

\begin{aarevision}
\Cref{prop:maph-props} describes what is observable when the pair $(\tau,\kappa)$ is seen in full. Under right censoring at a time $c$ the pair is not seen: all that is observed is the event $\{\tau>c\}$, and the eventual cause remains unknown. The quantities this leaves available are collected separately in the following proposition, which supplies the likelihood contribution of a censored record and the conditional law of what remains unobserved.

\begin{proposition}[Right-censored observation]
\label{prop:maph-censored}
Let $\zeta=(\tau,\kappa)$ be an $\text{MAPH}_{m,n}(\alpha,T,D)$ random pair and let $c>0$. Write
\[
S(c)=\Prob(\tau>c)=\alpha e^{Tc}\mathbf{1}_m = 1-F(c),
\]
which is strictly positive. Conditionally on the censoring event $\{\tau>c\}$ the following hold.
\begin{itemize}
    \item[(i)] The likelihood contribution of a record censored at $c$ is the survival probability $S(c)$.
    \item[(ii)] The conditional probability of eventual absorption through cause $k$ is
    $\Prob(\kappa = k \mid \tau > c) = \dfrac{-\alpha e^{Tc}T^{-1}D_k}{S(c)}$, and these probabilities sum to unity over $k=1,\ldots,n$.
    \item[(iii)] The conditional law of the residual pair is again multi-absorbing phase-type,
    \[
    (\tau-c,\;\kappa)\ \big|\ \{\tau>c\}\ \sim\ \text{MAPH}_{m,n}(\alpha_c,T,D),
    \qquad
    \alpha_c=\frac{\alpha e^{Tc}}{S(c)},
    \]
    with the same $T$ and $D$ and a restarted initial vector $\alpha_c$. In particular the conditional residual sub-density is $f(u,k\mid\tau>c)=\alpha_c e^{Tu}D_k$ for $u \ge 0$.
\end{itemize}
\end{proposition}
\begin{proof}
The proof is given in \Cref{app:proofs-cens}.
\end{proof}
Item (i) is the censored counterpart of the exact-event sub-density of item (ii) of \Cref{prop:maph-props}: an exact record contributes $\alpha e^{T\tau}D_\kappa$, a censored record contributes $\alpha e^{Tc}\mathbf{1}_m$. Item~(iii) is the Markov property of the underlying MJP read at the censoring time: conditioning on survival past $c$ leaves the dynamics unchanged and only re-initializes the phase distribution to $\alpha_c$, the conditional distribution of the phase occupied at time $c$. Items (ii) and (iii) are what the E-step below exploits when it completes a censored path through its eventual absorption.
\end{aarevision}

\begin{aarevision}
Equation~\eqref{eq:complete-log-like-jump} is the complete-data log-likelihood, evaluated on paths completed through eventual absorption, so its algebraic form is unchanged by censoring. What changes is the observed-data contribution: by \Cref{prop:maph-censored}(i) a record censored at $c$ contributes the survival probability $S(c)=\alpha e^{Tc}\mathbf 1_m$ rather than the exact-event density $\alpha e^{Tc}D_k$. Censoring therefore enters through the E-step expectations, not as a separate term in \eqref{eq:complete-log-like-jump}, and the M-step keeps its unchanged form.
\end{aarevision}

\begin{aarevision}
\paragraph{The E-step.} A censored contribution conditions on $\{\tau>y\}$ instead of $\{\tau=y,\kappa=k\}$, with the E-step completing the latent path through its eventual absorption. \end{aarevision}

Write
\begin{equation}
\label{eq:censored-definitions}
S(c)=\Prob(\tau>c)=\alpha e^{Tc}\mathbf 1_m .
\end{equation}
Define the row vectors
\begin{equation}
\label{eq:censored-ag}
a(c)=\alpha e^{Tc},
\qquad
g(c)=a(c)(-T)^{-1},
\end{equation}
and, for any non-negative column vector $b\in\mathbb R^m$, define the $m\times m$ matrix
\begin{equation}
\label{eq:censored-Cb}
C_{ij}(c;b)=\int_0^c
       \big(\alpha e^{Tu}\big)_i
       \big(e^{T(c-u)}b\big)_j\,du.
\end{equation}
The vector $g(c)$ records the unnormalised expected occupation times after $c$, while $C(c;b)$ accounts for the part of the path before $c$ with terminal weight $b$.

\begin{proposition}[Right-censored E-step]
\label{prop:e-step-censored}
For an $\text{MAPH}_{m,n}(\alpha,T,D)$ distribution and a record censored at $c>0$, the conditional expectations of the complete-path sufficient statistics, given $\{\tau>c\}$, are
\begin{equation}
\label{eq:censored-estats}
\begin{aligned}
\ex{B_{i} \mid \tau>c}
  &=\frac{\alpha_i\big(e^{Tc}\mathbf 1_m\big)_i}{S(c)},
&\qquad
\ex{Z_i \mid \tau>c}
  &=\frac{C_{ii}(c;\mathbf 1_m)+g_i(c)}{S(c)},\\[3pt]
\ex{M_{ij} \mid \tau>c}
  &=\frac{T_{ij}\left[C_{ij}(c;\mathbf 1_m)+g_i(c)\right]}{S(c)},\quad i\ne j,
&
\ex{E_{ik} \mid \tau>c}
  &=\frac{g_i(c)\,D_{ik}}{S(c)},
\end{aligned}
\end{equation}
with $\ex{M_{ii} \mid \tau>c}=0$, and the expected number of exits from $s_i$ is
\begin{equation}
\label{eq:censored-N}
\ex{N_i \mid \tau>c}=\sum_{j=1}^m\ex{M_{ij} \mid \tau>c}+\sum_{k=1}^n\ex{E_{ik} \mid \tau>c}.
\end{equation}
\end{proposition}
\begin{proof}
Given in \Cref{app:e-step-censored}.
\end{proof}

\begin{aarevision}
These replace \eqref{eq:estats-exact} by
\end{aarevision}
\begin{equation}
\label{eq:estats-cens}
\begin{aligned}
b_{i}^\ell&=\ex{B_i \mid \tau>y^\ell},
&z_i^\ell&=\ex{Z_i \mid \tau>y^\ell},\\[4pt]
n_{i}^\ell&=\ex{N_i \mid \tau>y^\ell},
&m_{ij}^\ell&=\ex{M_{ij} \mid \tau>y^\ell},\\[4pt]
e_{ik}^\ell&=\ex{E_{ik} \mid \tau>y^\ell}. &&
\end{aligned}
\end{equation}
\begin{aarevision}
Where an exact record carries an indicator that assigns its completed path to the observed cause, a censored record has no observed cause, and the current model instead spreads its contribution over every possible eventual cause in the proportions given by \Cref{prop:maph-censored}(ii).

\paragraph{The M-step.} The M-step \eqref{eq:approx-m-step} is unchanged. It is applied to the sums over the whole sample of the per-record contributions, taking \eqref{eq:estats-exact} for a record with $\delta^\ell=1$ and \eqref{eq:estats-cens} for one with $\delta^\ell=0$. \subchg{All completed paths enter $\widehat P^\lambda$ regardless of cause, and a censored subject contributes to it through its posterior distribution over the eventual cause rather than through the posterior probability of one distinguished cause.}

\paragraph{The algorithm.} The M-step and the conversion are unchanged, and the initialization differs only in that it uses the censoring-compatible simplified initialization of \Cref{sec:init-simplified} as soon as any record is censored. The E-step acquires a second branch: if $\delta^\ell=1$, compute the exact-event expectations of \Cref{prop:e-step} and form \eqref{eq:estats-exact}; if $\delta^\ell=0$, compute the cause-resolved right-censored expectations of \Cref{prop:e-step-censored} and form \eqref{eq:estats-cens}.
\end{aarevision}

\begin{subrevision}
\subsection{Properties of the iteration}
\label{sec:em-properties}

Two properties of the algorithm follow from \Cref{prop:approx-mle} and are worth stating, since they
are what distinguishes the present iteration from one whose M-step maximizes a surrogate.

First, no feasibility step is required. The M-step returns the exact maximizer of the expected
complete-data log-likelihood over $\Theta$, and that maximizer lies in $\Theta$ by construction: the
updated rows of $[\,\widehat P^\lambda \mid \widehat P^\mu\,]$ are probability vectors by the flow
balance \eqref{eq:flow-balance}, and $I-\widehat P^\lambda$ is non-singular, so the conversion of
the conversion always returns a valid $\text{MAPH}_{m,n}$ distribution. The absorption matrix
\eqref{eq:R-from-mle} of that distribution satisfies the bound \eqref{eq:R-bound} automatically.

Second, the iteration is a genuine EM iteration, so by \Cref{cor:monotone} the observed-data
log-likelihood is non-decreasing along it,
$\ell_{\mathrm{obs}}(\theta^{(\nu+1)}) \ge \ell_{\mathrm{obs}}(\theta^{(\nu)})$. A single-threshold
stopping rule on the change in $\ell_{\mathrm{obs}}$ is therefore sound: the sequence
$\{\ell_{\mathrm{obs}}(\theta^{(\nu)})\}$ is monotone and, being bounded above on any sample with
$L<\infty$, converges. We stop when its increment falls below $\epsilon$, or when a maximum number
of iterations is reached. Monotonicity of the likelihood does not by itself imply convergence of the
iterates $\theta^{(\nu)}$, on which we comment in \Cref{sec:conclusion}.
\end{subrevision}

%% file: sections/implementation.tex
\section{Implementation and application examples}
\label{sec:impl}

\subsection{Verification via Synthetic Data}
\label{sec:verification}

\begin{rcrevision}
We verify the exact-event implementation by fitting it to uncensored data simulated from known MAPH laws and checking that the recovered distribution matches the data-generating one. These particular experiments exercise only the exact-event path; the censored E-step is verified separately, \aachg{in the companion package's own test suite --- against direct numerical integration of the exact-event E-step over the tail $\{t>c\}$, against the structural identities of \Cref{prop:e-step-censored}, and against their $c\to0$ limit, in which they must reduce to the unconditional complete-path expectations --- and} end to end on the censored ICU analysis of \Cref{sec:realdata}. Because MAPH distributions are not identifiable (\Cref{sec:maph-dist}), the generator parameters $(\alpha,T,D)$ cannot themselves be the object of comparison; we therefore compare the fitted and true laws through parameterization-invariant summaries.
\end{rcrevision}

As ground truth we use the running $\text{MAPH}_{4,3}$ example of \Cref{ex:running} -- the sparse law of \Cref{fig:example-maph} with rates \eqref{eq:verif-truth}, whose three competing causes we read as a completely healthy discharge ($k=1$) and two discharge-requiring-further-care types ($k=2,3$). Its causes are well balanced ($\Prob(\kappa=k)=0.384,\,0.282,\,0.335$) and its conditional squared coefficients of variation straddle unity ($1.18$, $0.96$, $0.87$), so that the moment-based initialization classifies the causes into both its hyper-exponential and its hypo-exponential block families; the worked example of \Cref{sec:init-example} traces this construction in detail. For each sample size $L\in\{200,500,1000,2000,5000\}$ we draw $L$ i.i.d.\ pairs $(\tau,\kappa)$, fit a fully connected $\text{MAPH}_{4,3}$ from the moment-based initialization with the phase order $m=4$ assumed known, and average the results over $12$ independent replications. The fitted generator is not expected to reproduce the sparsity pattern of \eqref{eq:verif-truth} -- by non-identifiability it need not -- and the comparison is through the parameterization-free summaries.

\Cref{tab:verif} reports four parameterization-free recovery errors between the fitted law $\hat d$ and the truth $d$: the largest absolute error in the marginal absorption probabilities, $\max_k|\Prob_{\hat d}(\kappa=k)-\Prob_{d}(\kappa=k)|$; the largest relative error in the per-cause conditional mean $\ex{\tau\mid\kappa=k}$ and in the conditional squared coefficient of variation (SCV); and the integrated absolute error between the sub-distribution functions, $\sum_{k}\frac1{|G|}\sum_{u\in G}|F_{\hat d}(u,k)-F_{d}(u,k)|$ over a grid $G$ on $(0,6]$. The final column reports the per-observation log-likelihood gap to the truth, $\frac1L\big(\ell(\hat d)-\ell(d)\big)$, the identifiability-free check that the fitted model explains the data as well as the generating one. All recovery errors decrease with $L$ (up to Monte-Carlo fluctuation at $L=2000$); the log-likelihood gap is positive at small $L$ -- $1.3\times10^{-2}$ nats at $L=200$, where the fitted model adapts to sampling fluctuations -- and shrinks to within about $2\times10^{-3}$ nats of the truth by $L=5000$. \aachg{The conditional SCVs are the exception, and \Cref{fig:verif-scaling} extends the same experiment to $L=20000$ to see whether they eventually follow. They do not: the absorption probabilities, conditional means and sub-distribution IAE all decay at close to the parametric $L^{-1/2}$ rate, whereas the conditional SCV error is flat at about $0.17$ throughout. The diagnosis is in the last column of \Cref{tab:verif}: the log-likelihood gap turns negative for $L\ge2000$, so at these sample sizes the iteration stops short of the maximizer rather than at it, and it is the second-moment functionals --- the most tail-sensitive of the four summaries --- that pay for it first. The distribution functions themselves are unaffected: as \Cref{fig:verif-cdf} shows for a single $L=5000$ replication, they are recovered nearly exactly for all three causes.}

\begin{table}[ht]
\centering
\begin{tabular}{r r r r r r}
\toprule
$L$ & $\pi$-error & cond.\ mean (rel.) & cond.\ SCV (rel.) & sub-cdf IAE & log-lik.\ gap / obs. \\
\midrule
 200 & 0.034 & 0.192 & 0.240 & 0.067 & $+1.7\times10^{-2}$ \\
 500 & 0.023 & 0.122 & 0.103 & 0.044 & $+7.0\times10^{-3}$ \\
1000 & 0.013 & 0.084 & 0.103 & 0.024 & $+3.3\times10^{-3}$ \\
2000 & 0.016 & 0.085 & 0.086 & 0.028 & $+5.6\times10^{-4}$ \\
5000 & 0.009 & 0.073 & 0.070 & 0.017 & $-5.6\times10^{-4}$ \\
\bottomrule
\end{tabular}
\caption{Recovery of the sparse structured $\text{MAPH}_{4,3}$ ground truth \eqref{eq:verif-truth} of \Cref{fig:example-maph} as the sample size $L$ grows, averaged over $12$ replications, \subchg{with the moment-based initialization capped at $5000$ iterations. All errors are parameterization-free distributional summaries, and all of them decrease with $L$.}}
\label{tab:verif}
\end{table}

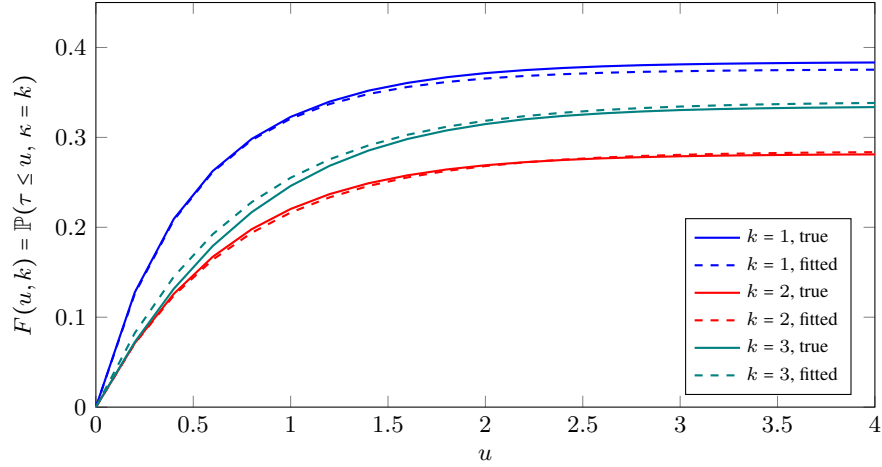
\begin{figure}[ht]
\centering
\begin{tikzpicture}
\begin{axis}[width=0.72\textwidth, height=0.42\textwidth,
  xlabel={$u$}, ylabel={$F(u,k)=\Prob(\tau\le u,\,\kappa=k)$},
  xmin=0, xmax=4, ymin=0, ymax=0.45, legend pos=south east,
  legend cell align=left, legend style={font=\scriptsize}, tick label style={font=\small}, label style={font=\small}]
\addplot[blue, thick] coordinates {(0.00,0.0000)(0.20,0.1278)(0.40,0.2093)(0.60,0.2626)(0.80,0.2983)(1.00,0.3228)(1.20,0.3398)(1.40,0.3520)(1.60,0.3606)(1.80,0.3669)(2.00,0.3715)(2.20,0.3748)(2.40,0.3772)(2.60,0.3790)(2.80,0.3803)(3.00,0.3812)(3.20,0.3819)(3.40,0.3824)(3.60,0.3828)(3.80,0.3831)(4.00,0.3833)};
\addlegendentry{$k=1$, true}
\addplot[blue, dashed, thick] coordinates {(0.00,0.0000)(0.20,0.1263)(0.40,0.2080)(0.60,0.2615)(0.80,0.2970)(1.00,0.3209)(1.20,0.3371)(1.40,0.3483)(1.60,0.3561)(1.80,0.3616)(2.00,0.3655)(2.20,0.3683)(2.40,0.3703)(2.60,0.3718)(2.80,0.3729)(3.00,0.3736)(3.20,0.3742)(3.40,0.3746)(3.60,0.3749)(3.80,0.3751)(4.00,0.3753)};
\addlegendentry{$k=1$, fitted}
\addplot[red, thick] coordinates {(0.00,0.0000)(0.20,0.0712)(0.40,0.1260)(0.60,0.1672)(0.80,0.1979)(1.00,0.2204)(1.20,0.2369)(1.40,0.2490)(1.60,0.2578)(1.80,0.2643)(2.00,0.2690)(2.20,0.2724)(2.40,0.2749)(2.60,0.2767)(2.80,0.2780)(3.00,0.2790)(3.20,0.2797)(3.40,0.2802)(3.60,0.2806)(3.80,0.2809)(4.00,0.2811)};
\addlegendentry{$k=2$, true}
\addplot[red, dashed, thick] coordinates {(0.00,0.0000)(0.20,0.0709)(0.40,0.1241)(0.60,0.1640)(0.80,0.1939)(1.00,0.2164)(1.20,0.2333)(1.40,0.2460)(1.60,0.2555)(1.80,0.2627)(2.00,0.2681)(2.20,0.2722)(2.40,0.2753)(2.60,0.2776)(2.80,0.2793)(3.00,0.2806)(3.20,0.2816)(3.40,0.2823)(3.60,0.2829)(3.80,0.2833)(4.00,0.2836)};
\addlegendentry{$k=2$, fitted}
\addplot[teal, thick] coordinates {(0.00,0.0000)(0.20,0.0725)(0.40,0.1316)(0.60,0.1792)(0.80,0.2168)(1.00,0.2460)(1.20,0.2684)(1.40,0.2854)(1.60,0.2982)(1.80,0.3077)(2.00,0.3148)(2.20,0.3201)(2.40,0.3239)(2.60,0.3267)(2.80,0.3288)(3.00,0.3304)(3.20,0.3315)(3.40,0.3323)(3.60,0.3329)(3.80,0.3333)(4.00,0.3337)};
\addlegendentry{$k=3$, true}
\addplot[teal, dashed, thick] coordinates {(0.00,0.0000)(0.20,0.0825)(0.40,0.1450)(0.60,0.1923)(0.80,0.2280)(1.00,0.2551)(1.20,0.2756)(1.40,0.2911)(1.60,0.3029)(1.80,0.3118)(2.00,0.3185)(2.20,0.3236)(2.40,0.3275)(2.60,0.3304)(2.80,0.3326)(3.00,0.3343)(3.20,0.3356)(3.40,0.3366)(3.60,0.3373)(3.80,0.3379)(4.00,0.3383)};
\addlegendentry{$k=3$, fitted}
\end{axis}
\end{tikzpicture}
\caption{True (solid) versus fitted (dashed) sub-distribution functions $F(u,k)$ for the three causes of the sparse $\text{MAPH}_{4,3}$ of \eqref{eq:verif-truth}, on a single sample of size $L=5000$. The fitted curves track the truth closely and their right limits recover the marginal absorption probabilities $\Prob(\kappa=k)$.}
\label{fig:verif-cdf}
\end{figure}

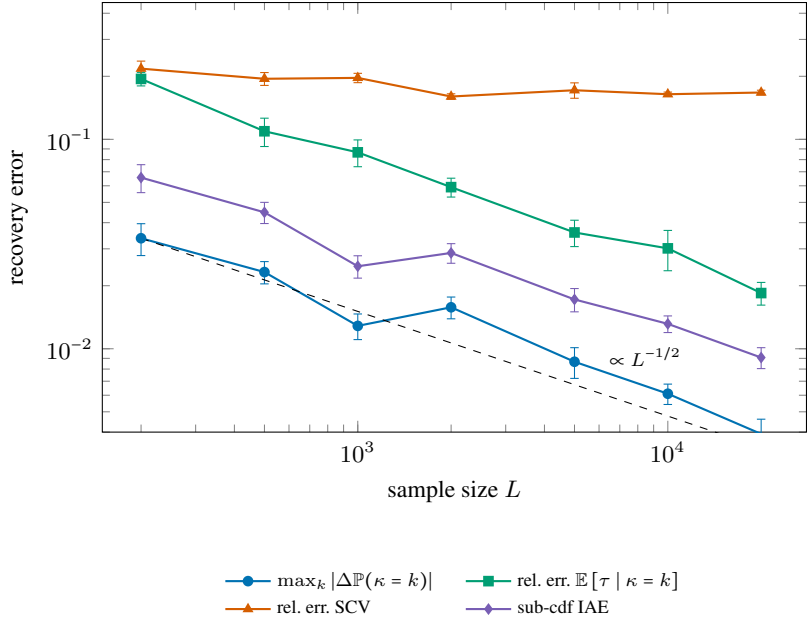
\begin{figure}[ht]
\centering
\begin{aarevision}
\definecolor{ourcol}{RGB}{0,114,178}
\definecolor{cmcol}{RGB}{0,158,115}
\definecolor{scvcol}{RGB}{213,94,0}
\definecolor{iaecol}{RGB}{120,94,180}
\begin{tikzpicture}
\begin{loglogaxis}[width=0.66\textwidth, height=0.44\textwidth,
  xlabel={sample size $L$}, ylabel={recovery error},
  xmin=150, xmax=28000, ymin=0.004, ymax=0.45,
  legend cell align=left, legend columns=2,
  legend style={at={(0.5,-0.30)}, anchor=north, draw=none, font=\scriptsize,
                /tikz/every even column/.append style={column sep=10pt}},
  tick label style={font=\small}, label style={font=\small}]
\addplot[ourcol, thick, mark=*, mark size=1.6pt, error bars/.cd, y dir=both, y explicit] coordinates {(200,0.03371) +- (0,0.00585) (500,0.02324) +- (0,0.00283) (1000,0.01286) +- (0,0.00180) (2000,0.01578) +- (0,0.00188) (5000,0.00867) +- (0,0.00145) (10000,0.00610) +- (0,0.00068) (20000,0.00392) +- (0,0.00069)};
\addlegendentry{$\max_k|\Delta\Prob(\kappa=k)|$}
\addplot[cmcol, thick, mark=square*, mark size=1.6pt, error bars/.cd, y dir=both, y explicit] coordinates {(200,0.19433) +- (0,0.01441) (500,0.10929) +- (0,0.01689) (1000,0.08668) +- (0,0.01268) (2000,0.05912) +- (0,0.00611) (5000,0.03592) +- (0,0.00517) (10000,0.03016) +- (0,0.00658) (20000,0.01845) +- (0,0.00229)};
\addlegendentry{rel.\ err.\ $\ex{\tau\mid\kappa=k}$}
\addplot[scvcol, thick, mark=triangle*, mark size=1.6pt, error bars/.cd, y dir=both, y explicit] coordinates {(200,0.21754) +- (0,0.01878) (500,0.19464) +- (0,0.01371) (1000,0.19654) +- (0,0.01008) (2000,0.16012) +- (0,0.00485) (5000,0.17157) +- (0,0.01449) (10000,0.16422) +- (0,0.00391) (20000,0.16731) +- (0,0.00433)};
\addlegendentry{rel.\ err.\ SCV}
\addplot[iaecol, thick, mark=diamond*, mark size=1.6pt, error bars/.cd, y dir=both, y explicit] coordinates {(200,0.06567) +- (0,0.01002) (500,0.04482) +- (0,0.00519) (1000,0.02477) +- (0,0.00303) (2000,0.02868) +- (0,0.00308) (5000,0.01720) +- (0,0.00219) (10000,0.01315) +- (0,0.00120) (20000,0.00908) +- (0,0.00104)};
\addlegendentry{sub-cdf IAE}
\addplot[black, dashed, thin, forget plot] coordinates {(200,0.03371) (500,0.02132) (1000,0.01508) (2000,0.01066) (5000,0.00674) (10000,0.00477) (20000,0.00337)};
\node[font=\scriptsize, anchor=south west, black] at (axis cs:6000,0.0075) {$\propto L^{-1/2}$};
\end{loglogaxis}
\end{tikzpicture}
\caption{Recovery error against sample size for the running $\text{MAPH}_{4,3}$ of \eqref{eq:verif-truth}, fitted from the moment-based initialization with $m=4$ known; $12$ replications per $L$, error bars one standard error across replications, both axes logarithmic. The dashed guide is the parametric $L^{-1/2}$ rate. Three of the four errors follow it, with fitted log-log slopes $-0.45$ for the absorption probabilities, $-0.49$ for the conditional means and $-0.41$ for the sub-distribution IAE. The conditional SCV does not: its slope is $-0.06$ and its relative error sits at about $0.17$ across a hundredfold increase in $L$. That plateau should not be read as an asymptotic property of the estimator. Over the same range the per-observation log-likelihood gap $\frac1L(\ell(\hat d)-\ell(d))$ turns negative, reaching $-3.3\times10^{-3}$ at $L=20000$, so the iteration is finishing \emph{below} the likelihood of the data-generating parameter on its own sample --- something a converged maximizer cannot do. At these sample sizes the relaxed M-step of \Cref{sec:em-alg} and the iteration cap, not the sampling error, are what limits the second-moment functionals; the first-order summaries are already at the parametric rate.}
\label{fig:verif-scaling}
\end{aarevision}
\end{figure}

\paragraph{Effect of the phase order.} The fit above fixes the order at the true value $m=4$. To see how the order affects both the fit and the cost of the M-step, we refit the same $L=5000$ sample for $m=2,\ldots,7$ from the moment-based initialization, each run capped at $500$ iterations. \Cref{tab:msweep} summarizes the runs. The fit is already excellent at the smallest orders -- the per-observation log-likelihood sits within $10^{-3}$ nats of the truth and the sub-distribution IAE is about $0.014$ throughout -- and barely changes as $m$ grows, reflecting that this ground-truth law is of low effective order: a small MAPH already captures its smooth sub-distribution functions, so extra phases buy almost no fit. \subchg{Two further points stand out. First, every fitted model attains a \emph{higher} per-observation log-likelihood than the ground truth on the sample it was fitted to, at every order; this is what a maximum-likelihood estimator must do, and it is a check that the M-step of \Cref{prop:approx-mle} is exact. Second, what grows steeply with the order is not the fit but the number of iterations: reaching a tolerance of $10^{-7}$ on the log-likelihood increment takes about $1300$ iterations at $m=2$ and about $1700$ at $m=3,4$, while the runs at $m=5,6,7$ had not reached it after $5000$. The larger and flatter the parameterization, the slower the ascent, so the computational character of the method is governed by the phase budget far more than by the fit, which saturates early. The observed-data log-likelihood increased monotonically in every one of these runs, as \Cref{cor:monotone} requires.}

\begin{table}[ht]
\centering
\begin{tabular}{r r r r r r}
\toprule
$m$ & \#\,params & per-obs.\ $\ell$ & $(\ell-\ell_{\text{true}})/L$ & sub-cdf IAE & iterations \\
\midrule
2 & 9  & $-1.6318$ & $+8.7\times10^{-4}$ & 0.0140 & 1343 \\
3 & 17 & $-1.6320$ & $+6.7\times10^{-4}$ & 0.0153 & 1741 \\
4 & 27 & $-1.6320$ & $+6.7\times10^{-4}$ & 0.0153 & 1724 \\
5 & 39 & $-1.6318$ & $+9.0\times10^{-4}$ & 0.0142 & $>5000$ \\
6 & 53 & $-1.6318$ & $+9.0\times10^{-4}$ & 0.0142 & $>5000$ \\
7 & 69 & $-1.6316$ & $+1.1\times10^{-3}$ & 0.0137 & $>5000$ \\
\bottomrule
\end{tabular}
\caption{\subchg{Refitting the running $\text{MAPH}_{4,3}$ ground truth at orders $m=2,\ldots,7$ on a single $L=5000$ sample, from the moment-based initialization, each run capped at $5000$ iterations. The fit summaries -- per-observation log-likelihood $\ell/L$, its gap to the truth ($\ell_{\text{true}}/L=-1.6327$), and the sub-distribution-function IAE against the truth -- are near-optimal already at the smallest orders and change little with $m$, while the iterations needed to reach a tolerance of $10^{-7}$ on the log-likelihood increment grow with the order; the runs at $m\ge5$ had not reached it at the cap. The true order is $m=4$.}}
\label{tab:msweep}
\end{table}

\subchg{\textbf{[TO BE RE-RUN.]} A remark on the iteration behaviour stood here. It reported that the observed-data log-likelihood was not monotone, that the iteration rode the feasibility boundary -- the projection being active on about $0.7$ of the $m=4$ rows per iteration -- and that a robust stopping rule was therefore needed in place of a single threshold. None of this applies to the algorithm of \Cref{sec:em-alg} as it now stands: the M-step is exact, no projection is performed, and by \Cref{cor:monotone} the observed-data log-likelihood is non-decreasing, so a single-threshold rule is sound (\Cref{sec:em-properties}). The convergence behaviour of the present iteration -- iterations to a given tolerance, and its dependence on the initialization -- has to be measured afresh before this paragraph can be rewritten.}

\begin{rcrevision}
\subsection{Application: Intensive-Care Length of Stay}
\label{sec:realdata}
\end{rcrevision}

\begin{rcrevision}
We apply the method to the intensive-care length-of-stay data of \citet{beyersmann_2012}. It is the natural test case for the extension of this paper --- it contains right-censored records, which can now be used rather than discarded --- and it has been modelled before with a phase-type competing-risks model \citep{lindqvist_phase-type_2022}, so it also permits a direct comparison. \aachg{Goodness of fit is measured against the nonparametric Aalen--Johansen estimator \citep{aalen_johansen_1978}, which handles the censoring through its risk sets. With only $0.9\%$ of records censored it is within $0.008$ of the raw empirical sub-distribution $L^{-1}\#\{\ell:y^\ell\le u,\ k^\ell=k\}$, which is what \Cref{fig:icu-cif} plots against.} For simulated data, where the data-generating law is known, we compare against it directly instead; see \Cref{fig:verif-cdf}.

The dataset contains $650$ intensive-care patients without pneumonia at admission \citep[Ch.~1]{beyersmann_2012}: $589$ alive discharges, $55$ hospital deaths and $6$ stays right censored at $8$, $8$, $8$, $14$, $20$ and $31$ days. Earlier versions of this analysis set the six censored records aside and fitted the remaining $L=644$ pairs. With the censored E-step of \Cref{prop:e-step-censored} in place we fit all $650$ records, each censored stay entering the likelihood through its survival term $\alpha e^{Ty^\ell}\mathbf 1_m$. We work throughout on the original day scale rather than on the mean-normalized scale used previously, so that the reported log-likelihoods may be compared directly with a model fitted outside our software; rescaling time shifts every log-likelihood by the same constant $L\log s$ and leaves all AIC \aachg{{\em differences}} unchanged.

\Cref{tab:icu} reports the order selection, alongside the same exercise restricted to the Coxian competing-risks model of \citet{lindqvist_phase-type_2022} --- that is, $\alpha=e_1$ with an upper-bidiagonal $T$ and a full $m\times n$ matrix $D$. Since the EM update preserves the zeros of $\alpha$ and of the off-diagonal of $T$ (\Cref{sec:em-alg}), starting from a Coxian-patterned generator fits exactly that submodel, and the two columns of the table are produced by the same code on the same records. \subchg{Because the likelihood surface has several local optima, every entry is the best of a fixed panel of starts:} the two initializations of \Cref{sec:init-heuristic}, slightly densified Coxian fits of the same order, and twenty random dense generators. A Coxian fit is itself a MAPH of the same order, so it is included in the panel for the general column as well.

\subchg{The general model attains its best AIC at $m=4$ and its best BIC at $m=3$; the Coxian submodel, which spends far fewer parameters for a nearly identical fit, is best at $m=5$, though its AIC at $m=4$ and $m=5$ differs by only half a unit, and it stays below the general model's AIC at every order. The extra freedom of a full $T$ buys $0.2$ nats of log-likelihood at $m=3$, and $0.8$ at $m=4$, at a cost of six and twelve parameters respectively --- a fair summary of what the unrestricted parameterization is worth on these data.} Across orders the fitted cause probabilities sit at $(0.914,0.086)$, matching the Aalen--Johansen limits $(0.909,0.084)$ on $[0,80]$ days.
\end{rcrevision}

\begin{table}[ht]
\centering
\begin{rcrevision}
\begin{tabular}{r r r r @{\hspace{2.2em}} r r r}
\toprule
& \multicolumn{3}{c}{general $\text{MAPH}_{m,2}$} & \multicolumn{3}{c}{Coxian $\text{MAPH}_{m,2}$} \\
\cmidrule(r){2-4}\cmidrule(l){5-7}
$m$ & log-lik. & \#\,par & AIC & log-lik. & \#\,par & AIC \\
\midrule
1 & $-2450.3$ & 2  & $4904.5$ & $-2450.3$ & 2  & $4904.5$ \\
2 & $-2404.6$ & 7  & $4823.1$ & $-2404.6$ & 5  & $4819.1$ \\
3 & $-2374.8$ & 14 & $4777.7$ & $-2375.0$ & 8  & $4766.0$ \\
4 & $-2359.9$ & 23 & $\mathbf{4765.7}$ & $-2360.6$ & 11 & $4743.2$ \\
5 & $-2355.7$ & 34 & $4779.4$ & $-2357.4$ & 14 & $\mathbf{4742.7}$ \\
\bottomrule
\end{tabular}
\caption{Order selection for the intensive-care data, fitted to all $L=650$ records ($n=2$ causes) with the right-censored EM of \Cref{sec:em-alg}, on the day scale. The left block is the general $\text{MAPH}_{m,2}$; the right block restricts the same fitting code to the Coxian competing-risks submodel of \citet{lindqvist_phase-type_2022}. Free parameters are $(m-1)+m(m-1)+mn$ and $(m-1)+mn$ respectively, the diagonal of $T$ being fixed by the zero-row-sum constraint. Each entry is the best of a fixed panel of starting points (see the text). \subchg{The panel now matters much less than it did for the earlier relaxed-and-projected M-step: running the general model once from its default initialization reaches $-2404.56$ at $m=2$, matching the panel, $-2374.96$ at $m=3$, short of it by $0.12$ nats, and $-2373.31$ at $m=4$, short by $13.5$.}}
\label{tab:icu}
\end{rcrevision}
\end{table}

\begin{figure}[!tb]
\centering
\begin{rcrevision}
\definecolor{empcol}{RGB}{60,60,60}
\definecolor{ourcol}{RGB}{0,114,178}
\definecolor{lindcol}{RGB}{213,94,0}
\begin{minipage}[t]{0.32\textwidth}\centering
\begin{tikzpicture}
\begin{axis}[width=\textwidth, height=0.92\textwidth,
  title={sub-distribution functions $F(u,k)$}, title style={font=\small},
  xlabel={length of stay $u$ (days)}, ylabel={$F(u,k)$},
  xmin=0, xmax=80, ymin=0, ymax=1,
  tick label style={font=\scriptsize}, label style={font=\small}]
\addplot[empcol, thick, const plot] coordinates {(0,0.0000)(2,0.0831)(4,0.2631)(6,0.4123)(8,0.5077)(10,0.5877)(12,0.6400)(14,0.6862)(16,0.7262)(18,0.7615)(20,0.7800)(22,0.7908)(24,0.8092)(26,0.8169)(28,0.8292)(30,0.8431)(32,0.8508)(34,0.8600)(36,0.8615)(38,0.8662)(40,0.8677)(42,0.8708)(44,0.8738)(46,0.8785)(48,0.8815)(50,0.8831)(52,0.8846)(54,0.8877)(56,0.8938)(58,0.8969)(60,0.8985)(62,0.8985)(64,0.9000)(66,0.9000)(68,0.9000)(70,0.9015)(72,0.9015)(74,0.9015)(76,0.9015)(78,0.9015)(80,0.9015)};
\addplot[ourcol, dashed, thick] coordinates {(0,0.0000)(2,0.0719)(4,0.2078)(6,0.3450)(8,0.4623)(10,0.5556)(12,0.6275)(14,0.6820)(16,0.7234)(18,0.7548)(20,0.7791)(22,0.7982)(24,0.8133)(26,0.8257)(28,0.8359)(30,0.8445)(32,0.8519)(34,0.8582)(36,0.8638)(38,0.8687)(40,0.8731)(42,0.8770)(44,0.8805)(46,0.8837)(48,0.8866)(50,0.8891)(52,0.8915)(54,0.8936)(56,0.8955)(58,0.8972)(60,0.8988)(62,0.9002)(64,0.9015)(66,0.9027)(68,0.9038)(70,0.9047)(72,0.9056)(74,0.9064)(76,0.9072)(78,0.9078)(80,0.9084)};
\addplot[lindcol, dotted, very thick] coordinates {(0,0.0000)(2,0.0450)(4,0.1588)(6,0.2794)(8,0.3861)(10,0.4762)(12,0.5512)(14,0.6134)(16,0.6649)(18,0.7075)(20,0.7428)(22,0.7720)(24,0.7962)(26,0.8162)(28,0.8327)(30,0.8464)(32,0.8578)(34,0.8672)(36,0.8749)(38,0.8813)(40,0.8867)(42,0.8911)(44,0.8947)(46,0.8977)(48,0.9002)(50,0.9023)(52,0.9040)(54,0.9054)(56,0.9066)(58,0.9075)(60,0.9083)(62,0.9090)(64,0.9096)(66,0.9100)(68,0.9104)(70,0.9107)(72,0.9110)(74,0.9112)(76,0.9114)(78,0.9115)(80,0.9116)};
\addplot[empcol, thick, const plot] coordinates {(0,0.0000)(2,0.0015)(4,0.0092)(6,0.0215)(8,0.0338)(10,0.0462)(12,0.0538)(14,0.0538)(16,0.0569)(18,0.0600)(20,0.0631)(22,0.0646)(24,0.0692)(26,0.0738)(28,0.0754)(30,0.0754)(32,0.0754)(34,0.0769)(36,0.0769)(38,0.0769)(40,0.0769)(42,0.0769)(44,0.0769)(46,0.0769)(48,0.0769)(50,0.0785)(52,0.0785)(54,0.0785)(56,0.0800)(58,0.0800)(60,0.0800)(62,0.0815)(64,0.0815)(66,0.0815)(68,0.0815)(70,0.0815)(72,0.0815)(74,0.0815)(76,0.0815)(78,0.0831)(80,0.0831)};
\addplot[ourcol, dashed, thick] coordinates {(0,0.0000)(2,0.0079)(4,0.0190)(6,0.0297)(8,0.0389)(10,0.0464)(12,0.0524)(14,0.0572)(16,0.0610)(18,0.0641)(20,0.0667)(22,0.0688)(24,0.0706)(26,0.0722)(28,0.0735)(30,0.0747)(32,0.0758)(34,0.0767)(36,0.0776)(38,0.0783)(40,0.0790)(42,0.0796)(44,0.0802)(46,0.0807)(48,0.0811)(50,0.0815)(52,0.0819)(54,0.0823)(56,0.0826)(58,0.0828)(60,0.0831)(62,0.0833)(64,0.0835)(66,0.0837)(68,0.0839)(70,0.0841)(72,0.0842)(74,0.0843)(76,0.0844)(78,0.0846)(80,0.0846)};
\addplot[lindcol, dotted, very thick] coordinates {(0,0.0000)(2,0.0042)(4,0.0151)(6,0.0267)(8,0.0370)(10,0.0457)(12,0.0530)(14,0.0590)(16,0.0639)(18,0.0680)(20,0.0715)(22,0.0743)(24,0.0766)(26,0.0785)(28,0.0801)(30,0.0815)(32,0.0825)(34,0.0835)(36,0.0842)(38,0.0848)(40,0.0853)(42,0.0858)(44,0.0861)(46,0.0864)(48,0.0866)(50,0.0868)(52,0.0870)(54,0.0871)(56,0.0873)(58,0.0874)(60,0.0874)(62,0.0875)(64,0.0875)(66,0.0876)(68,0.0876)(70,0.0877)(72,0.0877)(74,0.0877)(76,0.0877)(78,0.0877)(80,0.0877)};
\node[font=\scriptsize, anchor=west] at (axis cs:40,0.82) {alive discharge};
\node[font=\scriptsize, anchor=west] at (axis cs:40,0.17) {hospital death};
\end{axis}
\end{tikzpicture}
\end{minipage}\hfill
\begin{minipage}[t]{0.32\textwidth}\centering
\begin{tikzpicture}
\begin{axis}[width=\textwidth, height=0.92\textwidth,
  title={sub-density $f(u,1)$, alive discharge}, title style={font=\small},
  xlabel={length of stay $u$ (days)}, ylabel={$f(u,1)$},
  xmin=0, xmax=80, ymin=0, ymax=0.095,
  tick label style={font=\scriptsize}, label style={font=\small},
  legend style={at={(0.5,-0.28)}, anchor=north, draw=none, font=\scriptsize,
                legend columns=3, /tikz/every even column/.append style={column sep=8pt}},
  legend cell align=left,
  ]
\addplot[empcol, fill=empcol!18, draw=empcol!55, ybar interval, forget plot] coordinates {(1.0,0.00000)(3.0,0.08538)(5.0,0.08615)(7.0,0.06154)(9.0,0.04462)(11.0,0.02923)(13.0,0.02462)(15.0,0.02462)(17.0,0.01769)(19.0,0.01154)(21.0,0.00538)(23.0,0.00692)(25.0,0.00846)(27.0,0.00538)(29.0,0.00769)(31.0,0.00538)(33.0,0.00385)(35.0,0.00231)(37.0,0.00077)(39.0,0.00231)(41.0,0.00077)(43.0,0.00154)(45.0,0.00231)(47.0,0.00154)(49.0,0.00077)(51.0,0.00154)(53.0,0.00154)(55.0,0.00231)(57.0,0.00154)(59.0,0.00154)(61.0,0.00000)(63.0,0.00077)(65.0,0.00000)(67.0,0.00000)(69.0,0.00000)(71.0,0.00077)(73.0,0.00000)(75.0,0.00000)(77.0,0.00000)(79.0,0.00000)};
\addplot[ourcol, dashed, thick] coordinates {(0.0,0.00000)(0.5,0.02227)(1.0,0.03897)(1.5,0.05120)(2.0,0.05983)(2.5,0.06560)(3.0,0.06911)(3.5,0.07084)(4.0,0.07119)(4.5,0.07049)(5.0,0.06899)(5.5,0.06691)(6.0,0.06442)(6.5,0.06165)(7.0,0.05870)(7.5,0.05567)(8.0,0.05262)(8.5,0.04959)(9.0,0.04662)(9.5,0.04375)(10.0,0.04099)(10.5,0.03835)(11.0,0.03584)(11.5,0.03347)(12.0,0.03124)(12.5,0.02915)(13.0,0.02719)(13.5,0.02536)(14.0,0.02365)(14.5,0.02206)(15.0,0.02059)(15.5,0.01922)(16.0,0.01795)(16.5,0.01678)(17.0,0.01569)(17.5,0.01469)(18.0,0.01376)(18.5,0.01290)(19.0,0.01210)(19.5,0.01137)(20.0,0.01069)(20.5,0.01007)(21.0,0.00949)(21.5,0.00895)(22.0,0.00845)(22.5,0.00800)(23.0,0.00757)(23.5,0.00718)(24.0,0.00681)(24.5,0.00647)(25.0,0.00616)(25.5,0.00586)(26.0,0.00559)(26.5,0.00534)(27.0,0.00510)(27.5,0.00488)(28.0,0.00467)(28.5,0.00448)(29.0,0.00430)(29.5,0.00413)(30.0,0.00397)(30.5,0.00382)(31.0,0.00367)(31.5,0.00354)(32.0,0.00341)(32.5,0.00330)(33.0,0.00318)(33.5,0.00308)(34.0,0.00297)(34.5,0.00288)(35.0,0.00279)(35.5,0.00270)(36.0,0.00262)(36.5,0.00254)(37.0,0.00246)(37.5,0.00239)(38.0,0.00232)(38.5,0.00225)(39.0,0.00219)(39.5,0.00213)(40.0,0.00207)(40.5,0.00201)(41.0,0.00195)(41.5,0.00190)(42.0,0.00185)(42.5,0.00180)(43.0,0.00175)(43.5,0.00171)(44.0,0.00166)(44.5,0.00162)(45.0,0.00158)(45.5,0.00154)(46.0,0.00150)(46.5,0.00146)(47.0,0.00142)(47.5,0.00139)(48.0,0.00135)(48.5,0.00132)(49.0,0.00129)(49.5,0.00125)(50.0,0.00122)(50.5,0.00119)(51.0,0.00116)(51.5,0.00114)(52.0,0.00111)(52.5,0.00108)(53.0,0.00105)(53.5,0.00103)(54.0,0.00100)(54.5,0.00098)(55.0,0.00096)(55.5,0.00093)(56.0,0.00091)(56.5,0.00089)(57.0,0.00087)(57.5,0.00085)(58.0,0.00083)(58.5,0.00081)(59.0,0.00079)(59.5,0.00077)(60.0,0.00075)(60.5,0.00073)(61.0,0.00071)(61.5,0.00070)(62.0,0.00068)(62.5,0.00066)(63.0,0.00065)(63.5,0.00063)(64.0,0.00062)(64.5,0.00060)(65.0,0.00059)(65.5,0.00057)(66.0,0.00056)(66.5,0.00055)(67.0,0.00053)(67.5,0.00052)(68.0,0.00051)(68.5,0.00050)(69.0,0.00049)(69.5,0.00047)(70.0,0.00046)(70.5,0.00045)(71.0,0.00044)(71.5,0.00043)(72.0,0.00042)(72.5,0.00041)(73.0,0.00040)(73.5,0.00039)(74.0,0.00038)(74.5,0.00037)(75.0,0.00036)(75.5,0.00036)(76.0,0.00035)(76.5,0.00034)(77.0,0.00033)(77.5,0.00032)(78.0,0.00032)(78.5,0.00031)(79.0,0.00030)(79.5,0.00029)(80.0,0.00029)};
\addlegendentry{$\text{MAPH}_{3,2}$ fitted here}
\addplot[lindcol, dotted, very thick] coordinates {(0.0,0.00000)(0.5,0.00856)(1.0,0.02279)(1.5,0.03610)(2.0,0.04646)(2.5,0.05371)(3.0,0.05828)(3.5,0.06077)(4.0,0.06171)(4.5,0.06156)(5.0,0.06064)(5.5,0.05921)(6.0,0.05744)(6.5,0.05548)(7.0,0.05341)(7.5,0.05128)(8.0,0.04916)(8.5,0.04706)(9.0,0.04501)(9.5,0.04302)(10.0,0.04109)(10.5,0.03924)(11.0,0.03745)(11.5,0.03574)(12.0,0.03411)(12.5,0.03254)(13.0,0.03105)(13.5,0.02962)(14.0,0.02825)(14.5,0.02695)(15.0,0.02571)(15.5,0.02452)(16.0,0.02339)(16.5,0.02231)(17.0,0.02128)(17.5,0.02030)(18.0,0.01936)(18.5,0.01847)(19.0,0.01761)(19.5,0.01680)(20.0,0.01602)(20.5,0.01528)(21.0,0.01458)(21.5,0.01390)(22.0,0.01326)(22.5,0.01265)(23.0,0.01206)(23.5,0.01151)(24.0,0.01098)(24.5,0.01047)(25.0,0.00998)(25.5,0.00952)(26.0,0.00908)(26.5,0.00866)(27.0,0.00826)(27.5,0.00788)(28.0,0.00752)(28.5,0.00717)(29.0,0.00684)(29.5,0.00652)(30.0,0.00622)(30.5,0.00593)(31.0,0.00566)(31.5,0.00540)(32.0,0.00515)(32.5,0.00491)(33.0,0.00468)(33.5,0.00447)(34.0,0.00426)(34.5,0.00406)(35.0,0.00388)(35.5,0.00370)(36.0,0.00353)(36.5,0.00336)(37.0,0.00321)(37.5,0.00306)(38.0,0.00292)(38.5,0.00278)(39.0,0.00266)(39.5,0.00253)(40.0,0.00242)(40.5,0.00230)(41.0,0.00220)(41.5,0.00210)(42.0,0.00200)(42.5,0.00191)(43.0,0.00182)(43.5,0.00173)(44.0,0.00165)(44.5,0.00158)(45.0,0.00151)(45.5,0.00144)(46.0,0.00137)(46.5,0.00131)(47.0,0.00125)(47.5,0.00119)(48.0,0.00113)(48.5,0.00108)(49.0,0.00103)(49.5,0.00098)(50.0,0.00094)(50.5,0.00089)(51.0,0.00085)(51.5,0.00081)(52.0,0.00078)(52.5,0.00074)(53.0,0.00071)(53.5,0.00067)(54.0,0.00064)(54.5,0.00061)(55.0,0.00058)(55.5,0.00056)(56.0,0.00053)(56.5,0.00051)(57.0,0.00048)(57.5,0.00046)(58.0,0.00044)(58.5,0.00042)(59.0,0.00040)(59.5,0.00038)(60.0,0.00036)(60.5,0.00035)(61.0,0.00033)(61.5,0.00032)(62.0,0.00030)(62.5,0.00029)(63.0,0.00027)(63.5,0.00026)(64.0,0.00025)(64.5,0.00024)(65.0,0.00023)(65.5,0.00022)(66.0,0.00021)(66.5,0.00020)(67.0,0.00019)(67.5,0.00018)(68.0,0.00017)(68.5,0.00016)(69.0,0.00016)(69.5,0.00015)(70.0,0.00014)(70.5,0.00013)(71.0,0.00013)(71.5,0.00012)(72.0,0.00012)(72.5,0.00011)(73.0,0.00011)(73.5,0.00010)(74.0,0.00010)(74.5,0.00009)(75.0,0.00009)(75.5,0.00008)(76.0,0.00008)(76.5,0.00008)(77.0,0.00007)(77.5,0.00007)(78.0,0.00007)(78.5,0.00006)(79.0,0.00006)(79.5,0.00006)(80.0,0.00005)};
\addlegendentry{Lindqvist \citeyearpar{lindqvist_phase-type_2022}}
\end{axis}
\end{tikzpicture}
\end{minipage}\hfill
\begin{minipage}[t]{0.32\textwidth}\centering
\begin{tikzpicture}
\begin{axis}[width=\textwidth, height=0.92\textwidth,
  title={sub-density $f(u,2)$, hospital death}, title style={font=\small},
  xlabel={length of stay $u$ (days)}, ylabel={$f(u,2)$},
  xmin=0, xmax=80, ymin=0, ymax=0.0075,
  tick label style={font=\scriptsize}, label style={font=\small},
  ]
\addplot[empcol, fill=empcol!18, draw=empcol!55, ybar interval, forget plot] coordinates {(1.0,0.00000)(3.0,0.00231)(5.0,0.00538)(7.0,0.00615)(9.0,0.00462)(11.0,0.00615)(13.0,0.00231)(15.0,0.00154)(17.0,0.00077)(19.0,0.00154)(21.0,0.00077)(23.0,0.00154)(25.0,0.00231)(27.0,0.00154)(29.0,0.00077)(31.0,0.00000)(33.0,0.00000)(35.0,0.00077)(37.0,0.00000)(39.0,0.00000)(41.0,0.00000)(43.0,0.00000)(45.0,0.00000)(47.0,0.00000)(49.0,0.00077)(51.0,0.00000)(53.0,0.00000)(55.0,0.00000)(57.0,0.00077)(59.0,0.00000)(61.0,0.00000)(63.0,0.00077)(65.0,0.00000)(67.0,0.00000)(69.0,0.00000)(71.0,0.00000)(73.0,0.00000)(75.0,0.00000)(77.0,0.00000)(79.0,0.00077)};
\addplot[ourcol, dashed, thick] coordinates {(0.0,0.00204)(0.5,0.00325)(1.0,0.00414)(1.5,0.00478)(2.0,0.00521)(2.5,0.00547)(3.0,0.00561)(3.5,0.00565)(4.0,0.00561)(4.5,0.00552)(5.0,0.00538)(5.5,0.00521)(6.0,0.00501)(6.5,0.00481)(7.0,0.00459)(7.5,0.00438)(8.0,0.00416)(8.5,0.00395)(9.0,0.00374)(9.5,0.00354)(10.0,0.00335)(10.5,0.00316)(11.0,0.00299)(11.5,0.00282)(12.0,0.00267)(12.5,0.00252)(13.0,0.00238)(13.5,0.00225)(14.0,0.00213)(14.5,0.00202)(15.0,0.00191)(15.5,0.00181)(16.0,0.00172)(16.5,0.00163)(17.0,0.00155)(17.5,0.00147)(18.0,0.00140)(18.5,0.00134)(19.0,0.00128)(19.5,0.00122)(20.0,0.00117)(20.5,0.00111)(21.0,0.00107)(21.5,0.00102)(22.0,0.00098)(22.5,0.00094)(23.0,0.00091)(23.5,0.00087)(24.0,0.00084)(24.5,0.00081)(25.0,0.00078)(25.5,0.00075)(26.0,0.00073)(26.5,0.00070)(27.0,0.00068)(27.5,0.00066)(28.0,0.00063)(28.5,0.00061)(29.0,0.00060)(29.5,0.00058)(30.0,0.00056)(30.5,0.00054)(31.0,0.00053)(31.5,0.00051)(32.0,0.00050)(32.5,0.00048)(33.0,0.00047)(33.5,0.00046)(34.0,0.00044)(34.5,0.00043)(35.0,0.00042)(35.5,0.00041)(36.0,0.00040)(36.5,0.00039)(37.0,0.00038)(37.5,0.00037)(38.0,0.00036)(38.5,0.00035)(39.0,0.00034)(39.5,0.00033)(40.0,0.00032)(40.5,0.00031)(41.0,0.00031)(41.5,0.00030)(42.0,0.00029)(42.5,0.00028)(43.0,0.00028)(43.5,0.00027)(44.0,0.00026)(44.5,0.00026)(45.0,0.00025)(45.5,0.00024)(46.0,0.00024)(46.5,0.00023)(47.0,0.00023)(47.5,0.00022)(48.0,0.00022)(48.5,0.00021)(49.0,0.00021)(49.5,0.00020)(50.0,0.00020)(50.5,0.00019)(51.0,0.00019)(51.5,0.00018)(52.0,0.00018)(52.5,0.00017)(53.0,0.00017)(53.5,0.00017)(54.0,0.00016)(54.5,0.00016)(55.0,0.00015)(55.5,0.00015)(56.0,0.00015)(56.5,0.00014)(57.0,0.00014)(57.5,0.00014)(58.0,0.00013)(58.5,0.00013)(59.0,0.00013)(59.5,0.00012)(60.0,0.00012)(60.5,0.00012)(61.0,0.00012)(61.5,0.00011)(62.0,0.00011)(62.5,0.00011)(63.0,0.00010)(63.5,0.00010)(64.0,0.00010)(64.5,0.00010)(65.0,0.00010)(65.5,0.00009)(66.0,0.00009)(66.5,0.00009)(67.0,0.00009)(67.5,0.00008)(68.0,0.00008)(68.5,0.00008)(69.0,0.00008)(69.5,0.00008)(70.0,0.00007)(70.5,0.00007)(71.0,0.00007)(71.5,0.00007)(72.0,0.00007)(72.5,0.00007)(73.0,0.00006)(73.5,0.00006)(74.0,0.00006)(74.5,0.00006)(75.0,0.00006)(75.5,0.00006)(76.0,0.00006)(76.5,0.00005)(77.0,0.00005)(77.5,0.00005)(78.0,0.00005)(78.5,0.00005)(79.0,0.00005)(79.5,0.00005)(80.0,0.00005)};
\addplot[lindcol, dotted, very thick] coordinates {(0.0,0.00000)(0.5,0.00075)(1.0,0.00211)(1.5,0.00340)(2.0,0.00442)(2.5,0.00514)(3.0,0.00559)(3.5,0.00584)(4.0,0.00594)(4.5,0.00593)(5.0,0.00584)(5.5,0.00571)(6.0,0.00554)(6.5,0.00535)(7.0,0.00515)(7.5,0.00495)(8.0,0.00474)(8.5,0.00454)(9.0,0.00434)(9.5,0.00415)(10.0,0.00397)(10.5,0.00379)(11.0,0.00361)(11.5,0.00345)(12.0,0.00329)(12.5,0.00314)(13.0,0.00300)(13.5,0.00286)(14.0,0.00273)(14.5,0.00260)(15.0,0.00248)(15.5,0.00237)(16.0,0.00226)(16.5,0.00215)(17.0,0.00205)(17.5,0.00196)(18.0,0.00187)(18.5,0.00178)(19.0,0.00170)(19.5,0.00162)(20.0,0.00155)(20.5,0.00148)(21.0,0.00141)(21.5,0.00134)(22.0,0.00128)(22.5,0.00122)(23.0,0.00116)(23.5,0.00111)(24.0,0.00106)(24.5,0.00101)(25.0,0.00096)(25.5,0.00092)(26.0,0.00088)(26.5,0.00084)(27.0,0.00080)(27.5,0.00076)(28.0,0.00073)(28.5,0.00069)(29.0,0.00066)(29.5,0.00063)(30.0,0.00060)(30.5,0.00057)(31.0,0.00055)(31.5,0.00052)(32.0,0.00050)(32.5,0.00047)(33.0,0.00045)(33.5,0.00043)(34.0,0.00041)(34.5,0.00039)(35.0,0.00037)(35.5,0.00036)(36.0,0.00034)(36.5,0.00032)(37.0,0.00031)(37.5,0.00030)(38.0,0.00028)(38.5,0.00027)(39.0,0.00026)(39.5,0.00024)(40.0,0.00023)(40.5,0.00022)(41.0,0.00021)(41.5,0.00020)(42.0,0.00019)(42.5,0.00018)(43.0,0.00018)(43.5,0.00017)(44.0,0.00016)(44.5,0.00015)(45.0,0.00015)(45.5,0.00014)(46.0,0.00013)(46.5,0.00013)(47.0,0.00012)(47.5,0.00011)(48.0,0.00011)(48.5,0.00010)(49.0,0.00010)(49.5,0.00009)(50.0,0.00009)(50.5,0.00009)(51.0,0.00008)(51.5,0.00008)(52.0,0.00007)(52.5,0.00007)(53.0,0.00007)(53.5,0.00007)(54.0,0.00006)(54.5,0.00006)(55.0,0.00006)(55.5,0.00005)(56.0,0.00005)(56.5,0.00005)(57.0,0.00005)(57.5,0.00004)(58.0,0.00004)(58.5,0.00004)(59.0,0.00004)(59.5,0.00004)(60.0,0.00004)(60.5,0.00003)(61.0,0.00003)(61.5,0.00003)(62.0,0.00003)(62.5,0.00003)(63.0,0.00003)(63.5,0.00003)(64.0,0.00002)(64.5,0.00002)(65.0,0.00002)(65.5,0.00002)(66.0,0.00002)(66.5,0.00002)(67.0,0.00002)(67.5,0.00002)(68.0,0.00002)(68.5,0.00002)(69.0,0.00002)(69.5,0.00001)(70.0,0.00001)(70.5,0.00001)(71.0,0.00001)(71.5,0.00001)(72.0,0.00001)(72.5,0.00001)(73.0,0.00001)(73.5,0.00001)(74.0,0.00001)(74.5,0.00001)(75.0,0.00001)(75.5,0.00001)(76.0,0.00001)(76.5,0.00001)(77.0,0.00001)(77.5,0.00001)(78.0,0.00001)(78.5,0.00001)(79.0,0.00001)(79.5,0.00001)(80.0,0.00001)};
\end{axis}
\end{tikzpicture}
\end{minipage}
\caption{\aachg{Intensive-care length of stay, all $650$ records, compared with the data on two footings. Grey is the empirical sub-distribution (left, as a step function) and the per-cause histogram of observed event times on $2$-day bins (centre and right, scaled to a density on the full sample); blue is the $\text{MAPH}_{3,2}$ fitted here with the censored EM, orange the published fit of \citet{lindqvist_phase-type_2022}. Left: both causes, the upper group of curves alive discharge and the lower hospital death. Centre and right: the corresponding sub-densities, on separate scales because discharge is an order of magnitude the more frequent. The sub-densities localize the discrepancy that the sub-distribution functions integrate away --- both fits miss the early discharge peak, but the orange curve is the flatter of the two, which is what the supremum distances of \Cref{tab:icu-lindqvist} ($0.133$ against $0.067$) record. Death is captured closely by both. The Coxian refit, omitted for legibility, lies between the two throughout.}}
\label{fig:icu-cif}
\end{rcrevision}
\end{figure}
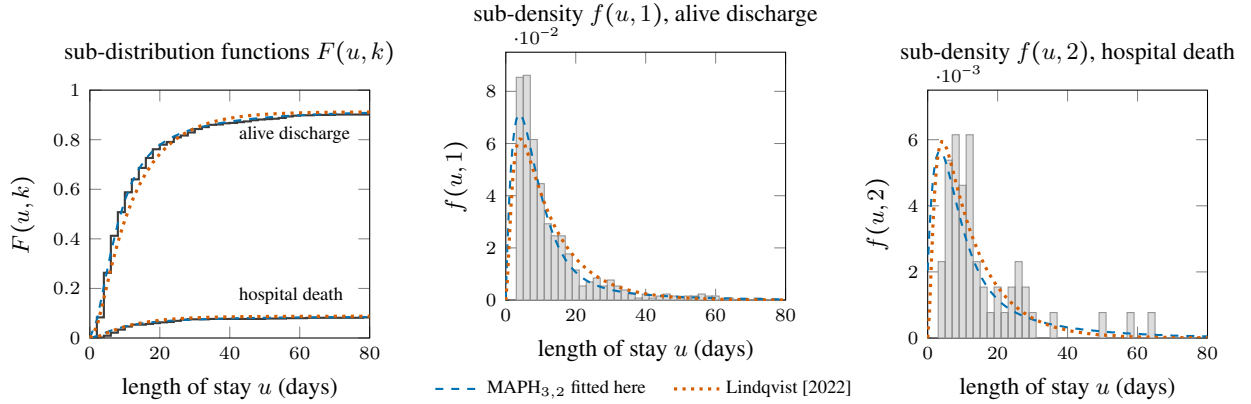

\begin{rcrevision}
\paragraph{Comparison with Lindqvist's fit.}
\aachg{\citet{lindqvist_phase-type_2022} fits a Coxian competing-risks model with $m=3$ to exactly these $650$ patients. Re-expressing his published estimates in our $(\alpha,T,D)$ coordinates and scoring them on the same records with the same likelihood puts all three fits on one footing, in \Cref{tab:icu-lindqvist}.}

Our fitting procedure attains a higher likelihood. Within Lindqvist's own model class and at his own order, refitting with the censored EM of \Cref{sec:em-alg} \subchg{improves the log-likelihood by $21.3$ nats using the same eight free parameters, and more than halves the supremum discrepancy against the Aalen--Johansen discharge incidence, from $0.133$ to $0.062$; the general $\text{MAPH}_{3,2}$, with its six extra parameters, adds only a further $0.2$ nats, so on these data the Coxian restriction costs essentially nothing at this order.} \subchg{The gap is a matter of which optimum is reached rather than of model class. Re-expressed in our coordinates the published canonical mixture is an $8$-phase MAPH, since each of its three mixture components contributes its own generalized-Erlang block; started {\em at} those estimates and left to run, our EM climbs from $-2396.34$ to $-2361.99$. That fit uses more phases than any row of \Cref{tab:icu-lindqvist} and is not comparable with them, but it does show that the published point is far from stationary for the likelihood used here.}

\aachg{All three rows are scored on our continuous-time likelihood, which is the appropriate common yardstick but need not be the objective under which the published estimates were obtained.}
\end{rcrevision}

\begin{table}[ht]
\centering
\begin{rcrevision}
\begin{tabular}{l r r r r r}
\toprule
fit & log-lik. & \#\,par & $\Prob(\kappa=1)$ & $\ex{\tau}$ (d) & $\sup|F-F_{\text{AJ}}|$ \\
\midrule
\citet{lindqvist_phase-type_2022}, published & $-2396.34$ & 8  & $0.9122$ & $12.65$ & $0.133$ \, / \, $0.011$ \\
Coxian $\text{MAPH}_{3,2}$, refitted here             & $-2375.02$ & 8  & $0.9142$ & $12.36$ & $0.062$ \, / \, $0.009$ \\
general $\text{MAPH}_{3,2}$, fitted here              & $\mathbf{-2374.84}$ & 14 & $0.9142$ & $12.36$ & $\mathbf{0.062}$ \, / \, $0.009$ \\
\bottomrule
\end{tabular}
\caption{The same $650$ intensive-care records, the same order $m=3$, and the same likelihood on the day scale. The first row is the fit of \citet{lindqvist_phase-type_2022} re-expressed in our coordinates from its published canonical parameters; the second restricts our censored EM to that same Coxian model class; the third is the unrestricted fit. The last column gives the supremum distance to the Aalen--Johansen cumulative incidence for alive discharge and for hospital death.}
\label{tab:icu-lindqvist}
\end{rcrevision}
\end{table}

\begin{rcrevision}
\aachg{\Cref{fig:icu-cif} sets the $m=3$ fits against the data, as sub-distribution functions and as sub-densities.} Both capture hospital death very accurately, and they differ mainly in how well they follow the steep early rise of the discharge incidence --- the feature \citet{lindqvist_phase-type_2022} singled out in reporting that ``the fit seems very good for the outcome `hospital death', while the model with $m=3$ is seemingly not able to pick up completely the steepness of the first part of the cumulative incidence function for `alive discharge'.'' \aachg{The sub-density panels show it directly: the observed discharge histogram peaks at about $0.086$ per day within the first two weeks, and neither fit reaches that peak --- ours rises to $0.071$, the published fit to $0.062$.} That difficulty is genuine and shared by both curves, but it is roughly halved for the fits obtained here, and it continues to shrink with the phase order, reaching $0.062$ at $m=5$. It is a limitation of low-order phase-type incidence rather than of the estimator.

Finally, the practical effect of using the censored records rather than discarding them is small on this dataset, as it should be: only $6$ of $650$ records are censored, $0.9\%$. \subchg{Refitting at $m=4$ from the $644$ events alone and scoring the result on all $650$ gives $-2359.88$ against $-2359.85$ for the censored fit, with absorption probabilities agreeing to three decimals and mean stays of $12.20$ against $12.37$ days} --- the censored fit's slightly longer mean reflecting that a stay censored at $31$ days is known to have lasted {\em at least} that long. At this censoring fraction the likelihood difference is dominated by which local optimum each multi-start run reached rather than by the censoring itself, and we would not read anything into it.

The value of the extended method here is therefore not that it moves these particular estimates far, but that it removes the need to choose between discarding records and misrepresenting them. That choice does become decisive when censoring is heavy. In a controlled check in the package's test suite, a known $\text{MAPH}_{3,2}$ observed under an administrative horizon that censors $54\%$ of subjects is recovered to within $2.3\%$ of its true mean absorption time by the censored fit, whereas discarding the censored records understates that mean by $74\%$, and recording each censoring time as an event understates it by $53\%$ while also displacing the absorption probabilities by $0.30$.
\end{rcrevision}
\begin{rcrevision}
The example shows that MAPH distributions fit real competing-risks data well when incidence accumulates early, as it does here. Where it does not --- when a cause has a delayed onset, so that its incidence stays at zero over an initial interval --- a low-order MAPH cannot follow it, since absorption is possible from the very first phase and the fitted incidence rises from $u=0$; reproducing a sharp delay needs many phases in an Erlang-like chain, at a corresponding cost in parameters. This is a property of phase-type modelling rather than of the estimator, and it is the kind of diagnostic \aachg{the sub-density panels of \Cref{fig:icu-cif} make visible}.

The right-censored extension of \Cref{sec:em-alg,prop:e-step-censored} is implemented, checked against numerical tail integration, and exercised on the full $650$-record ICU sample, where it improves on the previously published fit to those data. What remains is a systematic synthetic study across censoring levels, censoring mechanisms and phase orders, of the kind reported in \Cref{sec:verification} for complete data; and, more substantially, the convergence behaviour of the relaxed-and-projected iteration. The ICU analysis makes the latter concrete: the fits reported above depend visibly on the starting point, and a single run from the default initialization stalls tens of nats below the multi-start best at every order tried. A principled starting rule, or a modification of the M-step that restores monotonicity, would matter more in practice than further extensions of the observation model.
\end{rcrevision}

\subsection{Software}
\label{sec:software}

\begin{rcrevision}
The computations above were carried out with two companion \textsf{Julia} packages.
\texttt{PhaseTypeDistributions.jl} constructs, evaluates and simulates phase-type and MAPH
distributions, representing a MAPH law as a \texttt{MAPHDist($\alpha$,T,D)} object in the
\texttt{Distributions.jl} interface. \texttt{PhaseTypeDistributionsFitting.jl} implements the
exact-event and right-censored E-steps of \Cref{prop:e-step,prop:e-step-censored}, the M-step
\eqref{eq:approx-m-step} and the initializations of \Cref{sec:init-heuristic}, and exposes them
through a \texttt{fit\_mle} interface that takes the exact \texttt{(time, cause)} pairs and the
right-censoring times as separate arguments, so that a censored record never carries a placeholder
cause label. Both are available at
\url{https://github.com/Julia-Matrix-Analytic-Probability/PhaseTypeDistributions.jl} and
\url{https://github.com/Julia-Matrix-Analytic-Probability/PhaseTypeDistributionsFitting.jl}.
\end{rcrevision}

%% file: sections/conclusion.tex
\section{Conclusion}
\label{sec:conclusion}

\begin{rcrevision}
We introduced the Multi-Absorbing Phase-Type (MAPH) family as a model for competing-risks data and developed \subchg{an} EM algorithm for exact event-time/cause records and independently right-censored records. The censored E-step conditions on survival to the censoring time and completes the subject's latent path through its eventual cause; the resulting cause-resolved sufficient-statistic expectations \subchg{feed the same closed-form M-step}, \subchg{and cost one block matrix exponential per censored record, whatever the number of causes.} Both E-steps are implemented in two companion Julia packages. On the intensive-care data of \Cref{sec:realdata}, analysed in full for the first time, the method attains a higher likelihood than the fit previously published for those data, at the same order and within the same model class.
\end{rcrevision}

\begin{rcrevision}
\subchg{Several directions remain open. The M-step being exact, the observed-data likelihood is non-decreasing along the iteration (\Cref{cor:monotone}), so the ascent guarantee that motivated much of the earlier discussion now holds. What monotone ascent does not deliver is convergence of the iterates themselves, nor convergence to a global maximum: the likelihood surface of a $\text{MAPH}_{m,n}$ law is not concave and the ICU fits show a marked sensitivity to the starting point, single-start runs stalling tens of nats below the best fit found by a multi-start panel. Stationary-point convergence under the usual EM regularity conditions, and a principled starting rule, are the natural next questions.} For the censoring extension specifically, the remaining priorities are an independent algebraic review of \Cref{prop:e-step-censored}, a systematic simulation study across censoring mechanisms and rates, and, for large or stiff generators, matrix-exponential actions or uniformization in place of the dense block exponentials used here. \subchg{Principled phase-order selection also remains open.}
\end{rcrevision}

\section*{Acknowledgements}

We thank Benoit Liquet for initial discussions. We also acknowledge the use of Claude Code (Anthropic) for assistance with formulations and proofreading.

%% file: sections/appendix.tex
\appendix
\appendixpage
\addappheadtotoc

\input{sections/initialization}
\input{sections/appendix-proofs-maph}
\input{sections/appendix-proofs-est}
\input{sections/appendix-proofs-cens}

%% file: sections/initialization.tex
\section{An Initialization Heuristic}
\label{sec:init-heuristic}

\begin{rcrevision}
This appendix presents the two initializations used to start the algorithm of \Cref{sec:em-alg}. For exact and right-censored data ${\cal D}=\{(y^\ell,\delta^\ell,k^\ell)\}_{\ell=1}^L$, partition the observed event indices by cause,
\[
\eta_k = \{\ell \in \{1,\ldots,L\}:\delta^\ell=1,~k^\ell = k\},
\qquad
k = 1,\ldots,n,
\]
and write $d_k=|\eta_k|$ and $d=\sum_{k=1}^n d_k$. We assume $d_k>0$ and use the observed-event proportions as heuristic routing weights,
\[
\hat{\pi}_k = \frac{d_k}{d},
\qquad
\text{with}
\qquad
\hat{\pi} := (\hat{\pi}_1,\ldots,\hat{\pi}_n).
\]
For the simplified initialization, use the exposure-per-event scale
\begin{equation}
\label{eq:censored-init-scale}
\bar\mu=\frac{\sum_{\ell=1}^L y^\ell}{d}.
\end{equation}
Its reciprocal is the all-cause exponential maximum-likelihood rate under independent right censoring, and it reduces to the ordinary sample mean when every record is an event. The raw event-time mean within cause $k$ is generally biased under censoring, so it must not be substituted for $\ex{\tau\mid\kappa=k}$. The moment-based initialization below is therefore used with censored data only when censoring-adjusted or external targets $(\hat\mu_k,\hat c_k^2)$ are supplied; otherwise the algorithm is started from the simplified initialization. For uncensored data, $\hat\mu_k$ and $\hat c_k^2$ retain their original meanings as the empirical conditional mean and SCV, and $\bar\mu=\hat\pi^\intercal\hat\mu$. \subchg{The phase budget below is filled greedily from $k=1$ upward, so we index causes in decreasing observed-event frequency in order to spend it on the most frequent causes,}
\begin{equation}
\label{eq:pi-order}
\hat{\pi}_{1} \ge \hat{\pi}_{2} \ge \ldots \ge \hat{\pi}_{n} > 0.
\end{equation}
The number of absorbing states $n$ is read off the data, while the number of transient phases $m \ge 1$ is a model size choice, as discussed in \Cref{sec:em-alg}.

The first initialization below is deliberately crude: it uses only $\hat\pi$ and $\bar\mu$, is compatible with right censoring, requires no minimal phase count, and yields strictly positive parameters. The second, moment-based initialization additionally matches supplied conditional means $\hat\mu_k$ and SCVs $\hat c_k^2$, cause by cause, when $m$ is large enough.
\end{rcrevision}

\subsection{Simplified initialization}
\label{sec:init-simplified}

\begin{rcrevision}
For $m\ge2$, fix a routing constant $\beta\in(0,1)$ (e.g.\ $\beta=\tfrac12$); for $m=1$, set $\beta=0$. Fix a rate $\omega>0$ as below and initialize, for $i\ne j$ and $k=1,\ldots,n$,
\end{rcrevision}
\begin{equation}
\label{eq:simplified-init}
\alpha_i = \frac{1}{m},
\qquad
T_{ii} = -\omega,
\qquad
T_{ij} = \frac{\beta\,\omega}{m-1},
\qquad
D_{ik} = (1-\beta)\,\omega\,\hat{\pi}_k.
\end{equation}
Each row of $[T ~\vert~ D]$ sums to zero, so this is a valid $\text{MAPH}_{m,n}$ parameterization; moreover all entries of $\alpha$, all off-diagonal entries of $T$, and all entries of $D$ are strictly positive, so \subchg{the assumptions of \Cref{sec:maph-dist} hold} and no parameter is initialized at the boundary of the parameter space.

Under \eqref{eq:simplified-init} the law of $(\tau,\kappa)$ is explicit. From every phase the sojourn is exponential with rate $\omega$, after which the chain moves to another transient phase with probability $\beta$ (uniformly) and absorbs with probability $1-\beta$, choosing cause $k$ with probability $\hat{\pi}_k$. Since all phases are exchangeable, the number of sojourns until absorption is geometric with success probability $1-\beta$, and a geometric sum of i.i.d.\ exponentials of rate $\omega$ is again exponential, with rate $(1-\beta)\omega$. Hence
\[
\tau \sim \mathrm{Exp}\big((1-\beta)\,\omega\big)
\quad \text{independently of} \quad \kappa,
\qquad
\Prob(\kappa = k) = \hat{\pi}_k.
\]
Choosing
\begin{equation}
\label{eq:omega-simplified}
\omega = \frac{1}{(1-\beta)\,\bar{\mu}}
\end{equation}
\begin{rcrevision}
sets the initialization's all-cause mean scale to $\ex{\tau}=\bar\mu$. For uncensored data this matches the overall sample mean; with censoring it matches the exposure-per-event scale \eqref{eq:censored-init-scale}. The routing probabilities are $\hat\pi$, and the common conditional distributions are exponential with mean $\bar\mu$, so every conditional SCV is initialized at $1$.
\end{rcrevision}

\paragraph{Symmetry breaking.} As constructed, \eqref{eq:simplified-init} is {\em exchangeable}: every phase has the same initial probability, sojourn rate, routing row and absorption row -- and the EM update preserves exchangeability, since each conditional expectation of the E-step then takes the same value in every phase, making every subsequent iterate exchangeable as well. An EM run started exactly at \eqref{eq:simplified-init} is therefore confined to the exchangeable submanifold, on which $\tau$ is exponential and independent of $\kappa$, and can never differentiate the phases; in practice such a run ``converges'' within a couple of iterations to that degenerate fit. A symmetry-breaking perturbation is therefore part of the specification, for $m \ge 2$. We perturb deterministically: for a jitter parameter $\delta \in (0,1)$ (default $\delta = \tfrac12$), row $i$ of $[T \,\vert\, D]$ is scaled by $1 + \delta z_i$, where $z_i = 2(i-1)/(m-1) - 1$ spreads $z_1,\ldots,z_m$ evenly over $[-1,1]$. Scaling whole rows changes the per-phase sojourn rates but not the embedded jump chain, so the matched absorption probabilities $\hat{\pi}$ are untouched; the generator is then rescaled by a common factor so that $\ex{\tau} = \bar{\mu}$ again holds exactly. Only the exponential shape of $\tau$ (and its independence of $\kappa$) becomes approximate.

\subsection{Moment based initialization}
\label{sec:init-moment}

\begin{rcrevision}
Given valid targets $(\hat\mu_k,\hat c_k^2)$---empirical for uncensored data or externally censoring-adjusted---we construct an $\text{MAPH}_{m,n}$ that attempts to match them. The construction routes the chain through a common front end into per-cause phase-type blocks. The transient phases are partitioned as $\cs = \cs_1 \cup \cs_2$ with $|\cs_1| \ge 1$. The phases of $\cs_1$ do not communicate with one another, and from every phase of $\cs_1$ the chain exits after an exponential sojourn of rate $\omega$; hence the time $\xi$ spent in $\cs_1$ satisfies $\xi \sim \mathrm{Exp}(\omega)$ regardless of the entry phase. Upon leaving $\cs_1$ the chain routes, with probabilities $\hat{\pi}_k$, either into a phase-type block built for cause $k$ (occupying phases of $\cs_2$) or, for causes for which no block is built, directly into the absorbing state $s^0_k$. Each block is a small phase-type distribution -- a two-phase hyper-exponential, a hypo-exponential (generalized Erlang), or a single exponential -- whose absorption leads only to its own cause. Thus, conditional on absorption through block $k$, the absorption time decomposes as $\tau = \xi + \tau_k$ with $\xi$ and $\tau_k$ independent and $\tau_k \sim \mathrm{PH}(\alpha_k, T_k)$, the duration of block $k$.
\end{rcrevision}

\paragraph{Block targets.} Since the front end contributes the independent summand $\xi \sim \mathrm{Exp}(\omega)$, the block for cause $k$ must not match $(\hat{\mu}_k, \hat{c}_k^2)$ itself, but corrected targets, as follows.

\begin{proposition}
\label{prop:init-correction}
Fix $\omega > 0$ and let $\xi \sim \mathrm{Exp}(\omega)$ be independent of $\tau_k$. Then $\ex{\xi + \tau_k} = \hat{\mu}_k$ and $\mathrm{SCV}(\xi+\tau_k) = \hat{c}_k^2$ hold if and only if the mean $\mu_k = \ex{\tau_k}$ and SCV $c_k^2 = \mathrm{SCV}(\tau_k)$ of the block duration are
\begin{equation}
\label{eq:mu_k-correction-omega}
    \mu_{k} = \hat{\mu}_{k} - \frac{1}{\omega},
\end{equation}
and,
\begin{equation}
\label{eq:scv_calculated}
    c_{k}^2 = \frac{\hat{c}_{k}^2(\omega \mu_{k}+1)^2-1}{(\omega \mu_{k})^2}.
\end{equation}
\end{proposition}
\begin{proof}
For the mean,
\[
\hat{\mu}_{k} = \EX[\xi+\tau_{k}] = \frac{1}{\omega}+\mu_k,
\]
which is \eqref{eq:mu_k-correction-omega}. For the SCV, using independence,
\[
\hat{c}^2_{k} = \frac{\var{\xi}+\var{\tau_{k}}}{\EX[\xi + \tau_k]^2}  =
\frac{\frac{1}{\omega^2}+\var{\tau_{k}}}{\frac{1}{\omega^2}+\frac{2}{\omega}\mu_k+\mu_k^2}
= \frac{\frac{1}{\omega^2 \mu_{k}^2}+c_{k}^2}{\frac{1}{\omega^2 \mu_{k}^2}+\frac{2}{\omega \mu_{k}}+1},
\]
and rearranging yields \eqref{eq:scv_calculated}.
\end{proof}

\paragraph{Choice of $\omega$.} The corrected targets are usable only if $\omega$ is large enough. Specifically we require, for every $k$: (1) the corrected mean is positive, $\mu_k > 0$; (2) the corrected SCV is positive, $c_k^2 > 0$; and (3) the correction does not change the sign of $c_k^2 - 1$ relative to $\hat{c}_k^2 - 1$, so that the family of block selected for cause $k$ (hyper- versus hypo-exponential, below) agrees with the empirical $\hat{c}_k^2$.

\begin{proposition}
\label{prop:init-omega}
Requirements (1)--(3) hold for every $k = 1,\ldots,n$ whenever $\omega > \max_{k} \omega^*_k$, where
\begin{equation}
\label{eq:omega-threshold}
\omega^*_k
=
\frac{1}{\hat{\mu}_k}
\times
\begin{cases}
1, & \hat{c}_k^2 \ge 1,\\[1ex]
\dfrac{1}{\sqrt{\hat{c}_k^2}}, & 0 < \hat{c}_k^2 < \tfrac12,\\[2ex]
\dfrac{1+\sqrt{2\hat{c}_k^2-1}}{1-\hat{c}_k^2}, & \tfrac12 \le \hat{c}_k^2 < 1.
\end{cases}
\end{equation}
Moreover, in the regime $\tfrac12 \le \hat{c}_k^2 < 1$ the threshold is sharp: there are values $\omega$ with $\omega\hat{\mu}_k > 1$ but $\omega < \omega^*_k$ for which $c_k^2 > 1$, violating (3).
\end{proposition}
\begin{proof}
Fix $k$ and write $\psi = \omega\hat{\mu}_k$. By \eqref{eq:mu_k-correction-omega}, $\omega\mu_k = \psi - 1$, so requirement (1) is $\psi > 1$, and \eqref{eq:scv_calculated} becomes
\begin{equation}
\label{eq:scv-psi}
c_k^2 = \frac{\hat{c}_k^2\,\psi^2 - 1}{(\psi-1)^2}.
\end{equation}
Hence requirement (2) reads $\psi > 1/\sqrt{\hat{c}_k^2}$, and
\[
c_k^2 - 1 = \frac{g(\psi)}{(\psi-1)^2},
\qquad
g(\psi) := (\hat{c}_k^2-1)\,\psi^2 + 2\psi - 2 .
\]
We consider the three regimes of \eqref{eq:omega-threshold}.

If $\hat{c}_k^2 \ge 1$ then for $\psi > 1$ we have $g(\psi) \ge 2\psi - 2 > 0$, so $c_k^2 > 1 > 0$: requirements (2) and (3) are automatic and only (1) binds, giving $\omega^*_k = 1/\hat{\mu}_k$.

If $\hat{c}_k^2 < 1$ then $g$ is a concave parabola with $g(1) = \hat{c}_k^2 - 1 < 0$ and discriminant $4(2\hat{c}_k^2-1)$. When $\hat{c}_k^2 < \tfrac12$ the discriminant is negative, so $g < 0$ everywhere and $c_k^2 < 1$ for every $\psi$: requirement (3) is automatic, and the binding requirement is (2), $\psi > 1/\sqrt{\hat{c}_k^2}$, which also implies (1) since $1/\sqrt{\hat{c}_k^2} > 1$.

When $\tfrac12 \le \hat{c}_k^2 < 1$ the parabola $g$ has real roots
\[
\psi_\pm = \frac{1 \pm \sqrt{2\hat{c}_k^2-1}}{1-\hat{c}_k^2},
\]
with $g > 0$ on $(\psi_-,\psi_+)$ and $g < 0$ outside $[\psi_-,\psi_+]$. For $\psi > \psi_+$ we therefore have $c_k^2 < 1$, which is (3). For (2), note that at $\psi = \psi_+$ equation \eqref{eq:scv-psi} gives $c_k^2 = 1$, so its numerator satisfies $\hat{c}_k^2\psi_+^2 - 1 = (\psi_+-1)^2 > 0$; since the numerator is increasing in $\psi$, it remains positive for all $\psi \ge \psi_+$, giving $c_k^2 > 0$. For (1), $\psi_+ \ge 1/(1-\hat{c}_k^2) > 1$. Hence $\omega^*_k = \psi_+/\hat{\mu}_k$ suffices. For sharpness, $g(1) < 0$ places $1$ outside the interval $(\psi_-,\psi_+)$, so $1 < \psi_- < \psi_+$; any $\psi \in (\psi_-,\psi_+)$ satisfies (1) yet has $g(\psi) > 0$, i.e.\ $c_k^2 > 1$, violating (3).
\end{proof}

Any $\omega$ exceeding the bound may be used, but the multiple matters. Small multiples place the construction far from the per-cause moment match: at $\omega = 2\max_k \omega^*_k$, for example, a cause with $\hat{c}_k^2 \ge 1$ has $\omega \approx 2/\hat{\mu}_k$, so the exponential front end consumes about half of its conditional mean, and empirically the EM started there tends to land in poorer local optima. Larger multiples shrink the front-end share $1/\omega$ of each conditional mean and drive $c_k^2 \to \hat{c}_k^2$, at the cost of a stiffer generator (larger rates) in the matrix-exponential computations of the E-step. Our implementation defaults to $\omega = 10\max_k \omega^*_k$, which we have found to balance the two effects well.

\paragraph{Block sizes and the phase budget.} With $\omega$ fixed, compute the corrected targets $(\mu_k, c_k^2)$ for all $k$. The number of phases needed by the block for cause $k$ is
\begin{equation}\label{eq:num_phase}
m_{k} = \begin{cases}
2, & c^2_{k}  > 1, \qquad\qquad (\text{hyper-exponential})\\[0.5ex]
\big\lceil 1/c^2_{k} \big\rceil, & 0 < c^2_k \le 1. \qquad (\text{hypo-exponential})
\end{cases}
\end{equation}
Note that $\lceil 1/c_k^2 \rceil = 1$ exactly when $c_k^2 = 1$, in which case the block is a single exponential phase. Since at least one phase is reserved for the front end $\cs_1$, at most $m-1$ phases are available for blocks, and we build blocks for as many causes as the budget allows, in the priority order \eqref{eq:pi-order}:
\begin{equation}
\label{eq:num-fitted}
p = \max\Big\{p' \in \{0,1,\ldots,n\} ~:~ \sum_{k=1}^{p'} m_k \le m-1\Big\},
\end{equation}
with the empty sum equal to zero, so $p$ is well defined. We then set $|\cs_2| = \sum_{k=1}^p m_k$ and $m_1 := |\cs_1| = m - |\cs_2| \ge 1$. If $p = n$, every cause is moment matched (and any leftover phases simply enlarge the front end with additional, exchangeable, entry phases). If $p < n$, the causes $k = p+1,\ldots,n$ are matched in absorption probability only: they absorb directly out of $\cs_1$, so their conditional time is initialized as $\mathrm{Exp}(\omega)$. In the extreme case $p = 0$ no blocks are built at all and, comparing with \Cref{sec:init-simplified}, the construction reduces to the simplified initialization with $\beta = 0$ (and a different choice of $\omega$); in that case the simplified initialization of \Cref{sec:init-simplified} is preferable.

\paragraph{Hyper-exponential blocks.} For $c_k^2 > 1$ the block is a two-phase hyper-exponential: an exponential of rate $\nu_{k1}/\mu_k$ with probability $a_k$, mixed with an exponential of rate $\nu_{k2}/\mu_k$ with probability $1-a_k$, where $(a_k, \nu_{k1}, \nu_{k2})$ describe the unit-mean version of the block. Matching the unit mean and the second moment (which equals $c_k^2+1$ for a unit-mean random variable with SCV $c_k^2$) requires
\begin{equation}
\label{eq:hyper-equations}
\frac{a_{k}}{\nu_{k1}}+\frac{1-a_{k}}{\nu_{k2}} = 1
\qquad\text{and}\qquad
\frac{a_{k}}{\nu_{k1}^2}+\frac{1-a_{k}}{\nu_{k2}^2} = \frac{c_{k}^2+1}{2}.
\end{equation}
These are two equations in three unknowns; we pin the remaining degree of freedom by fixing $\nu_{k1} = \frac{1}{c_k^2+1}$, under which the unique solution of \eqref{eq:hyper-equations} is
\begin{equation}
\label{eq:hyper-solution}
a_{k} = \frac{c^2_{k}-1}{(c_{k}^2+1)(2c_{k}^2-1)},
\qquad
\nu_{k2} =  \frac{2c_{k}^2}{c_{k}^2+1},
\end{equation}
as is verified by direct substitution. For $c_k^2 > 1$ we have $a_k \in (0,1)$ and $\nu_{k1},\nu_{k2} > 0$, and as $c_k^2 \downarrow 1$ the block degenerates continuously to a single unit-rate exponential ($a_k \to 0$, $\nu_{k2} \to 1$). The phase-type representation of the block is
\[
\alpha_{k} = [a_{k}~~~1-a_{k}],
\qquad
T_{k} =
\frac{1}{\mu_{k}}
\begin{bmatrix}
-\nu_{k1} & 0 \\
0 & -\nu_{k2}
\end{bmatrix}.
\]

\paragraph{Hypo-exponential blocks.} For $0 < c_k^2 \le 1$ the block is a hypo-exponential (generalized Erlang): a sum of $m_k = \lceil 1/c_k^2 \rceil$ independent exponentials, the first of rate $\nu_{k1}/\mu_k$ and the remaining $m_k - 1$ of common rate $\nu_{k2}/\mu_k$, where the unit-mean rates are
\begin{equation}
\label{eq:hypo-rates}
\nu_{k1} = \frac{m_k}{1+\sqrt{(m_k-1)(m_k c_{k}^2-1)}},
\qquad
\nu_{k2} =  \frac{(m_k-1)\,\nu_{k1}}{\nu_{k1}-1}.
\end{equation}
The radicand is non-negative since $m_k c_k^2 \ge 1$ by the choice of $m_k$ in \eqref{eq:num_phase}, and for $c_k^2 < 1$ we have $\nu_{k1} > 1$ so that $\nu_{k2} > 0$; when $m_k = 1$ (i.e.\ $c_k^2 = 1$) the block is the single unit-rate exponential $\nu_{k1} = 1$ and $\nu_{k2}$ is not needed. A direct computation confirms the unit mean and the SCV,
\[
\frac{1}{\nu_{k1}} + \frac{m_k-1}{\nu_{k2}} = 1,
\qquad
\frac{1}{\nu_{k1}^2} + \frac{m_k-1}{\nu_{k2}^2} = c_k^2 .
\]
The phase-type representation of the block is $\alpha_k = [1~~0~~\cdots~~0]$ and the $m_k \times m_k$ bidiagonal sub-generator
\[
T_{k} = \frac{1}{\mu_{k}}
\begin{bmatrix}
-\nu_{k1} & \nu_{k1} & & \\
 & -\nu_{k2} & \ddots & \\
 & & \ddots & \nu_{k2}\\
 & & & -\nu_{k2}
\end{bmatrix},
\]
with absorption (at rate $\nu_{k2}/\mu_k$) out of the last phase only.

\paragraph{Assembly.} Order the phases with $\cs_1$ first, followed by the blocks $k = 1,\ldots,p$, and let $d_k = -T_k \mathbf{1}_{m_k}$ denote the exit-rate (column) vector of block $k$. The initialized $\text{MAPH}_{m,n}(\alpha, T, D)$ is
\begin{equation}
\label{eq:init-assembly}
\alpha = \Big[\tfrac{1}{m_1}\mathbf{1}_{m_1}^\top ~~\Big\vert~~ 0 ~\cdots~ 0\Big],
\qquad
T=\begin{bmatrix}
-\omega I_{m_1} & \omega\hat{\pi}_1\,\mathbf{1}_{m_1}\alpha_1 & \cdots & \omega\hat{\pi}_p\,\mathbf{1}_{m_1}\alpha_p\\
 & T_1 & &\\
 & & \ddots &\\
 & & & T_p
\end{bmatrix},
\end{equation}
where blank blocks are zero, and the $m \times n$ matrix $D$ has the following non-zero entries: in the rows of $\cs_1$, $D_{ik} = \omega\hat{\pi}_k$ for the unfitted causes $k = p+1,\ldots,n$ (these columns are zero elsewhere); and in the rows of block $k \le p$, the $k$'th column equals $d_k$ (these rows are zero elsewhere). Each $\cs_1$ row of $[T \,\vert\, D]$ sums to $-\omega + \omega\sum_{k\le p}\hat{\pi}_k\,\alpha_k\mathbf{1}_{m_k} + \omega\sum_{k>p}\hat{\pi}_k = 0$, and each block row sums to zero by $T_k\mathbf{1}_{m_k} + d_k = 0$, so \eqref{eq:init-assembly} is a valid parameterization.

\begin{proposition}
\label{prop:init-matching}
Let $\omega$ satisfy the bound of \Cref{prop:init-omega} and let $p$ be as in \eqref{eq:num-fitted}. The $\text{MAPH}_{m,n}$ of \eqref{eq:init-assembly} satisfies
\[
\Prob(\kappa = k) = \hat{\pi}_k, \quad k = 1,\ldots,n;
\qquad
\ex{\tau \mid \kappa = k} = \hat{\mu}_k
~~\text{and}~~
\mathrm{SCV}(\tau \mid \kappa = k) = \hat{c}_k^2, \quad k \le p;
\]
while $\tau \mid \kappa = k \sim \mathrm{Exp}(\omega)$ for $k > p$.
\end{proposition}
\begin{proof}
The chain starts in $\cs_1$, holds for $\xi \sim \mathrm{Exp}(\omega)$, and then routes: into block $k$ (entering according to $\alpha_k$) with probability $\hat{\pi}_k$ for $k \le p$, or directly to $s^0_k$ with probability $\hat{\pi}_k$ for $k > p$; the routing is independent of $\xi$. A path routed into block $k$ is absorbed in $s^0_k$ with probability one, so $\Prob(\kappa=k) = \hat{\pi}_k$ for every $k$. Given routing into block $k \le p$, $\tau = \xi + \tau_k$ with $\xi$ and $\tau_k$ independent and $\tau_k \sim \mathrm{PH}(\alpha_k, T_k)$ of mean $\mu_k$ and SCV $c_k^2$ by the block constructions; \Cref{prop:init-correction} then gives the conditional mean $\hat{\mu}_k$ and conditional SCV $\hat{c}_k^2$. Given routing directly to $s^0_k$ with $k > p$, $\tau = \xi \sim \mathrm{Exp}(\omega)$.
\end{proof}

\paragraph{Regularization.} \subchg{As constructed, $\alpha$ vanishes on $\cs_2$, so the phases of $\cs_2$ carry no initial mass and the M-step ratio $\hat\rho_{ik}=\sum_\ell b^\ell_{ik}/\sum_\ell b^\ell_i$ of \eqref{eq:approx-m-step} has a vanishing denominator there. We therefore regularize with a small constant $\epsilon>0$, replacing $\alpha$ by $(1-\epsilon)\,\alpha + \tfrac{\epsilon}{m}\mathbf{1}_m^\top$. After regularization every phase carries positive initial probability, which is all the M-step requires: the remaining denominators count expected visits to a phase over all causes, and are positive as soon as the phase is visited. That block $k$ leads only to $s^0_k$, so that $\rho_{ik}=0$ off the block, is harmless -- these are exactly the vanishing absorption probabilities admitted by \Cref{lem:cond-jump-probs}. The construction is exact at $\epsilon=0$ and its law is continuous in $\epsilon$, so a small value trades a small perturbation of the matched quantities of \Cref{prop:init-matching} for the regularity required by the EM iteration.}

\subsection{A worked example}
\label{sec:init-example}

We illustrate the moment-based initialization on the running $\text{MAPH}_{4,3}$ of \Cref{ex:running}, the ground truth of \Cref{sec:verification}. Drawing $L=5000$ competing-risks pairs from \eqref{eq:verif-truth} and forming the empirical summaries at the start of this appendix gives
\[
\hat\pi = (0.378,\,0.271,\,0.351),\qquad
\hat\mu = (0.534,\,0.625,\,0.750),\qquad
\hat c^2 = (1.19,\,0.95,\,0.88).
\]
The \subchg{most frequent} cause is the healthy discharge $k=1$, whose conditional SCV exceeds one, while the two further-care causes have conditional SCV below one. The thresholds of \eqref{eq:omega-threshold} are $\omega^* = (1.9,\,56.3,\,21.6)$, dominated by the type-2 cause whose $\hat c^2_2\approx0.95$ lies just below one -- the regime $\tfrac12\le\hat c^2<1$, in which $\omega^*_k$ grows without bound as $\hat c^2_k\uparrow1$ -- so the default $\omega = 10\max_k\omega^*_k \approx 563$.

With $m=5$ phases, at most $m-1=4$ are available for blocks. The corrected targets \eqref{eq:scv_calculated} are $c^2 = (1.20,\,0.95,\,0.89)$, so by \eqref{eq:num_phase} each cause requests a two-phase block: a hyper-exponential for $k=1$ (since $c^2_1>1$) and a hypo-exponential for each of $k=2,3$ (since $c^2<1$). Taking causes in order of decreasing empirical frequency -- $k=1$, then $k=3$, then $k=2$ -- the budget admits the first two blocks ($2+2=4$ phases) but not the third, so $p=2$: the healthy cause $k=1$ and the type-3 cause $k=3$ are moment matched, while the type-2 cause -- last in priority, and the one inflating $\omega$ -- absorbs directly out of the single front-end phase as $\mathrm{Exp}(\omega)$. The resulting (regularized) initialization is $\alpha = (0.99,\,0.00,\,0.00,\,0.00,\,0.00)$ together with
\[
{\small
T \approx \begin{pmatrix}
-562.80 & 13.64 & 198.98 & 197.43 & 0\\
0.02 & -0.87 & 0 & 0 & 0\\
0.02 & 0 & -2.06 & 0 & 0\\
0.02 & 0 & 0 & -1.44 & 1.42\\
0.02 & 0 & 0 & 0 & -22.48
\end{pmatrix},
\qquad
D \approx \begin{pmatrix}
0 & 152.74 & 0\\
0.86 & 0 & 0\\
2.05 & 0 & 0\\
0 & 0 & 0\\
0 & 0 & 22.47
\end{pmatrix}.}
\]
Here phase $1$ is the front end; phases $2$--$3$ form the hyper-exponential block for $k=1$ (two parallel exponentials, both absorbing into $s^0_1$); and phases $4$--$5$ form the bidiagonal hypo-exponential block for $k=3$ (absorbing into $s^0_3$ out of phase $5$). The entry $D_{12}\approx153 = \omega\hat\pi_2$ routes the front end directly to the unmatched cause $k=2$, and the small column-$1$ entries of $T$ are the feedback $\vartheta$ of the regularization. By \Cref{prop:init-matching} the induced law reproduces the absorption probabilities and matches the conditional mean and SCV of the two fitted causes -- numerically $\Prob(\kappa=k) = (0.378,\,0.271,\,0.350)$, $\ex{\tau\mid\kappa=k} = (0.54,\,0.007,\,0.74)$ and $\mathrm{SCV} = (1.19,\,154,\,0.89)$ -- while the conditional time of the unmatched cause $k=2$ is the front-end $\mathrm{Exp}(\omega)$, with near-zero mean and a large SCV.

This single initialization thus exercises both block families: the hyper-exponential, for the supra-unit-SCV \subchg{leading} cause, and the hypo-exponential, for a sub-unit-SCV cause. It also exposes two practical points. First, the phase budget binds: at the order $m=4$ actually used to fit this law in \Cref{sec:verification}, only \subchg{the first} block fits and just the cause $k=1$ is moment matched, so capturing a hypo-exponential block as well requires $m\ge5$. Second, a conditional SCV close to one (here $\hat c^2_2\approx0.95$) forces a large $\omega$ through \eqref{eq:omega-threshold}, which both leaves that cause poorly timed while it is unmatched and yields a stiff generator; the subsequent EM iterations correct the conditional times, but a less aggressive $\omega$ rule would be a worthwhile refinement.

%% file: sections/appendix-proofs-maph.tex
\section{Proofs: MAPH distributions}
\label{app:proofs-maph}
\label{app:proof-maph-props}

This appendix proves the results stated in \Cref{sec:maph-dist}.

In this appendix we prove \Cref{prop:maph-props} \aachg{and \Cref{prop:maph-censored}}.

Throughout, $\zeta=(\tau,\kappa)$ is the random pair induced by the MJP $\{X(t)\}_{t\ge0}$ on $\overline{\cs}=\cs\cup\cs^0$ with generator $Q$ parameterized by $(\alpha,T,D)$ as in \Cref{sec:maph-dist}, and $D_k$ denotes the $k$'th column of $D$. As established preceding \Cref{prop:maph-props}, $-T$ is a non-singular $M$-matrix, so $T^{-1}$ exists, $-T^{-1}\ge0$ entrywise, and every eigenvalue of $T$ has strictly negative real part; in particular $e^{Tu}\to 0$ as $u\to\infty$ and the integrals below converge.

\paragraph{Transient transition probabilities.}
For transient states $s_i,s_j\in\cs$ let
\[
P_{ij}(u) \;=\; \Prob\big(X(u)=s_j,\ \tau>u \;\big|\; X(0)=s_i\big)
\]
be the probability that the chain is in transient state $s_j$ at time $u$ and has not yet been absorbed, and collect these in the $m\times m$ matrix $P(u)=\big(P_{ij}(u)\big)$. Since the absorbing states are never left, the chain stays unabsorbed only by moving among the transient states, whose transition rates are the entries of $T$. The forward (Kolmogorov) equations of the chain, restricted to the transient states, are therefore $P'(u)=P(u)\,T$ with $P(0)=I$, whose unique solution is the matrix exponential
\begin{equation}\label{eq:transient-prob}
P(u)=e^{Tu}.
\end{equation}

\paragraph{(ii) Sub-density.}
The event $\{\tau\in[u,u+\mathrm{d}u),\ \kappa=k\}$ occurs when the chain is in some transient state $s_j$ just before $u$ and jumps to the absorbing state $s^0_k$ in $[u,u+\mathrm{d}u)$, an event of probability $\mu_{jk}\,\mathrm{d}u=D_{jk}\,\mathrm{d}u$. Conditioning on the initial distribution $\alpha$ and summing over $s_j$,
\[
f(u,k)\,\mathrm{d}u \;=\; \sum_{i=1}^m \alpha_i \sum_{j=1}^m P_{ij}(u)\, D_{jk}\,\mathrm{d}u
\;=\; \alpha\, e^{Tu} D_k\,\mathrm{d}u,
\]
using \eqref{eq:transient-prob}. Hence $f(u,k)=\alpha e^{Tu}D_k$.

\paragraph{(i) Sub-distribution and all-cause distribution functions.}
Integrating the sub-density and using that $T^{-1}$ and $e^{Tu}$ commute,
\[
F(u,k)=\int_0^u f(s,k)\,\mathrm{d}s
=\alpha\Big(\int_0^u e^{Ts}\,\mathrm{d}s\Big)D_k
=\alpha\, T^{-1}\big(e^{Tu}-I\big)D_k
=-\alpha\big(I-e^{Tu}\big)T^{-1}D_k.
\]
Summing over $k$ and using $D\mathbf{1}_n=-T\mathbf{1}_m$,
\[
F(u)=\sum_{k=1}^n F(u,k)
=-\alpha\big(I-e^{Tu}\big)T^{-1}D\mathbf{1}_n
=\alpha\big(I-e^{Tu}\big)\mathbf{1}_m
=1-\alpha e^{Tu}\mathbf{1}_m,
\]
where the last step uses $\alpha\mathbf{1}_m=1$.

\paragraph{(iii) Marginal absorption probability and conditional density.}
Letting $u\to\infty$ in (i) and using $e^{Tu}\to0$,
\[
\Prob(\kappa=k)=F(\infty,k)=-\alpha T^{-1}D_k,
\]
which is the $k$'th entry of $\alpha R$ with $R=-T^{-1}D$. \subchg{For a cause with $\Prob(\kappa=k)>0$}, the conditional density of $\tau$ given $\kappa=k$ is then $f(u\mid \kappa=k)=f(u,k)/\Prob(\kappa=k)$ by definition of conditional density.

\paragraph{(iv) Cause-specific hazard.}
In the competing-risks framework the cause-specific hazard for cause $k$ is
\[
H(u,k)=\lim_{\Delta\downarrow0}\frac{1}{\Delta}\,\Prob\big(\tau\in[u,u+\Delta),\ \kappa=k \;\big|\; \tau>u\big)
=\frac{f(u,k)}{\Prob(\tau>u)}=\frac{f(u,k)}{1-F(u)}.
\]
Note that the denominator is the all-cause survival function $1-F(u)=\alpha e^{Tu}\mathbf{1}_m$ from (i), \emph{not} the cause-specific tail. Substituting (ii) gives $H(u,k)=\dfrac{\alpha e^{Tu}D_k}{\alpha e^{Tu}\mathbf{1}_m}$.

\paragraph{(v) Laplace--Stieltjes transform.}
For $s\ge0$ every eigenvalue of $T-sI$ still has strictly negative real part (subtracting $s$ shifts the eigenvalues of $T$ further left), so the integral below converges and
\[
\phi(s,k)=\int_0^\infty e^{-su} f(u,k)\,\mathrm{d}u
=\alpha\Big(\int_0^\infty e^{(T-sI)u}\,\mathrm{d}u\Big)D_k
=-\alpha\,(T-sI)^{-1}D_k
=\alpha\,(sI-T)^{-1}D_k.
\]

\paragraph{(vi) Moment generating function.}
For $s$ in a neighbourhood of $0$ (small enough that every eigenvalue of $T+sI$ still has strictly negative real part), the same computation with $-s$ in place of $s$ gives
\[
M(s,k)=\int_0^\infty e^{su} f(u,k)\,\mathrm{d}u
=-\alpha\,(sI+T)^{-1}D_k.
\]

\paragraph{(vii) Partial moments.}
Using $\int_0^\infty u^j e^{Tu}\,\mathrm{d}u = j!\,(-T)^{-(j+1)}=(-1)^{j+1}j!\,T^{-(j+1)}$, valid since every eigenvalue of $T$ has strictly negative real part,
\[
\mathcal{M}_{j,k}=\int_0^\infty u^j f(u,k)\,\mathrm{d}u
=\alpha\Big(\int_0^\infty u^j e^{Tu}\,\mathrm{d}u\Big)D_k
=(-1)^{j+1}\,j!\,\alpha\,T^{-(j+1)}D_k.
\]
This completes the proof. \qed

\begin{proof}[Proof of \Cref{lem:cond-jump-probs}]
Fix $i \neq j$. Conditional probabilities given $\{I_0 = s_i\}$ are understood with respect to the
jump chain started at $s_i$, so that no assumption on $\alpha_i$ is required; the conditioning event
of \textup{(i)} then has probability $\rho_{ik} > 0$. Under this proviso,
\begin{align*}
p^\lambda_{ij|k}=\mathbb{P}(I_{1}=s_j\vert I_0=s_i, I_{\infty}=s_k^0) =&\frac{\mathbb{P}(I_{1}=s_j, I_0=s_i, I_{\infty}=s^0_k)}{\mathbb{P}(I_0=s_i, I_{\infty}=s_k^0)}\\
=&\frac{\mathbb{P}(I_0=s_i)\mathbb{P}(I_{1}=s_j\vert I_0=s_i)\mathbb{P}(I_{\infty}=s_k^0\vert I_{1}=s_j,I_0=s_i)}{\mathbb{P}(I_0=s_i)\mathbb{P}(I_{\infty}=s^0_k\vert I_0=s_i)}\\
=& \mathbb{P}(I_{1}=s_j\vert I_0=s_i) \frac{\mathbb{P}(I_{\infty}=s^0_k\vert I_{0}=s_j)}{\mathbb{P}(I_{\infty}=s^0_k\vert I_0=s_i)},
\end{align*}
where the last equality uses the Markov property and time homogeneity of $\{I_j\}_{j=0}^\infty$ to
replace $\mathbb{P}(I_{\infty}=s_k^0\vert I_{1}=s_j,I_0=s_i)$ by $\mathbb{P}(I_{\infty}=s_k^0\vert
I_{0}=s_j)$. With the notation \eqref{eq:transition_prob} and \eqref{eq:rho_ik} this is
\eqref{eq:main-pijgivenk-equation}, proving \textup{(i)}. Note that no positivity of $\rho_{jk}$ is
used here, only $\rho_{ik}>0$.

For \textup{(ii)}, if $\rho_{jk}>0$ the factor $\rho_{jk}/\rho_{ik}$ in
\eqref{eq:main-pijgivenk-equation} is strictly positive and may be divided out, which gives
\eqref{eq:pij-from-pijk-cond}. That identity expresses the single quantity $p^\lambda_{ij}$, which
carries no dependence on $k$, so equating its right hand side for two indices $k$ and $l$ with
$\rho_{ik},\rho_{jk},\rho_{il},\rho_{jl}>0$ yields \eqref{eq:pijk-consistency}.

For \textup{(iii)}, setting $\rho_{jk}=0$ in \eqref{eq:main-pijgivenk-equation} gives
$p^\lambda_{ij|k}=0$ irrespective of the value of $p^\lambda_{ij}$, so the map
$p^\lambda_{ij} \mapsto p^\lambda_{ij|k}$ is constant and cannot be inverted; that $p^\lambda_{ij}>0$
is compatible with $\rho_{jk}=0$ and $\rho_{ik}>0$ is seen from the row-wise form of
\eqref{eq:r-matrix}, namely $\rho_{ik} = \sum_{l=1}^m p^\lambda_{il}\rho_{lk} + p^\mu_{ik}$, in which
$\rho_{ik}>0$ may be produced by terms other than $j$.

For \textup{(iv)}, when $\rho_{ik}=0$ the event $\{I_0=s_i, I_\infty=s^0_k\}$ has probability
$\rho_{ik}=0$ for the chain started at $s_i$ (and $\alpha_i\rho_{ik}=0$ for a general initial
distribution), so the conditional probability defining $p^\lambda_{ij|k}$ is a ratio with zero
denominator and is undefined; the convention $p^\lambda_{ij|k}:=0$ is consistent with \textup{(iii)}, since $s^0_k$ is
then not absorbed from $s_i$ at all. The stated consequences follow from the same row-wise identity
$\rho_{ik} = \sum_{l=1}^m p^\lambda_{il}\rho_{lk} + p^\mu_{ik}$: all summands are non-negative, so
$\rho_{ik}=0$ forces $p^\mu_{ik}=0$ and $p^\lambda_{il}\rho_{lk}=0$ for every $l$, whence
$\rho_{lk}=0$ whenever $p^\lambda_{il}>0$.

For \textup{(v)}, split the sum according to the sign of $\rho_{ik}$. Indices with $\rho_{ik}=0$
contribute $0$ by the convention of \textup{(iv)}. For the remaining indices \textup{(i)} applies, so
\[
\sum_{k=1}^n \rho_{ik}\, p^\lambda_{ij|k}
= \sum_{k \,:\, \rho_{ik}>0} \rho_{ik}\, p^\lambda_{ij}\, \frac{\rho_{jk}}{\rho_{ik}}
= p^\lambda_{ij} \sum_{k \,:\, \rho_{ik}>0} \rho_{jk}.
\]
If $p^\lambda_{ij}=0$ both sides vanish. If $p^\lambda_{ij}>0$ then, by \textup{(iv)}, $\rho_{ik}=0$
forces $\rho_{jk}=0$, so the omitted indices contribute nothing to $\sum_k \rho_{jk}$, which equals
$1$ by the row-stochasticity of $R$ established after \eqref{eq:r-matrix}. In either case the sum is
$p^\lambda_{ij}$. Probabilistically, \eqref{eq:pijk-aggregation} is the law of total probability for
the first jump, decomposed over the cause of absorption.

Finally, for $i=j$ we have $p^\lambda_{ii}=0$ by \eqref{eq:transition_prob}, and hence
$p^\lambda_{ii|k}=0$ by \eqref{eq:main-pijgivenk-equation} when $\rho_{ik}>0$, and by the convention
of \textup{(iv)} otherwise.
\end{proof}

\begin{proof}[Proof of \Cref{lem:feasible-jump}]
For \textup{(i)}, sum the definition of $p^\mu_{ik}$ over $k$ and use $\sum_k \rho_{ik}=1$ and
$\sum_k \rho_{jk}=1$:
\[
\sum_{k=1}^n p^\mu_{ik}
= \sum_{k=1}^n \rho_{ik} - \sum_{j=1}^m p^\lambda_{ij} \sum_{k=1}^n \rho_{jk}
= 1 - \sum_{j=1}^m p^\lambda_{ij}.
\]
For \textup{(ii)}, the inequality $p^\mu_{ik}\ge0$ is by definition
$\sum_j \rho_{jk}p^\lambda_{ij} \le \rho_{ik}$, which is
\eqref{eq:second-param-main-constraint}. Under it $p^\mu_{ik} \le \rho_{ik} \le 1$, since the
subtracted sum is non-negative, and \textup{(i)} then says that row $i$ of
$[\,P^\lambda \mid P^\mu\,]$ is a probability vector; in particular
$\sum_j p^\lambda_{ij} \le 1$. Part \textup{(iii)} is immediate from \textup{(i)}: the $n$
constraints of row $i$ are equalities exactly when $p^\mu_{ik}=0$ for every $k$, which by
\textup{(i)} is exactly $\sum_j p^\lambda_{ij}=1$. For \textup{(iv)}, the $P^\mu$ defined above is the one-step absorption matrix by
\eqref{eq:first-step-1}, and the recoveries are parts \textup{(i)}, \textup{(iv)} and \textup{(v)} of
\Cref{lem:cond-jump-probs}.
\end{proof}

%% file: sections/appendix-proofs-est.tex
\section{Proofs: estimation and the EM algorithm}
\label{app:proofs-est}
\label{app:e-step}
\label{app:e-step-exact}

This appendix proves the results stated in \Cref{sec:full-path-mle}.

\begin{proof}[Proof of \Cref{prop:path-likelihood}]
For \textup{(i)}, fix a path and suppress $\ell$; let it have $n_0$ jumps through the states
$I_0,\ldots,I_{n_0}$ of indices $i_0,\ldots,i_{n_0-1}$, absorbed at $I_{n_0}=s^0_{k_0}$, at times
$0=\sigma_0<\sigma_1<\cdots<\sigma_{n_0}$. Conditionally on the embedded chain $\{I_n\}$, the sojourn
in the $n$-th visited state is exponential with rate $q_{i_{n-1}}$, independently across $n$, and the
jump chain moves with the probabilities \eqref{eq:transition_prob}. Multiplying the initial
probability $\alpha_{i_0}$, the $n_0$ sojourn densities, the $n_0-1$ transient jump probabilities and
the final exit probability $p^\mu_{i_{n_0-1}k_0}$ gives
\begin{equation}
\label{eq:path-likelihood-raw}
f(X^\ell~;~\theta)
=
\alpha_{i_0}
\Big(
\prod_{n=1}^{n_0} q_{i_{n-1}} e^{-q_{i_{n-1}}(\sigma_n-\sigma_{n-1})}
\Big)
\Big(
\prod_{n=1}^{n_0-1} p^\lambda_{i_{n-1} i_n}
\Big)
p^\mu_{i_{n_0-1}\,k_0}.
\end{equation}
No conditioning on the cause is involved -- the path determines it -- which is why the probabilities
appearing here are the unconditional ones of \eqref{eq:transition_prob}.

To express this through the statistics, partition the sojourn product by visited phase,
\[
\prod_{n=1}^{n_{0}} q_{i_{n-1}}e^{-q_{i_{n-1}}(\sigma_n-\sigma_{n-1})} = \prod_{i=1}^m ~~\prod_{n \in \{1,\ldots,n_0\} ~:~ i_{n-1} = i} ~~q_{i}e^{-q_{i}(\sigma_n-\sigma_{n-1})}
=
\prod_{i=1}^m
q_i^{N_i^\ell} e^{-q_i Z_i^\ell},
\]
and regroup the jump probabilities by transition type, the factor $p^\lambda_{ij}$ occurring
$M^\ell_{ij}$ times. Exactly one $B^\ell_i$ equals one, so $\alpha_{i_0}=\prod_i\alpha_i^{B^\ell_i}$;
and $E^\ell_{ik}$ vanishes unless $i=i_{n_0-1}$ and $k=k_0$, so
$p^\mu_{i_{n_0-1}k_0}=\prod_i\prod_k (p^\mu_{ik})^{E^\ell_{ik}}$, every other factor of that product
being $1$. This is \eqref{eq:path-likelihood-general}, in which no free cause index remains.

For \textup{(ii)}, the paths are independent, so ${\cal L}(\theta)$ is the product of the per-path
densities. Taking logarithms and summing over $\ell$, each per-path exponent aggregates into its
per-sample counterpart, $B_i=\sum_\ell B^\ell_i$, $N_i=\sum_\ell N^\ell_i$, $Z_i=\sum_\ell Z^\ell_i$,
$M_{ij}=\sum_\ell M^\ell_{ij}$ and $E_{ik}=\sum_\ell E^\ell_{ik}$.
\end{proof}

\begin{proof}[Proof of \Cref{lem:cause-specific}]
If a path visits $s_i$ and is absorbed in $s^0_k$, absorption in $s^0_k$ from $s_i$ has positive
probability, so $\rho_{ik}>0$; paths visiting a phase with $\rho_{ik}=0$ and absorbed in $s^0_k$ form
a null set. Case \textup{(i)} of \Cref{lem:cond-jump-probs} therefore applies at every visited phase.

Start from \eqref{eq:path-likelihood-general} and replace the jump and exit
probabilities by their conditional counterparts, keeping track of the factors this introduces. The
initial and sojourn groups contain no jump probability and are left alone. For the jump group, by
$p^\lambda_{ij|k}=p^\lambda_{ij}\,\rho_{jk}/\rho_{ik}$,
\[
\prod_{i=1}^m\prod_{j=1}^m \big(p^\lambda_{ij|k}\big)^{M^\ell_{ij}}
=
\Big( \prod_{i=1}^m\prod_{j=1}^m \big(p^\lambda_{ij}\big)^{M^\ell_{ij}} \Big)
\prod_{i=1}^m\prod_{j=1}^m \Big( \frac{\rho_{jk}}{\rho_{ik}} \Big)^{M^\ell_{ij}} ,
\]
and the trailing factor collects, after renaming the dummy index in the numerator, into
\[
\prod_{i=1}^m\prod_{j=1}^m \Big( \frac{\rho_{jk}}{\rho_{ik}} \Big)^{M^\ell_{ij}}
=
\prod_{j=1}^m \rho_{jk}^{\;\sum_{i} M^\ell_{ij}}
\cdot
\prod_{i=1}^m \rho_{ik}^{\;-\sum_{j} M^\ell_{ij}}
=
\prod_{i=1}^m \rho_{ik}^{\;\sum_{j} M^\ell_{ji} \,-\, \sum_{j} M^\ell_{ij}} .
\]
For the exit group, by $p^\mu_{ik|k}=p^\mu_{ik}/\rho_{ik}$,
\[
\prod_{i=1}^m \big(p^\mu_{ik|k}\big)^{E^\ell_{ik}}
=
\Big( \prod_{i=1}^m \big(p^\mu_{ik}\big)^{E^\ell_{ik}} \Big)
\prod_{i=1}^m \rho_{ik}^{\,-E^\ell_{ik}} .
\]
Multiplying the two trailing factors, the exponent of $\rho_{ik}$ is
\[
\underbrace{\sum_{j=1}^m M^\ell_{ji}}_{=\,N^\ell_i - B^\ell_i}
\;-\;
\underbrace{\Big( \sum_{j=1}^m M^\ell_{ij} + E^\ell_{ik} \Big)}_{=\,N^\ell_i}
\;=\;
\big( N^\ell_i - B^\ell_i \big) - N^\ell_i
\;=\;
- B^\ell_i ,
\]
the two braces being the two halves of the flow balance \eqref{eq:flow-balance}: the left counts the
entries into $s_i$ other than the start of the path, the right counts all the exits from $s_i$, jumps
to another transient phase together with the final jump to absorption. The exit term must be grouped
with the jump term, since it is the pair of them that forms one side of the balance.

Splitting the initial group of \eqref{eq:path-likelihood-cause} as
$\prod_i (\alpha_i\rho_{ik})^{B^\ell_i} = \big(\prod_i \alpha_i^{B^\ell_i}\big)\prod_i
\rho_{ik}^{B^\ell_i}$ and substituting the above,
\begin{align*}
&\Big( \prod_{i=1}^m (\alpha_i \rho_{ik})^{B^\ell_i} \Big)
\Big( \prod_{i=1}^m q_i^{N^\ell_i} e^{-q_i Z^\ell_i} \Big)
\Big( \prod_{i,j} \big(p^\lambda_{ij|k}\big)^{M^\ell_{ij}} \Big)
\Big( \prod_{i} \big(p^\mu_{ik|k}\big)^{E^\ell_{ik}} \Big)
\\[2pt]
&\qquad =
\Big( \prod_{i=1}^m \alpha_i^{B^\ell_i} \Big)
\Big( \prod_{i=1}^m \rho_{ik}^{B^\ell_i} \Big)
\Big( \prod_{i=1}^m q_i^{N^\ell_i} e^{-q_i Z^\ell_i} \Big)
\Big( \prod_{i,j} \big(p^\lambda_{ij}\big)^{M^\ell_{ij}} \Big)
\Big( \prod_{i} \big(p^\mu_{ik}\big)^{E^\ell_{ik}} \Big)
\Big( \prod_{i=1}^m \rho_{ik}^{-B^\ell_i} \Big)
\\[2pt]
&\qquad =
\Big( \prod_{i=1}^m \alpha_i^{B^\ell_i} \Big)
\Big( \prod_{i=1}^m q_i^{N^\ell_i} e^{-q_i Z^\ell_i} \Big)
\Big( \prod_{i=1}^m \prod_{j=1}^m \big(p^\lambda_{ij}\big)^{M^\ell_{ij}} \Big)
\Big( \prod_{i=1}^m \prod_{k'=1}^n \big(p^\mu_{ik'}\big)^{E^\ell_{ik'}} \Big) ,
\end{align*}
the second equality cancelling $\prod_i \rho_{ik}^{B^\ell_i}$ against $\prod_i \rho_{ik}^{-B^\ell_i}$
and, in the last group, restoring the product over all causes,
\[
\prod_{i=1}^m \big(p^\mu_{ik}\big)^{E^\ell_{ik}}
\;=\;
\prod_{i=1}^m \prod_{k'=1}^n \big(p^\mu_{ik'}\big)^{E^\ell_{ik'}} ,
\qquad
\text{since } E^\ell_{ik'}=0 \text{ and so } \big(p^\mu_{ik'}\big)^{E^\ell_{ik'}}=1
\text{ for every } k'\neq k .
\]
The expression obtained is the right hand side of \eqref{eq:path-likelihood-general}, which by
\Cref{prop:path-likelihood}\,\textup{(i)} is $f(X^\ell~;~\theta)$. The two factorizations therefore
define the same density.
\end{proof}

\begin{proof}[Proof of \Cref{prop:approx-mle}]
\emph{The objective separates.} The set $\Theta$ is a product: $\alpha$ ranges over the simplex in
$\R^m$, each $q_i$ over $(0,\infty)$, and, for each $i$, the $i$-th row
$\big((p^\lambda_{ij})_{j\neq i},(p^\mu_{ik})_k\big)$ over the simplex in $\R^{m-1+n}$.
Correspondingly \eqref{eq:complete-log-like-jump} splits as
\[
\log{\cal L}(\theta)
=
\sum_{i=1}^m B_i\log\alpha_i
\;+\;
\sum_{i=1}^m \big( N_i\log q_i - q_i Z_i \big)
\;+\;
\sum_{i=1}^m \Big( \sum_{j \neq i} M_{ij}\log p^\lambda_{ij} + \sum_{k=1}^n E_{ik}\log p^\mu_{ik} \Big),
\]
the term $j=i$ being absent because $p^\lambda_{ii}=0$ and $M_{ii}=0$. No variable occurs in two of
these groups and they vary independently over $\Theta$, so each may be maximized on its own. Each is
concave, and strictly concave in those coordinates whose count is positive; a coordinate whose count
vanishes contributes $0\log 0 = 0$ whatever its value, and the maximum places no mass on it, since
that mass would otherwise be taken from a coordinate with a positive count. It is therefore enough,
in each group, to locate the stationary point subject to the constraint.

\emph{The initial distribution.} Maximize $\sum_i B_i\log\alpha_i$ subject to $\sum_i \alpha_i=1$.
With a multiplier $\gamma$ for the constraint, the stationary conditions are
$B_i/\alpha_i = \gamma$, that is $\alpha_i = B_i/\gamma$. Summing over $i$ and using
$\sum_i \alpha_i = 1$ gives $\gamma = \sum_i B_i$, and each path starts in exactly one phase, so
$\sum_i B_i = L$. Hence
\[
\hat\alpha_i = \frac{B_i}{L}.
\]

\emph{The exit rates.} For each $i$, maximize $N_i\log q_i - q_i Z_i$ over $q_i>0$. Its derivative is
$N_i/q_i - Z_i$ and its second derivative $-N_i/q_i^2<0$, so it is strictly concave with the single
stationary point
\[
\hat q_i = \frac{N_i}{Z_i},
\]
where $N_i>0$ by hypothesis and $Z_i>0$ because a visited phase is occupied for a positive length of
time.

\emph{The jump and exit probabilities.} Fix $i$ and maximize
$\sum_{j\neq i} M_{ij}\log p^\lambda_{ij} + \sum_k E_{ik}\log p^\mu_{ik}$ subject to
$\sum_{j \neq i} p^\lambda_{ij} + \sum_k p^\mu_{ik} = 1$. With a multiplier $\gamma_i$, the
stationary conditions are
\[
\frac{M_{ij}}{p^\lambda_{ij}} = \gamma_i \quad (j \neq i),
\qquad
\frac{E_{ik}}{p^\mu_{ik}} = \gamma_i ,
\qquad\text{that is}\qquad
p^\lambda_{ij} = \frac{M_{ij}}{\gamma_i},
\quad
p^\mu_{ik} = \frac{E_{ik}}{\gamma_i} .
\]
Summing these over $j \neq i$ and over $k$ and using the constraint,
\[
1 = \frac{1}{\gamma_i}\Big( \sum_{j \neq i} M_{ij} + \sum_{k=1}^n E_{ik} \Big) = \frac{N_i}{\gamma_i},
\]
the last equality being the flow balance \eqref{eq:flow-balance}. Hence $\gamma_i = N_i$ and
\[
\hat p^\lambda_{ij} = \frac{M_{ij}}{N_i} \quad (j \neq i),
\qquad
\hat p^\mu_{ik} = \frac{E_{ik}}{N_i},
\]
which completes \eqref{eq:mle}.

\emph{Feasibility.} All the ratios in \eqref{eq:mle} are non-negative, and by the flow balance the
$i$-th row sums to $\big(\sum_{j\neq i} M_{ij} + \sum_k E_{ik}\big)/N_i = 1$, so
$\hat\theta\in\Theta$ with no adjustment. The derivations above used only concavity and the
constraints, never that the statistics are integers, which gives the second half of \textup{(ii)}.

\emph{Non-degeneracy, proving \textup{(i)}.} The matrix $\widehat P^\lambda$ is sub-stochastic, its
$i$-th row summing to $1-\sum_k \hat p^\mu_{ik}$, so any phase with $\sum_k E_{ik}>0$ has row sum
strictly below one. For a phase with $\sum_k E_{ik}=0$, the counts $M_{ij}$ record jumps out of $s_i$
on paths that are absorbed, so following indices with $\hat p^\lambda_{ij}>0$ traces path segments
and reaches, in finitely many steps, a phase from which absorption was recorded. Hence
$\widehat P^\lambda$ is transient and $I-\widehat P^\lambda$ is non-singular. Since
$-\widehat T = \mathrm{Diag}(\hat q)(I-\widehat P^\lambda)$ with $\hat q>0$, the matrix $-\widehat T$
is a non-singular $M$-matrix, so absorption is certain and $\widehat R$ of \eqref{eq:R-from-mle} is
well defined and, by the computation following \eqref{eq:r-matrix}, non-negative and row-stochastic.
Finally $\widehat P^\mu = \widehat R - \widehat P^\lambda \widehat R \ge 0$ by construction, which is
the bound \eqref{eq:R-bound} by \Cref{lem:feasible-jump}\,\textup{(ii)}.
\end{proof}

\begin{proof}[Proof of \Cref{cor:monotone}]
The first claim is \Cref{prop:approx-mle} applied to the expected statistics, legitimate by
\textup{(ii)} of that proposition, together with the linearity of
\eqref{eq:complete-log-like-jump} in the statistics, which makes $Q(\theta\mid\theta^{(\nu)})$ that
same expression evaluated at the conditional expectations. The second is the standard EM ascent
property: $\ell_{\mathrm{obs}}(\theta)=Q(\theta\mid\theta^{(\nu)})-H(\theta\mid\theta^{(\nu)})$ with
$H(\theta\mid\theta^{(\nu)})=\ex{\log p(\text{path}\mid{\cal D},\theta) \mid {\cal D}}$ maximized over
$\theta$ at $\theta=\theta^{(\nu)}$ by Jensen's inequality, so increasing $Q$ cannot decrease $\ell$.
\end{proof}

\subsection{The exact-event E-step}

This appendix proves the E-step expectations of the algorithm of \Cref{sec:em-alg}: the conditional expectations, given a single observation $(\tau,\kappa)=(t,k)$, of the complete-data sufficient statistics of \Cref{tab:stats}. We use the notation of \Cref{sec:maph-dist}: $T$ is the sub-generator with off-diagonal rates $T_{ij}=\lambda_{ij}\ge0$ ($i\neq j$) and diagonal entries $-q_i$, the matrix $D$ has entries $D_{ik}=\mu_{ik}\ge0$, $D_k$ denotes the $k$-th column of $D$, and $\alpha$ is the initial (row) distribution.

The statement of \Cref{prop:e-step} uses the density $f(t,k)$ and the matrix $\mathbf C(t,k)$ defined in \eqref{eq:estep-fC} of \Cref{sec:em-exact}.

\subchg{Fix a cause $k$ that can occur, that is $\Prob(\kappa=k)>0$, which by item (iii) of \Cref{prop:maph-props} is $-\alpha T^{-1}D_k>0$; this holds for every cause represented in the data. Then $f(t,k)>0$ for every $t>0$: writing $f(t,k)=\sum_{i,j}\alpha_i (e^{Tt})_{ij}\,\mu_{jk}$ and noting that $(e^{Tt})_{ij}>0$ for $t>0$ exactly when $s_j$ is reachable from $s_i$ within $\cs$, the sum is positive precisely when some phase in the support of $\alpha$ reaches cause $k$. Consequently} the regular conditional distribution given $\{\tau=t,\kappa=k\}$ is well defined and the conditional expectations below exist.

\begin{proof}[Proof of \Cref{prop:e-step}]
Write $f=f(t,k)$, and recall the cause-specific path density $\Prob(\tau\in dt,\kappa=k\mid X(0)=s_i)=\big(e^{Tt}D_k\big)_i\,dt$, so that $f\,dt = \alpha e^{Tt}D_k\,dt = \Prob(\tau\in dt,\kappa=k)$.

For the initial-state indicator $B_i=\mathbf 1\{X(0)=s_i\}$,
\[
\ex{B_i \mid \tau=t,\ \kappa=k}=\Prob(X(0)=s_i\mid \tau=t,\kappa=k)
=\frac{\alpha_i\,\big(e^{Tt}D_k\big)_i}{f}.
\]
For the occupation time $Z_i=\int_0^\infty \mathbf 1\{X(u)=s_i\}\,du$, only $u\le t$ contributes; the integrand is non-negative, so Tonelli's theorem lets us exchange expectation and integration, and by the Markov property $\Prob(X(u)=s_i,\tau\in dt,\kappa=k)=\big(\alpha e^{Tu}\big)_i\,\big(e^{T(t-u)}D_k\big)_i\,dt$, whence
\[
\ex{Z_i \mid \tau=t,\ \kappa=k}=\int_0^t \frac{\big(\alpha e^{Tu}\big)_i\,\big(e^{T(t-u)}D_k\big)_i}{f}\,du=\frac{C_{ii}(t,k)}{f}.
\]
For the transient jump counts $M_{ij}$ with $i\neq j$, discretize time in steps of $\varepsilon$ and use $\Prob\big(X((\ell{+}1)\varepsilon)=s_j\mid X(\ell\varepsilon)=s_i\big)=\big(e^{T\varepsilon}\big)_{ij}=T_{ij}\varepsilon+o(\varepsilon)$. Summing the contributions $\big(\alpha e^{T\ell\varepsilon}\big)_i\,\big(e^{T\varepsilon}\big)_{ij}\,\big(e^{T(t-(\ell+1)\varepsilon)}D_k\big)_j/f$ over $\ell$ and letting $\varepsilon\downarrow0$ yields a Riemann sum converging to
\[
\ex{M_{ij} \mid \tau=t,\ \kappa=k}=\frac{T_{ij}\displaystyle\int_0^t \big(\alpha e^{Tu}\big)_i\,\big(e^{T(t-u)}D_k\big)_j\,du}{f}=\frac{T_{ij}\,C_{ij}(t,k)}{f}.
\]
The exchange of the $\varepsilon\downarrow0$ limit with the conditional expectation is justified by dominated convergence: the discretized counts are bounded by the total number of transient jumps before absorption, which has finite conditional expectation, and converge pointwise as $\varepsilon\downarrow0$, while $(e^{T\varepsilon}-I)/\varepsilon\to T$ and the continuity of $u\mapsto e^{Tu}$ identify the limit (cf.\ \citet{Asmussen1996FittingPD}).
Finally $E_{ik}$ counts the single absorbing jump $s_i\to s_k^0$, which occurs at the terminal time $t$ (so no occupation integral arises); the same discretization gives $\Prob\big(X(t-\varepsilon)=s_i,\,\tau\in dt,\,\kappa=k\big)=\big(\alpha e^{T(t-\varepsilon)}\big)_i\,D_{ik}\,dt\,\big(1+o(1)\big)$, and letting $\varepsilon\downarrow0$,
\[
\ex{E_{ik} \mid \tau=t,\ \kappa=k}=\frac{\big(\alpha e^{Tt}\big)_i\,D_{ik}}{f}.
\]
Equation \eqref{eq:e-step-N-total} follows because, conditional on $\kappa=k$, the exits from $s_i$ are its transient jumps $\sum_j M_{ij}$ together with the absorbing jump counted by $E_{ik}$, which is the right hand equality of the flow balance \eqref{eq:flow-balance}.
\end{proof}

\paragraph{Relation to the phase-type case.} When $n=1$ the single column $D_1=-T\mathbf 1_m$ is the exit vector of an ordinary phase-type distribution and $\kappa$ is degenerate; \eqref{eq:cond-exp-estats} then reduces to the classical EM E-step for phase-type distributions \citep{Asmussen1996FittingPD}, \subchg{its four entries being the conditional expectations of the initial-state indicator, the occupation time, the transient jump counts and the absorption count}. The MAPH E-step is the cause-resolved generalization in which the single exit vector $-T\mathbf 1_m$ is replaced by the observed column $D_k$.

\paragraph{Computation.} For each observation, $e^{Tt}$ is obtained from any matrix-exponential routine, and all entries of $\mathbf C(t,k)$ from a single $2m\times2m$ exponential through the block-triangular identity of \citet{vanloan_1978},
\begin{equation}
\label{eq:vanloan-block}
\exp\!\left( t\begin{bmatrix} T^\intercal & \alpha^\intercal D_k^\intercal \\[2pt] 0 & T^\intercal \end{bmatrix}\right)
=\begin{bmatrix} e^{T^\intercal t} & \mathbf C(t,k) \\[2pt] 0 & e^{T^\intercal t}\end{bmatrix},
\end{equation}
where $\alpha^\intercal D_k^\intercal$ is the rank-one $m\times m$ matrix with entries $(\alpha^\intercal D_k^\intercal)_{ij}=\alpha_i D_{jk}$, so that $\mathbf C(t,k)$ is read off the top-right block. Observations sharing the same absorbing state $k$ reuse the same column $D_k$.

\paragraph{Alternative evaluations.} The block-triangular identity \eqref{eq:vanloan-block} is only one way to evaluate these conditional expectations. The same expected occupation times and expected transition counts can be obtained by \emph{uniformization} (randomization), which represents the transient chain through a Poisson-subordinated discrete-time chain and is a standard tool for transient Markov-chain analysis \citep{grassmann_1977, reibman_trivedi_1988}; this is the route taken in the refined phase-type EM of \citet{okamura_2011}. Alternatively, the integrals $C_{ij}(t,k)$ satisfy a linear system of ordinary differential equations that can be integrated directly by a Runge--Kutta scheme, as in the original phase-type EM of \citet{Asmussen1996FittingPD}. For large $m$, forming the full $2m\times2m$ exponential can be avoided altogether by computing the \emph{action} of the matrix exponential on a vector through Krylov-subspace or scaling-and-squaring methods \citep{sidje_1998, almohy_higham_2011}. We adopt the Van~Loan identity \eqref{eq:vanloan-block} for its simplicity at the modest phase counts of our examples; the alternatives above may be preferable at larger scale, and adding them to the companion software is left as future work.

%% file: sections/appendix-proofs-cens.tex
\section{Proofs: adaptations for right censoring}
\label{app:proofs-cens}
\label{app:e-step-censored}

This appendix proves the results stated in \Cref{sec:censoring}.

\begin{aarevision}
\paragraph{Proof of \Cref{prop:maph-censored}.}
By item (i) of \Cref{prop:maph-props}, $S(c)=\Prob(\tau>c)=1-F(c)=\alpha e^{Tc}\mathbf{1}_m$, which gives (i); it is strictly positive because $e^{Tc}$ is the transition matrix of the transient chain over $[0,c]$ and $\alpha$ is a probability vector on $\cs$, so absorption cannot be certain by any finite time.

For (iii), let $\alpha_c$ denote the conditional distribution of the phase occupied at time $c$ given $\{\tau>c\}$. Since $\Prob(X(c)=s_i,\ \tau>c)=(\alpha e^{Tc})_i$, we have
\[
\alpha_c=\frac{\alpha e^{Tc}}{S(c)},
\]
which is a probability vector on $\cs$ by (i). The MJP $\{X(t)\}_{t\ge0}$ is time-homogeneous and Markov, and $\{\tau>c\}=\{X(c)\in\cs\}$ is measurable with respect to $X(c)$; hence, conditionally on $\{\tau>c\}$, the post-censoring process $\{X(c+u)\}_{u\ge0}$ is again an MJP with the same generator $Q$, started from $\alpha_c$. Its absorption time is $\tau-c$ and its absorbing state is $s^0_\kappa$, so $(\tau-c,\kappa)$ given $\{\tau>c\}$ is the MAPH pair of $\text{MAPH}_{m,n}(\alpha_c,T,D)$. Applying item (ii) of \Cref{prop:maph-props} to this law gives the residual sub-density $f(u,k\mid\tau>c)=\alpha_c e^{Tu}D_k$.

For (ii), apply item (iii) of \Cref{prop:maph-props} to the law of (iii):
\[
\Prob(\kappa=k\mid\tau>c)=-\alpha_cT^{-1}D_k=\frac{-\alpha e^{Tc}T^{-1}D_k}{S(c)}.
\]
Summing over $k$ and using $-T^{-1}D\mathbf{1}_n=-T^{-1}(-T\mathbf{1}_m)=\mathbf{1}_m$, which is \eqref{eq:D-matrix-constraint}, gives $\sum_{k=1}^n\Prob(\kappa=k\mid\tau>c)=\alpha_c\mathbf{1}_m=1$. \qed
\end{aarevision}

For a subject right censored at time $c>0$, the observation is the event $\{\tau>c\}$; the eventual cause is unobserved. We retain the complete-data convention of \Cref{sec:full-path-mle}: the latent path is completed through its eventual absorption. \subchg{Hence the E-step must complete the path on both sides of $c$ and average over the unobserved eventual cause. Only the exit statistic $E_{ik}$ retains a cause index; the remaining statistics enter the M-step \eqref{eq:approx-m-step} aggregated over causes, and are computed here in that form.}

The statement of \Cref{prop:e-step-censored} uses the survival function $S(c)$ and the
quantities $a(c)$, $g(c)$ and $C(c;b)$ of \eqref{eq:censored-definitions}--\eqref{eq:censored-Cb}, defined in \Cref{sec:em-censored}.

\begin{proof}[Proof of \Cref{prop:e-step-censored}]
\subchg{We have $S(c)=\alpha e^{Tc}\mathbf{1}_m>0$ for every finite $c$, since $e^{Tc}$ has strictly positive diagonal entries and $\alpha$ is a probability vector.} For a non-negative complete-path statistic $H$, conditioning on $\{\tau>c\}$ and applying the tower property with respect to $(\tau,\kappa)$ gives
\begin{equation}
\label{eq:tail-mixture}
\ex{H\mid\tau>c}
 =\frac{\ex{H\ind{\tau>c}}}{S(c)}
 =\frac{1}{S(c)}\sum_{k=1}^n\int_c^\infty
   \ex{H\mid\tau=t,\kappa=k}\,f(t,k)\,dt ,
\end{equation}
the last step integrating against $\Prob(\tau\in dt,\kappa=k)=f(t,k)\,dt$ and summing over the eventual cause, which every completed path has exactly one of. \subchg{Substituting \eqref{eq:cond-exp-estats}, the denominators $f(t,k)$ cancel in every case, and the sum over $k$ may then be taken inside the integrals. Two elementary integrals suffice. Since $\sum_k D_k = D\mathbf 1_n = -T\mathbf 1_m$ and $\int_c^\infty e^{Tt}\,dt = e^{Tc}(-T)^{-1}$, valid because $T$ is a transient sub-generator,}
\begin{equation}
\label{eq:censored-two-integrals}
\subchg{\sum_{k=1}^n\int_c^\infty e^{Tt}D_k\,dt = e^{Tc}\mathbf 1_m ,
\qquad
\int_c^\infty \big(\alpha e^{Tt}\big)_i\,dt = g_i(c) .}
\end{equation}

\paragraph{Initial state.} \subchg{Taking $H=B_i$ and substituting $\ex{B_i\mid\tau=t,\kappa=k}f(t,k)=\alpha_i(e^{Tt}D_k)_i$,}
\[
S(c)\,\ex{B_{i} \mid \tau>c}
 =\sum_{k=1}^n\int_c^\infty \alpha_i\big(e^{Tt}D_k\big)_i\,dt
 =\alpha_i\big(e^{Tc}\mathbf 1_m\big)_i
\]
\subchg{by the first integral of \eqref{eq:censored-two-integrals}.}

\paragraph{Absorbing jump.} \subchg{Taking $H=E_{ik}$, only the term with that same cause survives the sum, since $\ex{E_{ik}\mid\tau=t,\kappa=k'}=0$ for $k'\neq k$. With $\ex{E_{ik}\mid\tau=t,\kappa=k}f(t,k)=(\alpha e^{Tt})_iD_{ik}$,}
\[
S(c)\,\ex{E_{ik} \mid \tau>c}
 =\int_c^\infty \big(\alpha e^{Tt}\big)_i D_{ik}\,dt
 =g_i(c)\,D_{ik} .
\]

\paragraph{Occupation time and transient jumps.} \subchg{These two share one computation. Taking $H=Z_i$ and $H=M_{ij}$ with $i \ne j$, and substituting $\ex{Z_i\mid\tau=t,\kappa=k}f(t,k)=C_{ii}(t,k)$ and $\ex{M_{ij}\mid\tau=t,\kappa=k}f(t,k)=T_{ij}C_{ij}(t,k)$, both reduce to the double integral}
\[
\sum_{k=1}^n\int_c^\infty C_{ij}(t,k)\,dt
 =\int_c^\infty\!\!\int_0^t
   \big(\alpha e^{Tu}\big)_i
   \big(e^{T(t-u)}(-T)\mathbf 1_m\big)_j\,du\,dt ,
\]
\subchg{the sum over $k$ having been taken inside and collapsed by $\sum_kD_k=-T\mathbf 1_m$. The integrand is non-negative, so Tonelli's theorem applies, and the domain $\{(u,t):c\le t<\infty,\ 0\le u\le t\}$ is, up to its boundary, the disjoint union of $\{0\le u<c,\ c\le t<\infty\}$ and $\{c\le u<\infty,\ u\le t<\infty\}$. On the first piece the inner integral is $\int_c^\infty e^{T(t-u)}(-T)\mathbf 1_m\,dt = e^{T(c-u)}\mathbf 1_m$, and on the second it is $\int_u^\infty e^{T(t-u)}(-T)\mathbf 1_m\,dt = \mathbf 1_m$, whose $j$-th entry is $1$. Hence}
\[
\sum_{k=1}^n\int_c^\infty C_{ij}(t,k)\,dt
 =\int_0^c \big(\alpha e^{Tu}\big)_i\big(e^{T(c-u)}\mathbf 1_m\big)_j\,du
  +\int_c^\infty \big(\alpha e^{Tu}\big)_i\,du
 =C_{ij}(c;\mathbf 1_m)+g_i(c) .
\]
\subchg{Dividing by $S(c)$ gives $\ex{Z_i\mid\tau>c}$ at $j=i$, and $\ex{M_{ij}\mid\tau>c}$ after multiplication by $T_{ij}$ for $i\neq j$.}

\paragraph{Exit counts.} \subchg{Every completed path leaves $s_i$ through its transient jumps and, if $s_i$ is its last transient phase, through the one absorbing jump, so $N_i=\sum_jM_{ij}+\sum_kE_{ik}$ pathwise by the flow balance \eqref{eq:flow-balance}. Taking conditional expectations proves \eqref{eq:censored-N}.}
\end{proof}

\paragraph{Checks and interpretation.}
\subchg{The identities $\sum_i\ex{B_i\mid\tau>c}=1$ and $\sum_{i,k}\ex{E_{ik} \mid \tau>c}=1$ follow directly from $D\mathbf 1_n=-T\mathbf 1_m$: the censored path starts in exactly one phase and is eventually absorbed exactly once. The two summands in $\ex{M_{ij} \mid \tau>c}$ have a pathwise meaning: $C_{ij}(c;\mathbf 1_m)$ counts jumps before censoring, weighted by survival to $c$, whereas $g_i(c)$ counts jumps after it. Likewise $C_{ii}(c;\mathbf 1_m)$ and $g_i(c)$ are the pre- and post-censoring occupation contributions.} When $n=1$, these formulas reduce to the classical right-censored phase-type E-step of \citet{Olsson_1996_censored}.

\paragraph{Computation.}
The matrices in \eqref{eq:censored-Cb} use the same Van~Loan construction as the exact-event E-step. \subchg{For the required weight $b=\mathbf 1_m$,}
\begin{equation}
\label{eq:censored-vanloan}
\exp\!\left(c\begin{bmatrix}
T^\intercal&\alpha^\intercal b^\intercal\\[2pt]0&T^\intercal
\end{bmatrix}\right)
=\begin{bmatrix}e^{T^\intercal c}&C(c;b)\\[2pt]0&e^{T^\intercal c}\end{bmatrix}.
\end{equation}
\subchg{Only the terminal weight $b=\mathbf 1_m$ is required, so one block exponential per censored record suffices, whatever the number of causes. Numerically, $g(c)$ should be computed by solving a linear system with $-T$ rather than forming an explicit inverse.}